%% file: main.tex
\documentclass[conference]{IEEEtran}

\usepackage{url}
\usepackage{graphicx}
\usepackage{amsthm}
\usepackage{amsmath}
\usepackage{amssymb}
\usepackage[linesnumbered,ruled,vlined]{algorithm2e}
\usepackage{array}
\usepackage{xcolor}
\usepackage{cite}
\usepackage[colorlinks=true, linkcolor=blue, urlcolor=blue, citecolor=blue]{hyperref}
\usepackage{multirow}
\usepackage{listings}
\usepackage{etoolbox}

\newboolean{extendedVersion}
\setboolean{extendedVersion}{true}

\ifbool{extendedVersion}{
}{
\usepackage[available,functional,reproduced]{ndssbadges}
}

\SetKwProg{Proc}{Procedure}{}{}

\newtheorem{manualtheoreminner}{Theorem}
\newenvironment{theorem}[1]{%
  \renewcommand\themanualtheoreminner{#1}%
  \manualtheoreminner
}{\endmanualtheoreminner}
\newtheorem{manuallemmainner}{Lemma}
\newenvironment{lemma}[1]{%
  \renewcommand\themanuallemmainner{#1}%
  \manuallemmainner
}{\endmanuallemmainner}

\newtheorem{definition}{Definition}

\begin{document}

    \title{MVP-ORAM: a Wait-free Concurrent ORAM\\ for Confidential BFT Storage}

\author{
	\IEEEauthorblockN{
         Robin Vassantlal 
         \hspace{8mm}
         Hasan Heydari
         \hspace{8mm}
         Bernardo Ferreira
         \hspace{8mm}
         Alysson Bessani
     }
     \IEEEauthorblockA{
         \textit{LASIGE, Faculdade de Ci\^encias, Universidade de Lisboa, Portugal}
     }
}

\IEEEoverridecommandlockouts
\makeatletter\def\@IEEEpubidpullup{4.5\baselineskip}\makeatother
\IEEEpubid{\parbox{\columnwidth}{
    \ifbool{extendedVersion}{
        This is the extended version of the paper published at the Proc. of the 33rd Network and Distributed System Security Symposium (NDSS 2026)~\cite{mvp_oram_ndss}.
    }{
        Network and Distributed System Security (NDSS) Symposium 2026\\
        23 - 27 February 2026, San Diego, CA, USA\\
        ISBN 979-8-9919276-8-0\\  
        https://dx.doi.org/10.14722/ndss.2026.241809\\
        www.ndss-symposium.org
    }
}
\hspace{\columnsep}\makebox[\columnwidth]{}}

\maketitle

\pagestyle{plain}

\begin{abstract}
It is well known that encryption alone is not enough to protect data privacy. 
Access patterns, revealed when operations are performed, can also be leveraged in inference attacks. 
Oblivious RAM (ORAM) hides access patterns by making client requests oblivious. 
However, existing protocols are still limited in supporting concurrent clients and Byzantine fault tolerance (BFT). 
We present MVP-ORAM, the \emph{first wait-free ORAM protocol} that supports concurrent fail-prone clients.
In contrast to previous works, MVP-ORAM avoids using trusted proxies, which necessitate additional security assumptions, and concurrency control mechanisms based on inter-client communication or distributed locks, which limit overall throughput and the capability to tolerate faulty clients. 
Instead, MVP-ORAM enables clients to perform concurrent requests and merge conflicting updates as they happen, satisfying wait-freedom, i.e., clients make progress \emph{independently of the performance or failures of other clients}.
Since wait and collision freedom are fundamentally contradictory goals that cannot be achieved simultaneously in an asynchronous concurrent ORAM service, we define a weaker notion of obliviousness that depends on the application workload and number of concurrent clients, and prove MVP-ORAM is \emph{secure in practical scenarios where clients perform skewed block accesses}.
By being wait-free, MVP-ORAM can be seamlessly integrated into existing confidential BFT data stores, creating the first BFT ORAM construction.
We implement MVP-ORAM on top of a confidential BFT data store and show \emph{our prototype can process hundreds of 4KB accesses per second} in modern clouds.
\end{abstract}

\IEEEpeerreviewmaketitle

\input{content/sec_introduction}

\input{content/sec_background_related_work}

\input{content/sec_system_adversary_model}

\input{content/sec_basic_oram}

\input{content/sec_mvp_oram}

\input{content/sec_bft_oram}

\input{content/sec_security_analysis}

\input{content/sec_strong_mvp_oram}

\input{content/sec_evaluation}

\input{content/sec_conclusions}

\section*{Acknowledgment}
We thank the anonymous reviewers for their constructive comments, which helped improve the paper. 
We also thank Cristiano Santos for his initial work on BFT ORAM services, which sparked the results presented in this paper. 
This work was supported by FCT through the Ph.D. scholarship, ref.  \href{https://doi.org/10.54499/2020.04412.BD}{2020.04412.BD}, the SMaRtChain and APOSTLE projects, ref. \href{https://doi.org/10.54499/2022.08431.PTDC}{2022.08431.PTDC} and \href{https://doi.org/10.54499/2023.12254.PEX}{2023.12254.PEX}, respectively, and the LASIGE Research Unit, ref. \href{https://doi.org/10.54499/UID/00408/2025}{UID/00408/2025}.


\bibliographystyle{IEEEtran}
\bibliography{references}

\appendices

\input{content/ap_mpv_oram_auxiliary_functions}

\ifbool{extendedVersion}{
    \input{content/ap_correctness}
    
    \input{content/ap_obliviousness}
    
    \input{content/ap_stash}

    \input{content/ap_strong_mvp_oram_algorithms}

}{}

\ifbool{extendedVersion}{
\input{content/artifact_extended}
}{
\input{content/artifact}
}

\end{document}

%% file: content/sec_introduction.tex
\section{Introduction}

\textbf{Context and motivation.}
Byzantine Fault-Tolerant State Machine Replication (BFT SMR) is a classical technique to implement fault- and intrusion-tolerant replicated services with strong consistency~\cite{schneider1990implementing,castro1999practical}.
The technique attracted significant attention in the last decade due to the emergence of decentralized systems and blockchains~\cite{nakamoto_2009,wood_2014}, which can be seen as replicated state machines.
BFT SMR systems offer \emph{data integrity} and \emph{availability} guarantees, even if up to $t$ of the $n$ replicas are compromised.
Some works have additionally studied how to make BFT SMR systems offer \emph{confidentiality} ~\cite{depspace,padilha_2011,basu_2019,calypso,vassantlal_2022}, guaranteeing \emph{data secrecy} even in the presence of Byzantine faults.
This is typically achieved by combining symmetric encryption with secret sharing~\cite{shamir1979share,blakely1979safeguarding}, a technique where a secret (e.g., an encryption key) is split into $n$ shares (one for each server) and any subset of $t+1$ of them is required for recovering it.

However, encryption, even if augmented with secret sharing, is not enough to ensure data secrecy.
Data access patterns, revealed when clients perform operations, can also be leveraged in inference attacks, sometimes with disastrous consequences~\cite{islam2014inference,naveed2015inference,zhang2016all,grubbs2017your}.
Access patterns typically leaked in a storage service include which entries are accessed, when, how often, if they are accessed with other entries, and whether they are being read or written.

\textbf{State of the art.}
Oblivious RAM (ORAM)~\cite{goldreich1987towards,ostrovsky1990efficient,goldreich1996software} is a cryptographic technique whose objective is to conceal access patterns and make them \emph{oblivious}, i.e., indistinguishable from each other.
However, it was initially designed for a single CPU accessing its RAM (or a client accessing its server).
Recent works have studied its suitability for supporting multiple CPUs/clients in different yet related lines of work, known as parallel ORAM~\cite{boyle2015,chen2016oblivious,hubert2017circuit,chan2018perfectly,chan2020perfectly,asharov2022optimal} and multi-client ORAM~\cite{goodrich2012privacy,sahin2016taostore,blass2017multi,maffei2017maliciously,concuroram}. 
However, these works typically require inter-client communication for synchronization~\cite{asharov2022optimal}, distributed locks~\cite{concuroram}, or trusted proxies/hardware~\cite{sahin2016taostore}, all of which severely limit concurrency or require additional security assumptions.
Concurrently, researchers have also studied how multiple servers can be leveraged in ORAM~\cite{stefanov2013multi,s3oram,larsen2020lower,hoang2020multi,macao} to reduce client-server bandwidth.
To the best of our knowledge, QuORAM~\cite{maiyya2022quoram} is the only multi-server protocol that increases ORAM availability but only addresses server crashes.
Hence, to date, no ORAM protocol can conceal access patterns in Byzantine-resilient systems, and existing protocols cannot be extended to meet the requirements of BFT SMR without sacrificing client fault tolerance.

\textbf{Problem statement.}
We address the problem of designing \emph{practical BFT-replicated storage services that can hide data access patterns}, going beyond existing works on confidential BFT systems (e.g.,~\cite{depspace,padilha_2011,basu_2019,calypso,vassantlal_2022}).
We aim to design a replicated ORAM service in which replicas can be subject to Byzantine failures while concurrently accessed by fail-prone clients.
In more detail, we are interested in a setting where concurrent clients access a replicated storage service in which (1) servers are subject to Byzantine faults, (2) multiple clients can concurrently access shared data stored in servers without external coordination or waiting for each other, and (3) data access patterns remain oblivious.
In this setting, a particularly important property to achieve is \emph{wait-freedom}~\cite{herlihy_1991}, which states that every client operation (ORAM access) is guaranteed to finish in a finite number of steps.
We stress the practical relevance of this property, as wait-free services are more robust than lock-based services~\cite{ZooKeeper} (the norm in the concurrent ORAM literature), as clients can finish their accesses independently of the delays and faults of other clients.

\textbf{Our solution.}
This paper presents Multi-Version Path ORAM (MVP-ORAM), the first ORAM protocol designed explicitly for concealing access patterns in BFT SMR systems while achieving wait-freedom.
MVP-ORAM is based on Path ORAM~\cite{stefanov2018path}, a simple ORAM protocol where a single client accesses a single server to store encrypted data blocks organized as a binary tree.
We selected Path ORAM as a starting point due to its low number of client-server round-trips per access, compared to more recent solutions (e.g., \cite{ren2015constants,asharov2023futorama}), which is crucial in BFT SMR systems.
In Path ORAM, the client accesses data by reading and writing a whole path of the tree where the block of interest is located.
To access the correct path, the client maintains a table mapping block addresses to paths in the tree where the blocks are stored.
The client also maintains an expected small stash of blocks that have temporarily overflowed from the tree.
Path ORAM matches the required bandwidth lower bound of $O(\log N)$ for obliviously accessing a data store with $N$ blocks~\cite{larsen2018yes}.

MVP-ORAM improves Path ORAM in two fundamental ways.
First and foremost, it supports multiple clients concurrently accessing data while satisfying strong consistency~\cite{herlihy_1990}.
Contrary to previous works on parallel ORAM, which use locks and other inter-client coordination mechanisms to ensure consistency and security, we aim to support client-independent, \textit{wait-free} ORAM accesses, as required in SMR-based services~\cite{castro1999practical}.
To the best of our knowledge, no ORAM satisfies this property.


To support wait-freedom, MVP-ORAM encrypts and stores the position map and stash in the server, along with the tree.
Clients read the position map to define the block's access path and then request the path and stash from the server.
With the ORAM data structures stored on the server, MVP-ORAM supports \emph{multiple clients concurrently accessing data by managing multiple versions of each data structure}.
More specifically, servers start with a single version of each ORAM data structure, but concurrent clients can read and generate new versions of these data structures, which are later merged by clients during their access, as illustrated in Fig.~\ref{fig:outstanding_versions}.
Notice that, although our goal is to support BFT ORAM, MVP-ORAM's wait-free design is of independent interest, as no existing single-server scheme has its unique set of features.

\begin{figure}[!t]
    \centering
    \includegraphics[width=\columnwidth]{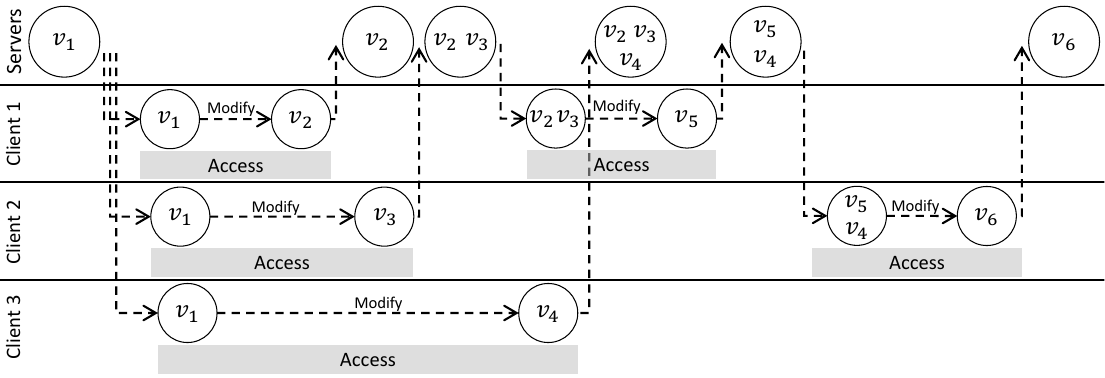}
    \caption{Multiple versions being created by concurrent accesses and later merged in MVP-ORAM.}
    \label{fig:outstanding_versions}
\end{figure}

Second, MVP-ORAM relaxes the trust assumption on the servers by tolerating Byzantine failures.
More specifically, besides allowing servers to observe their internal state, as in the semi-honest model used in Path ORAM and most ORAM schemes, MVP-ORAM tolerates $t$ malicious servers by employing a BFT SMR protocol to replicate a deterministic ORAM service in $n > 3t$ servers.

To make our oblivious BFT data store feature-complete, we tackle the problem of managing the shared keys used to encrypt server data. 
Contrary to previous works on the multi-client setting, which assume that encryption keys are shared between clients in some way, MVP-ORAM integrates a secret-sharing framework to distribute encryption keys through the servers alongside the ORAM state~\cite{basu_2019,vassantlal_2022}.
Fig.~\ref{fig:oram_bft_stack} illustrates our construction, where MVP-ORAM runs on top of secret sharing and BFT SMR.

\begin{figure}[t!]
    \centering
    \includegraphics[width=\columnwidth]{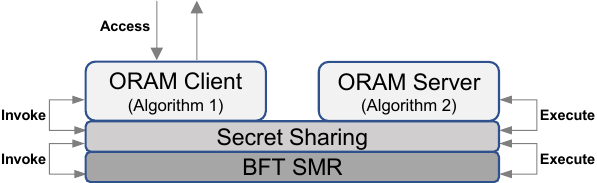}
    \caption{MVP-ORAM protocol stack.}
    \label{fig:oram_bft_stack}
\end{figure}

\textbf{Security and performance.}
From a security perspective, achieving wait-freedom in ORAM introduces new challenges, as without client synchronization, it becomes impossible to ensure that no two clients access the same block at the same time, a property known in ORAM as \textit{collision-freedom}~\cite{boyle2015}. 
Indeed, we argue that in asynchronous networks it is fundamentally impossible to conciliate wait- and collision-freedom.
To circumvent this limitation and increase security, when a client wants to access a block, MVP-ORAM chooses and requests at random any of the paths that contain it.
Additionally, to increase the number of paths available to access a block, evictions move the most popular blocks closer to the tree root.
Assuming ORAM accesses follow a Zipfian distribution, meaning that a very high percentage of accesses are done to a very small percentage of blocks (as has been shown to happen in many natural and digital systems~\cite{newman2005power,dan1992characterization,lynch1988selectivity}), this solution allows MVP-ORAM to preserve wait-freedom while its security approximates that of collision-free ORAMs.
Nonetheless, collisions can still happen (e.g., two clients simultaneously access a leaf block), even if with low probability.
We address this limitation by proposing a variant of MVP-ORAM, which uses dummy requests to ensure obliviouness at the cost of performance and synchrony assumptions.

In terms of performance, MVP-ORAM requires an amount of bandwidth linearly proportional to the number of servers and quadratic in the number of \emph{active concurrent clients} accessing the system at a time, being thus an \textit{adaptive wait-free construction}~\cite{afek_95}.

We implemented MVP-ORAM on top of COBRA~\cite{vassantlal_2022}, an open-source confidential BFT SMR library, and evaluated its performance.
Our results show that our prototype can process more than $350$ (resp. $700$) accesses per second with a latency of less than $140$ ms (resp. $70$ ms) in a system with $10$ servers (resp. a single server) and $50$ concurrent clients. 
This shows MVP-ORAM can achieve performance numbers in line with other practical (concurrent/fault-tolerant) protocols~\cite{concuroram, maiyya2022quoram}.
Our implementation is open source~\cite{mvp_oram}, and our results are fully reproducible, as described in Appendix~\ref{ap:artifact}.


\textbf{Contributions.}
We claim the following contributions:

\begin{enumerate}
    \item We initiate the study of the problem of implementing a wait-free ORAM (\S\ref{sec:preliminaries});
    \item We design MVP-ORAM, the first asynchronous wait-free ORAM supporting concurrent fail-prone clients (\S\ref{sec:mvp-oram}).
    Besides detailing the basic protocol that satisfies a weaker version of obliviousness, we present a stronger version of MVP-ORAM that provides standard parallel ORAM security (\S\ref{sec:strong_mvp_oram});
    \item We use MVP-ORAM to hide access patterns of a confidential BFT SMR-based storage service, providing the first ORAM that tolerates Byzantine-faulty servers (\S\ref{sec:bft_oram});
    \item We present a detailed analysis of MVP-ORAM's stash size and bandwidth requirements, along with security proofs that demonstrate both its correctness and obliviousness (\S\ref{sec:analysis});
    \item We implement and evaluate MVP-ORAM to show its performance in practical settings (\S\ref{sec:evaluation}). 
\end{enumerate}

%% file: content/sec_background_related_work.tex
\section{Background and Related Work}
\label{sec:related_work}


\textbf{Confidential BFT.} 
The seminal work on intrusion tolerance by Fraga and Powell~\cite{Fra85} was the first to consider information scattering for protecting data confidentiality in a replicated synchronous system.
Later works like Secure Store~\cite{Lak03} and CODEX~\cite{marsh_2004} ensured confidentiality, integrity, and availability of stored data in asynchronous systems by using Byzantine quorum protocols~\cite{Mal98} together with secret sharing~\cite{shamir1979share}.
To the best of our knowledge, DepSpace~\cite{depspace} was the first work to use secret sharing for achieving confidentiality in a BFT SMR system.
Although DepSpace and follow-up works such as Belisarius~\cite{padilha_2011} achieved performance similar to non-confidential BFT SMR, neither of them supported features required in practice, such as replica state recovery, replica group reconfiguration, or protection against a mobile adversary.
The same can be said about works adding confidentiality based on secret sharing to blockchains (e.g., \cite{calypso}).
Basu et al.~\cite{basu_2019} partially solved this by introducing a confidential state recovery protocol for static BFT SMR.
COBRA~\cite{vassantlal_2022} proposed the first confidential BFT SMR system with all the practical features required by these.
Nevertheless, none of these works tackles the problem of hiding access patterns, as an adversary can still observe which data entries are being accessed.
This feature can, for example, be used to improve the privacy of a service like Arke~\cite{arke}, which provides confidential contact discovery using a BFT storage service.

\textbf{Classical ORAM.}
Oblivious RAM was first introduced by Goldreich and Ostrovsky~\cite{goldreich1987towards,ostrovsky1990efficient,goldreich1996software}, in the context of software protection.
Subsequent works improved its efficiency in different scenarios~\cite{ostrovsky1997private,williams2008building,pinkas2010oblivious,damgaard2011perfectly}.
In the 2010s, with the rise of cloud computing, renewed interest in ORAM led to new improvements, including Path ORAM~\cite{stefanov2018path}, which was the first ORAM protocol capable of achieving logarithmic bandwidth overhead (shown to be optimal in the storage-only server setting~\cite{goldreich1996software,larsen2018yes}).
Subsequent works reduced this overhead to $O(1)$ by assuming server computations~\cite{ren2015constants,devadas2016onion}.
However, these typically require homomorphic encryption to be secure, thus reducing bandwidth overhead but decreasing overall performance.
Other works improved performance by using trusted execution environments~\cite{sasy2017zerotrace,ahmad2018obliviate}, but this requires shifting trust from well-established cryptographic assumptions to closed-source solutions from hardware manufacturers.
More recently, researchers have revisited the original hierarchical ORAM of Goldreich and Ostrovsky to make it more practical~\cite{patel2018panorama,asharov2022optorama,asharov2023futorama}.
In all these works, a single client accesses a single server.

\textbf{Multi-client and Parallel ORAM.} 
Recent works have explored how to support multiple concurrent clients in ORAM in different yet related research lines known as multi-client~\cite{goodrich2012privacy,sahin2016taostore,blass2017multi,maffei2017maliciously,concuroram,crooks2018obladi,maiyya2022quoram,cheng2023tianji} and parallel ORAM~\cite{boyle2015,chen2016oblivious,hubert2017circuit,chan2018perfectly,chan2020perfectly,asharov2022optimal}. 
Multi-client ORAM focuses on the client-server model, while parallel ORAM focuses on multi-core CPUs, but both try to solve the same problem: how to support concurrency in ORAM accesses.
Here, the challenge is twofold: first, obliviousness should be ensured not only for access sequences from individual clients but also between clients. 
This means that if all clients decide to access the same data block at the same time, the resulting ORAM accesses should still look random and independent to the server.
The second is how to efficiently deal with concurrency and synchronize local ORAM client data without making the resulting system inherently sequential.
Many works solved these issues by introducing a trusted, confidential proxy (either in the network or in the server) between the clients and the server~\cite{sahin2016taostore,crooks2018obladi,maiyya2022quoram,cheng2023tianji}.
The ORAM protocol is then executed between the trusted component and the server, essentially serializing requests and synchronizing client data.
Other works avoided the trusted component by relying on inter-client communication~\cite{boyle2015,chen2016oblivious,hubert2017circuit,chan2018perfectly,chan2020perfectly} or lock-based distributed algorithms~\cite{maffei2017maliciously,concuroram} to serialize conflicting concurrent requests.
However, these approaches limit client concurrency and prevent wait-freedom from being achieved.
Moreover, most of these works only support a single server, making them vulnerable to server faults.

\textbf{Multi-server ORAM.}
Another related research vector is the use of multiple ORAM servers to reduce client bandwidth requirements~\cite{stefanov2013multi,s3oram,larsen2020lower,hoang2020multi,macao,cheng2023tianji}.
Each server plays a critical role in these works, so fault tolerance is not supported.
As far as we know, QuORAM~\cite{maiyya2022quoram} is the only ORAM protocol that uses multiple servers to tolerate faults.
However, it only tolerates benign (crash) faults and requires trusted proxies attached to servers, an additional strong security assumption.
Indeed, in QuORAM, each server plus proxy constitutes an isolated ORAM instance, and a variant of the classical ABD protocol~\cite{ABD} is used to replicate read/write operations on those proxies, which act as (single) clients to their ORAM servers.

\textbf{The research gap: BFT ORAM.}
As evidenced by our previous discussion, no ORAM protocol can be integrated ``as is'' into a confidential BFT system to hide access patterns.
The state-of-the-art in ORAM fault tolerance is QuORAM~\cite{maiyya2022quoram}, but it only tolerates crashes and requires trusted execution support on servers, i.e., each server needs a trusted proxy.
Extending it to tolerate Byzantine faults seems doable, but it would still require trusted proxies.
Regarding multi-client support without using proxies, ConcurORAM~\cite{concuroram} is the state-of-the-art.
However, it has three main limitations.
First, it heavily relies on multi-threading and locks at the server, which introduces nondeterminism, significantly complicating replication (e.g.,~\cite{locks}).
Second, wait-freedom is impossible to achieve using locks, seriously compromising client fault tolerance.
Third, it requires 18 client-server interactions to complete one ORAM access (query and eviction), 
which, if replicated, would require Byzantine consensus for totally ordering each request.
This large number of client-server iterations is also a limitation of a recent optimal parallel ORAM by Asharov et al.~\cite{asharov2022optimal}.
In contrast, by extending Path ORAM (which only requires two round-trips per access) to keep multiple versions of the ORAM state, MVP-ORAM supports concurrent clients without using locks and requiring just three round-trips per ORAM access (see Fig.~\ref{fig:basic_oram_protocol}), making it thus more appropriate to be integrated into confidential BFT SMR systems.

Notice that the need for a multi-client ORAM free of locks or inter-client coordination to replicate using BFT SMR made us address another research gap of independent interest: the lack of wait-free multi-client ORAMs.


%% file: content/sec_system_adversary_model.tex
\section{Model and Definitions}
\label{sec:preliminaries}

\textbf{System model.}
We consider a fully connected distributed system in which processes are divided into two sets: a set of $n$ servers/replicas $\Sigma = \{r_1,r_2,\dots, r_n\}$, and an unbounded set of clients $\Gamma = \{c_1,c_2,\dots\}$.
We assume a trusted setup in which each replica and client has a unique identifier that can be verified by every other process through standard means, e.g., a public key infrastructure.
We also assume the system has sufficient synchrony to implement BFT SMR and consensus.
For instance, our prototype requires a \emph{partially synchronous model}~\cite{dwork_1988} in which the system is asynchronous until some \emph{unknown} global stabilization time, after which it becomes synchronous, with known time bounds for computation and communication.\footnote{
MVP-ORAM construction is oblivious to the used BFT SMR implementation. 
Nothing precludes MVP-ORAM from being implemented on top of asynchronous protocols (e.g.,~\cite{BEAT,DumboNG}).}
Finally, every pair of processes communicates through \emph{private and authenticated fair links}, i.e., messages can be delayed but not forever.

\textbf{Service model.}
Clients access the replicated storage service, which contains $N$ data blocks, by sending requests and receiving replies to/from the service replicas.
Servers globally store a two-part state $\Omega=\langle C, P \rangle$.
The common state $C$ comprises ORAM data, and the private state $P$ comprises encryption keys used to encrypt $C$.
Each server~$r_i$ locally maintains a state $\Omega_i=\langle C, P_i\rangle$.
The common state $C$ is encrypted and replicated across all servers, i.e., all servers store the same state, while the private state is distributed using the secret sharing protocol.
Hence, $r_i$'s private state $P_i$ comprises shares of the encryption keys.
The functionality of our ORAM service offers a single operation, described in the following way:

\begin{itemize}
    \item $\langle \mathit{data} \rangle = \mathsf{access}(c_i, \mathit{op}, \mathit{addr}, \mathit{data}^*)$: client $c_i$ invokes $\mathsf{access}$ to read or write block addressed by $\mathit{addr}$, i.e., 
    $c_i$ invokes $\mathsf{access}(c_i, \mathit{read}, \mathit{addr}, \bot)$ to read the block and  $\mathsf{access}(c_i, \mathit{write}, $ $\mathit{addr}, \mathit{data})$ to write $\mathit{data}$ to the block.
\end{itemize}

Finally, we assume applications using our storage service generate a \emph{skewed block access pattern}.
More specifically, we assume that storage clients \emph{collectively} induce an access pattern in which a small fraction of the stored blocks are accessed much more frequently than the others.
This skewed pattern, typically modeled by a Zipfian distribution~\cite{kingsley1932selected}, is commonly observed in datasets~\cite{newman2005power} and in storage systems accesses (e.g.,~\cite{dan1992characterization,vanrenen2024,leung2008}), being modelled in popular storage systems benchmarks~\cite{leutenegger1993,ycsb}.
This skewness is exploited in most real systems through caching and load-balancing techniques.
In this paper, we use it to characterize the obliviousness of wait-free ORAM.

\textbf{Adversary model.}
We consider an adversary that can fully control a fraction of the replicas and the scheduling of messages, but has limited access to clients. 
In particular, we assume that the adversary can maliciously corrupt some of the replicas and crash clients, but can not inject concurrent queries, as that would allow it to force collisions between client accesses. 
This assumption is somewhat similar to the models of Pancake~\cite{pancake} and Waffle~\cite{waffle}, which consider a \textit{passive persistent adversary} that can observe all accesses but cannot inject its own queries.
We believe this model accurately captures the typical security guarantees of BFT data stores, where a set of semi-trusted clients store shared data using untrusted servers. 
Nonetheless, in the Strong MVP-ORAM variant (§\ref{sec:strong_mvp_oram}), we remove this assumption and consider that the adversary additionally can inject concurrent queries and force collisions.

More formally, we consider a probabilistic polynomial-time adaptive adversary that can control the network and may at any time decide to corrupt a fraction $t < n/3$ of the replicas or crash clients.
Corrupted replicas can deviate arbitrarily from the protocol, i.e., they are prone to Byzantine failures.
Such replicas are said to be faulty or corrupted.
A process that is not faulty is said to be correct or honest.
The adversary can learn about the private state that corrupted replicas store and the access patterns of operations received.
Clients are assumed to be honest, so they can only fail by crashing and cannot be influenced by the adversary in any other way.

As in other oblivious datastores and confidential BFT services~\cite{pancake,waffle,vassantlal_2022,basu_2019}, we do not consider fully malicious clients, as there is little point in protecting the confidentiality of a service if malicious clients have permission to access the data.
In practice, our service supports multiple ORAMs, each of which is shared by a set of mutually trusted clients.
Nonetheless, this restriction can be alleviated through mechanisms for verifiable computation, such as ZK-Proofs~\cite{maffei2017maliciously,backes2016anonymous} or MPC-based proxies~\cite{chen2020metal,chen2022titanium}.
We leave the integration of these techniques with MVP-ORAM for future work.

\textbf{Security definition.}
Beyond ensuring the Safety, Liveness, and Secrecy properties that are standard in confidential BFT services~\cite{depspace,basu_2019,vassantlal_2022}, MVP-ORAM additionally aims at ensuring \emph{Obliviousness} (i.e., Access Pattern Secrecy)~\cite{goldreich1987towards}. 

Safety (i.e., Linearizability), requires the replicated service to emulate a centralized service~\cite{herlihy_1990}; Liveness (i.e., Wait-Freedom) requires all correct client requests to be executed~\cite{herlihy_1991}; and Secrecy (i.e., Confidentiality) requires that no private information about the stored data be leaked as long as the failure threshold of the system is respected~\cite{vassantlal_2022}.



As for Obliviousness, we start with the definition from parallel ORAM~\cite{boyle2015,chen2016oblivious,chan2017}: given any two sequences of parallel operations $\overrightarrow{y_1}$ and $\overrightarrow{y_2}$ of equal length, they should look indistinguishable to the adversary, except with negligible probability in~$N$.
This definition requires the ORAM to be \emph{collision-free}~\cite{boyle2015}, i.e., no two clients ever access the same address concurrently.
However, we argue that in asynchronous networks, no ORAM protocol can simultaneously be collision- and wait-free, as the former is impossible to achieve without client synchronization (e.g., distributed locks~\cite{concuroram}, inter-client communication~\cite{boyle2015}), which in turn prevents the latter (since a single client failure can prevent others from progressing).

Hence, we propose a new obliviousness definition for asynchronous wait-free ORAM: the indistinguishability between $\overrightarrow{y_1}$ and $\overrightarrow{y_2}$ is characterized by the statistical distance of their access patterns, which depends not only on the ORAM size $N$, but also the number of concurrent clients $c$ and the distribution of concurrent accesses $\mathcal{D}$, sampled from the universe of all accesses $\mathcal{U}$, from which both $\overrightarrow{y_1},\overrightarrow{y_2}$ are themselves sampled.

Since the adversary can control the number of concurrent clients accessing the service through network scheduling, we assume the worst-case scenario in which all $c$ clients are accessing the ORAM simultaneously.
This increases the likelihood that multiple clients will request the same block concurrently.
We define a \emph{timestep} as the interval from the start of the first concurrent access to the end of the last concurrent access among the group of $c$ clients.
With this notion, we now provide the security definition for wait-free ORAM:

\begin{definition} [Asynchronous Wait-Free ORAM]
    \label{def:opram}
    Given $c, N \in \mathbb{N}$ and $\mathcal{D} \in \mathcal{U}$,
    let $\overrightarrow{b_e} = \{b_i\}_{i \in \{1, \dots, c\}}$ denote a set of $c$ concurrent operations in timestep $e$ and $\overrightarrow{y}=(\overrightarrow{b_1}, \overrightarrow{b_2}, \dots)$ denote a sequence of such concurrent operations in each timestep.
    Protocol $\Pi$ is an Asynchronous Wait-Free Oblivious Parallel RAM (or simply Asynchronous Wait-Free ORAM) if there exists a function $\mu$ such that:
    \begin{itemize}
        \item \textbf{Correctness:} Given $\overrightarrow{y} \stackrel{\$}{\gets} \mathcal{D}$, the execution of $\Pi$ returns the last written version of each block requested in $\overrightarrow{y}$ (i.e., for each block, the version with the highest timestamp).

        \item \textbf{Obliviousness:} Let $A(\overrightarrow{y})$ denote the access pattern generated by $\Pi$ when $\overrightarrow{y}$ is executed.
        We say $\Pi$ is secure if for any two sequences of concurrent operations $\overrightarrow{y_1}$,\, $\overrightarrow{y_2} \stackrel{\$}{\gets} \mathcal{D}$ of the same length, with inputs chosen by clients, the statistical distance between $A(\overrightarrow{y_1})$ and $A(\overrightarrow{y_2})$ is bounded by $\mu(N,c,\mathcal{D})$.
    \end{itemize}
\end{definition}

This definition is weaker than the one used in traditional parallel ORAM (e.g., \cite{boyle2015,chen2016oblivious,chan2017}), as it does not allow the adversary to inject queries and it depends on additional security parameters that may make $\mu$ non-negligible, namely $c$ and $\mathcal{D}$. 
In MVP-ORAM, $\mu$ will be negligible if, per timestep, a single client accesses the ORAM or multiple clients access different blocks. However, if multiple clients access the same block concurrently, $\mu$ may not be necessarily negligible, although it can be arbitrarily small (see \S\ref{sec:security_analysis}).

In applications where this may be a problem (e.g., if $\mathcal{D}$ is expected to be uniform instead of Zipfian) and if network synchrony can be assumed, our Strong MVP-ORAM (\S\ref{sec:strong_mvp_oram}) can be used instead, sacrificing performance and asynchrony but fulfilling wait-freedom and parallel ORAM security.

%% file: content/sec_basic_oram.tex
\section{A First Multi-Client ORAM Protocol}
\label{sec:basic_oram}

We start by presenting a first attempt to design a multi-client ORAM protocol based on Path ORAM~\cite{stefanov2018path} that does not require distributed locks, inter-client communication, or trusted proxies to serialize client requests.
This first protocol achieves a liveness property known as \emph{obstruction-freedom}~\cite{herlihy_2003}, meaning a client can finish an invoked ORAM access only if all other clients stop making new requests.

\subsection{Path ORAM}

Path ORAM is a simple protocol in which a client invokes an access operation to read or write data blocks from/to an ORAM server. 
The server keeps $N$ fixed-size \emph{blocks}, each associated with a logical \emph{address}, in a \emph{binary tree} of height $L$ and $2^L$ leaves.
Each tree node is called a \emph{bucket} and contains $Z$ blocks.
The client locally maintains a \emph{position map} associating each block to a path in the tree.
Let $l \in \{0, \dots, 2^L-1\}$ be a leaf node of the binary tree. 
A path $\mathcal{P}_l=\{\mathcal{B}_0, \dots, \mathcal{B}_L\}$, contains all buckets from the root to node $l$.
The client also maintains a \emph{stash} with blocks waiting to be written to the tree because their paths are full.

To access a block, the client starts by discovering its path in the position map and requests all buckets of the path from the server, adding them to the stash.
It then reads/modifies the block, changes its path at random, refreshes the encryption of all fetched blocks, and attempts to evict all blocks in the stash back to the server.
Evictions follow a read-path eviction strategy, meaning that blocks can only be written back in the intersection between their assigned paths and the read path.

This simple scheme guarantees obliviousness by reading a whole path per access, instead of a single block, and by randomly changing the path of blocks each time they are accessed. It 
requires a bandwidth of $O(\log N)$ bits, matching the lower bound for storage-only ORAM protocols~\cite{larsen2018yes}, and only requires two round-trips per access, the lowest amongst practical ORAMs~\cite{ren2015constants,asharov2022optorama,asharov2023futorama}. 
This is an important metric in BFT SMR, and the main reason for selecting Path ORAM as a starting point, as server requests must be totally ordered via consensus before being processed by the servers, thereby making each access very costly.

\subsection{Extending Path ORAM to Multiple Clients}

Two issues must be addressed to extend Path ORAM to support multiple clients.
First, Path ORAM requires the client to keep the position map and stash, so multiple clients must have access to shared, up-to-date versions of these data structures.
Second, concurrency must be managed carefully, not only to avoid concurrent accesses leaking information, but also to prevent tree inconsistencies.

At a high level, our first multi-client Path ORAM addresses these challenges by moving client storage (encrypted) to the server and having clients fetch and update it during their access with the help of the server to manage concurrency.

Fig.~\ref{fig:basic_oram_protocol} illustrates our first multi-client Path ORAM.
The protocol requires servers to implement the following functionality:

\begin{figure}[!t]
    \centering
    \includegraphics[width=0.9\columnwidth]{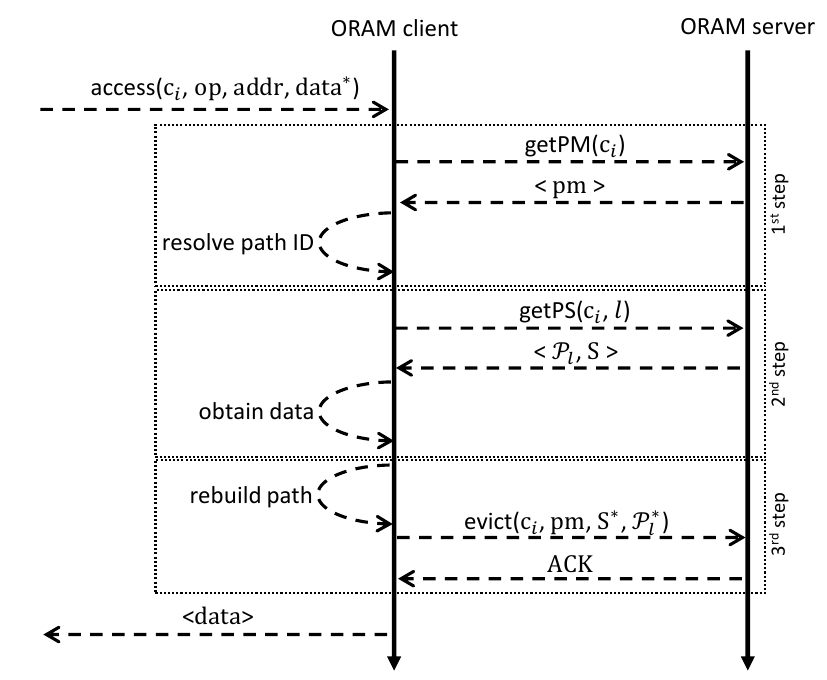}
    \caption{Simple Multi-Client Path ORAM protocol.}
    \label{fig:basic_oram_protocol}
\end{figure}
\begin{itemize}
    \item $\mathit{pm} \gets \mathsf{getPM}(c_i)$: if there is no \emph{active client}, sets client $c_i$ as active (i.e., started an access operation) and retrieves the current position map $\mathit{pm}$; else, it returns $\bot$.
    
    \item $\langle \mathcal{P}_l, S \rangle \gets \mathsf{getPS}(c_i, l)$: if $c_i$ is still the \emph{active client}, the server returns path $\mathcal{P}_l$ and stash $S$; otherwise, the server returns $\bot$.
    
    \item $\mathsf{evict}(c_i, \mathit{pm}, S^*, \mathcal{P}_l^*)$: evicts updated position map $\mathit{pm}$, new stash $S^*$, and new path $\mathcal{P}_l^*$ to the server.
    If $c_i$ is still the active client, the server stores the received data and sets the active client to $\bot$.
\end{itemize}

Let $b$ be a block with address $\mathit{addr}$ that client $c_i$ wants to access, and $\mathit{pm}$ and $S$ be the position map and stash, respectively.
All data received/sent from/to the server must be encrypted/decrypted, but we omit these operations for simplicity.
The $\mathsf{access}$ operation (\S\ref{sec:preliminaries}) has three steps:

\begin{enumerate}

\item To access $b$, $c_i$ must first discover its path.
This is done by invoking $\mathsf{getPM}$ from the server to retrieve $\mathit{pm}$ and accessing $\mathit{pm}[\mathit{addr}]$ to obtain path id $l$.
If the server returns $\bot$, $c_i$ retries after a random back-off time.

\item Client $c_i$ invokes $\mathsf{getPS}$ to obtain path $\mathcal{P}_l$ and stash $S$ from the server. Since $c_i$ is still the active client, the server returns the requested data.
Client $c_i$ adds all blocks from $\mathcal{P}_l$ and $S$ to a working set $W$, reads/writes block $b$, and assigns a new random path to $b$ in $\mathit{pm}$.

\item The access ends with $c_i$ evicting the blocks from $W$. 
The client first populates a new path $\mathcal{P}_l^*$ with the blocks from $W$.
The path $\mathcal{P}_l^*$ is filled from leaf to root with blocks in the intersection between $l$ and their paths.
Overflowing blocks are stored in a new stash $S^*$.
Then $c_i$ sends the new path $\mathcal{P}_l^*$ to the server, along with the updated $\mathit{pm}$ and $S^*$, by invoking $\mathsf{evict}$.
Upon receiving this request, the server checks if $c_i$ is still the active client and, if so, replaces its path $\mathcal{P}_l$ in the tree by $\mathcal{P}_l^*$, its position map by $\mathit{pm}$, and its stash by $S^*$.

\end{enumerate}

Although each individual step is \emph{atomic} at the server, the access is not, as it requires three steps, and different clients can interleave these steps, interrupting accesses one from another.
As a result, this first protocol only ensures a client completes its access if no other client accesses the ORAM concurrently, satisfying obstruction-freedom~\cite{herlihy_2003}.
Another consequence of this design is that a failure of a client during an access, which will never end, might block other clients forever.

%% file: content/sec_mvp_oram.tex
\section{Multi-Version Path ORAM}
\label{sec:mvp-oram}



Extending the previous protocol to support concurrent wait-free accesses requires addressing two fundamental problems.
The first is how to avoid breaking obliviousness on concurrent accesses to the same address.
Indeed, when clients access the same address in the same timestep, they will request the same path. 
This breaks \emph{collision-freedom}~\cite{boyle2015}, as it allows the server to infer that clients may be accessing the same address, even if it can not pinpoint exactly which one is being accessed.
%
The second problem is how to preserve data consistency between concurrent evictions.
Since access operations are not atomic nor serialized through client synchronization, multiple (possibly conflicting) versions of the tree will be generated.


To tackle these problems, we propose \emph{Multi-Version Path ORAM} (MVP-ORAM). 
In further detail, to tackle the first problem, we make the server store the exact slot of the bucket where the block is located in the position map, instead of its path.
This key idea allows clients to retrieve a block through any path that passes through the slot where the block is stored.
Specifically, when clients want to access a block $b$, they first discover its location $\mathit{sl}$ using the position map.
Then, they extend the location to one of the paths that pass through $\mathit{sl}$, and use that path to retrieve the block. 
Since clients select these paths randomly, multiple clients accessing the same block will request different paths with increasing probability as the block is higher in the tree.

To further increase the number of available paths, we keep the last accessed block in the stash, meaning it can be accessed again using any path, and evict the most frequently accessed blocks to higher levels of the tree.
Specifically, when performing an access, the accessed block always goes to the stash (if it is not already there) along with the non-dummy blocks from $Z$ slots uniformly selected at random.
Then, these $Z$ slots are filled with $Z$ random blocks previously in the stash, i.e., we swap $Z$ blocks between the stash and the accessed path.
The constant $Z$ is important to bound the stash size.
Additionally, after the swap, we reorder blocks in the path by placing the most frequently accessed blocks higher in the tree.
The result is that after each access, the most popular blocks in a skewed access pattern will be accessible through many paths, improving the ORAM obliviousness.

To address the second problem, we enable clients to complete their access in isolation and store updates as new versions of the tree on the server.
During an access, each client fetches the existing versions currently stored in the server and merges them into a single, updated tree (as illustrated in Fig.~\ref{fig:outstanding_versions}).
Note that in practice, clients only need to merge paths and stashes that they will retrieve in an access, rather than entire trees.

When multiple versions accessed by clients are merged together, such a merge needs to be done (1) without losing block updates,\footnote{
Linerizability~\cite{herlihy_1990}, or register atomicity, requires a read executed after a write to always return the last update on the stored data.} 
(2) by keeping more frequently accessed blocks on higher levels of the tree, and (3) by avoiding block duplication in different tree nodes during concurrent evictions.
To satisfy these requirements, the server assigns a sequence number to each access during the invocation of its $\mathsf{Server.getPM}$.
This sequence number is used to create a logical \emph{block timestamp} $\mathit{ts}_b = \langle v, a, s \rangle$ for each block $b$ touched during an access, where \emph{version} $v$ is the sequence number of the last write on this block, \emph{access} $a$ is the sequence number of the last read or write of this block, and \emph{sequence} $s$ is the sequence number of the last time the block was moved. 
For instance, if the sequence number of an access to block $b$ is $x$, $\mathit{ts}_b$ after the operation will be $\langle x, x, x \rangle$ if $b$ is written or $\mathit{ts}_b = \langle \_, x, x \rangle$ if $b$ is read. 
Further, any other block that changed its slot during this access' eviction will have its timestamp set to $\langle \_, \_, x \rangle$.

Using three values on the block timestamp ensures that the merge requirements 1-3 described above are satisfied.
When clients perform an access, they may retrieve multiple paths and stashes with different block versions, and merge them into a single version consistent with the highest timestamp found on the position map for each block, respecting the following rule:

\vspace{2mm}
$\langle v,a,s \rangle > \langle v',a',s' \rangle \implies (v > v') \lor (v=v' \land a>a') \lor$\\
\hspace*{3.92cm}$(v=v' \land a=a' \land s>s').$


\begin{table}[!t]
    \centering
    \caption{MVP-ORAM data structures.}
    \begin{tabular}{|c|m{5.7cm}|}
        \hline
        \textbf{Data Structure} & \multicolumn{1}{c|}{\textbf{Description}} \\
        \hline
        \hline
        Block & A tuple $\langle \mathit{addr}, \mathit{data}, \mathit{ts} \rangle$, where $\mathit{addr} \in \{0, \dots, N - 1\}$ is an address identifying the block, $\mathit{data}$ is the data of the block, and $\mathit{ts}$ is the block timestamp $\mathit{ts} = \langle v, a, s \rangle$.\\
        \hline
        Slot & Identifier of a position in a binary tree where a real or dummy block is stored.\\
        \hline
        Bucket & Set of $Z$ slots indexed from $0$ to $Z-1$.\\
        \hline
        Multi-Version Tree & A binary tree of height $L > 0$, where each node contains a set of buckets created concurrently with blocks of different versions. A path $\mathcal{P}_l$ contains the nodes from leaf $l$ to the tree's root. We use the notation $\mathcal{P}_l(\mathit{sl})$ to denote the set of blocks on slot $\mathit{sl}$ of a path from leaf $l$.\\
        \hline
        Position Map & This structure maps block addresses to the current block slot and logical timestamp. $\mathit{pm}[\mathit{addr}] = \langle \mathit{sl}, \mathit{ts} \rangle$ means block with address $\mathit{addr}$ is stored in slot $\mathit{sl}$ with timestamp $\mathit{ts}$.\\
        \hline
        Path Map & Set of tuples $M_\mathit{l}=\{\langle \mathit{addr}, \mathit{sl}, \mathit{ts} \rangle, ...\}$ with the position map updates performed during an access, i.e., for each updated block $\mathit{addr}$, its new slot $\mathit{sl}$ and timestamp $\mathit{ts}$.\\
        \hline
        Stash & List of overflowing blocks.\\
        \hline
        ORAM State & A tuple $\langle \mathcal{T}, \mathcal{S}, \mathcal{H}_\mathit{pathMaps} \rangle$ that stores a multi-version tree $\mathcal{T}$ (with multiple buckets per node), a set of stashes $\mathcal{S}$, one for each version of $\mathcal{T}$, and a set of path maps $\mathcal{H}_\mathit{pathMaps}=\{M_\mathit{l}, ...\}$ that when consolidated define a position map $pm$. We use notation $\mathcal{T}(l)$ and $\mathcal{T}(l,\mathit{sl})$ to denote path $\mathcal{P}_l$ in $\mathcal{T}$ and slot $\mathit{sl}$ of that path. \\
        \hline
        Context & A list of ORAM States, one for each concurrent client that started an access and has not yet finished it.\\
        \hline
        
    \end{tabular}
    \label{tab:mvp_oram_data_structures}
\end{table}
\subsection{The MVP-ORAM Protocol}

\begin{figure*}[!t]
    \centering
    \includegraphics[width=\textwidth]{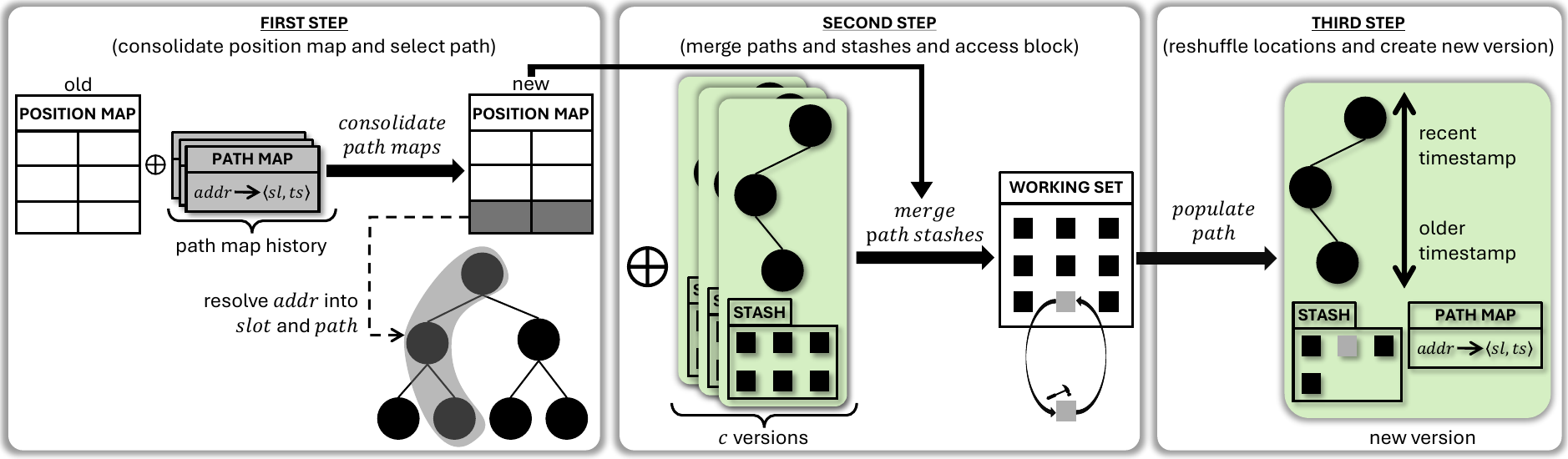}
    \caption{Overview of the MVP-ORAM protocol.}
    \label{fig:mvp_oram_protocol_overview}
\end{figure*}

\input{algorithms/alg_mvp_oram_client}
\input{algorithms/alg_mvp_oram_server}

Here we present a detailed description of the MVP-ORAM protocol.
The data structures used in the protocol are summarized in Table~\ref{tab:mvp_oram_data_structures}.

Client $c_i$ accesses block $b$, identified by address $\mathit{addr}$, by invoking function $\mathsf{access}$, described in Algorithm~\ref{alg:wait_free_client}.
The function accesses $b$ in three steps, just like the protocol of the previous section, by invoking the server functions specified in Algorithm~\ref{alg:wait_free_server}. 
The local functions invoked by clients to merge concurrent versions and create an updated version of the ORAM state are underlined in Algorithm~\ref{alg:wait_free_client} and illustrated in Fig.~\ref{fig:mvp_oram_protocol_overview}.
Their formal specification is deferred to Appendix~\ref{ap:mvp_oram_auxiliary_functions}.
Note that in the algorithms, all data clients send to servers (except for $i$ and $l$, which are basic information) is encrypted, but this is omitted for simplicity.

\textbf{First step (A\ref{alg:wait_free_client}, L2-L5).}
In the first step, client $c_i$ will define a path to retrieve $b$.
It starts by invoking $\mathsf{Server.getPM}$, sending its id and receiving the history of path maps $\mathcal{H}_\mathit{pathMaps}$ and sequence number $\mathit{seq}$ that identifies this access.
$\mathcal{H}_\mathit{pathMaps}$ contains the location updates of blocks evicted so far.
In practice, $c_i$ retrieves new updates since its last access.

When the server receives the request, it stores a reference to the current ORAM state in $c_i$'s context until the client completes its access.
Then, it returns the path map history $\mathcal{H}_\mathit{pathMaps}$, and the access sequence number $\mathit{seq}$ (A\ref{alg:wait_free_server}, L4-L8).

After receiving $\mathcal{H}_\mathit{pathMaps}$, $c_i$ consolidates it into position map $\mathit{pm}$ by invoking $\underline{\mathit{consolidatePathMaps}}$ (first step of Fig.~\ref{fig:mvp_oram_protocol_overview}).
This function applies the updates contained in $\mathcal{H}_\mathit{pathMaps}$ to the local position map, retaining for each block the location update with the highest timestamp.
With the updated position map $\mathit{pm}$ containing the most recent locations of the blocks, $c_i$ discovers the current slot $\mathit{sl}$ where $b$ is stored.
The first step terminates with $c_i$ randomly choosing a path $\mathcal{P}_l$ that contains $\mathit{sl}$ (A\ref{alg:wait_free_client}, L5).

\textbf{Second step (A\ref{alg:wait_free_client}, L6-L12).}
Next, $c_i$ will retrieve $b$ from the server and read/write its content. Since $b$ can either be in the tree or stash, $c_i$ fetches the multi-version path $\mathcal{P}_l$ and stashes from the server by invoking $\mathsf{Server.getPS}$.

The server processes $c_i$'s request (A\ref{alg:wait_free_server}, L9-11) by retrieving the ORAM state from $c_i$'s context.
Then, it collects path $\mathcal{P}_l$ from $\mathcal{T}$ and returns it along with stashes $\mathcal{S}$.
By using the tree and stashes from $c_i$'s context, the server ensures the tree is consistent with the $\mathit{pm}$ consolidated in the previous step.

When $c_i$ receives the response of $\mathsf{Server.getPS}$, it merges the multiple versions of the blocks received in $\mathcal{P}_l$ and $\mathcal{S}$ into a working set $W$ by invoking $\underline{\mathit{mergePathStashes}}$ (second step of Fig.~\ref{fig:mvp_oram_protocol_overview}).
This function uses the consolidated $\mathit{pm}$ as a reference to filter blocks by retaining those with timestamps contained in $\mathit{pm}$, i.e., the more recent versions.
From $W$, $c_i$ retrieves $b$ and updates its content and version if the operation is of type \textit{write} (A\ref{alg:wait_free_client}, L8-12).

\textbf{Third step (A\ref{alg:wait_free_client}, L13-L15).}
In the last step, $c_i$ will evict blocks from $W$ back to the server in a new path and stash.
The redistribution of blocks must ensure two fundamental properties for MVP-ORAM: (1) the stash's expected size is bounded and proportional to the number of concurrent clients, and (2) the most accessed blocks are expected to be either in the highest levels of the tree or in the stash, giving more path options for clients to access them.

\textbf{Eviction in detail.}
This is achieved by the $\underline{\mathit{populatePath}}$ auxiliary function (third step of Fig.~\ref{fig:mvp_oram_protocol_overview}).
First, $c_i$ constructs a new path $\mathcal{P}^*_l$ by placing blocks from $W$ into their correct slots according to $\mathit{pm}$. If multiple blocks are assigned to the same slot due to concurrent accesses, then the one with the highest timestamp is placed in the path, while the rest remain in the working set.

Then, $c_i$ exchanges $Z$ blocks from the working set with up to $Z$ non-dummy blocks from random $Z$ slots of $\mathcal{P}^*_l$, including the accessed block if it was in the path.
This step ensures that the accessed block has the maximum number of paths available to retrieve it in the next access, while the expected stash size remains bounded.
Note that the stash size can decrease if some of the $Z$ selected slots are empty, since in this case, we remove blocks from the stash and add to these empty slots.

Next, $c_i$ reorders blocks in $\mathcal{P}^*_l$ according to their timestamps, with more recently accessed blocks placed higher in the path, thus increasing the number of available paths for frequently accessed blocks.

Finally, $c_i$ adds the remaining blocks in $W$, including the accessed block, to a new stash $S$. 
It also updates the timestamp of blocks that were moved and update their locations on a new path map $M_l$.
The function then returns $\mathit{P}^*_l$, $S$, and $M_l$.

After populating the path, $c_i$ invokes $\mathsf{Server.evict}$ to send $M_l$, $S$, and $\mathcal{P}^*_l$ to the server (A\ref{alg:wait_free_client}, L14).
When the server receives an eviction request from the client, it first reads and cleans the client's context and obtains the current ORAM state (A\ref{alg:wait_free_server}, L13-14).
Then, it applies the modifications proposed by the client (A\ref{alg:wait_free_server}, L15-19) by (1) updating path $\mathcal{P}_l$ in the current ORAM state, replacing the slots read by the client by the ones received in the eviction, (2) updating the set of stashes by replacing the retrieved stashes with the new stash, and (3) adding the received path map to the path map history.\footnote{To prevent unlimited growth of the history, clients send the consolidated position map every $\gamma$ accesses, deleting outdated path maps.}
These updated data structures are then stored in the ORAM state.

%% file: algorithms/alg_mvp_oram_client.tex
\begin{algorithm}[t!]
\SetKwProg{Fn}{Function}{}{}
\DontPrintSemicolon
\caption{MVP-ORAM client $c_i$.}
\label{alg:wait_free_client}
{\small

\Fn{$\mathsf{access}(c_i, \mathit{op}, \mathit{addr}, \mathit{data}^*)$}{
    $\langle \mathcal{H}_\mathit{pathMaps}, \mathit{seq}\rangle \gets \mathsf{Server.getPM}(c_i)$\\
    $\mathit{pm} \gets \underline{\mathit{consolidatePathMaps}}(\mathcal{H}_\mathit{pathMaps})$\\
    $\langle \mathit{sl}, \_ \rangle \gets \mathit{pm}[\mathit{addr}]$\\
    $l \gets $ random path that passes through slot $\mathit{sl}$\\
    $\langle \mathcal{P}_l, \mathcal{S} \rangle \gets \mathsf{Server.getPS}(c_i, l)$\\
    $W \gets \underline{\mathit{mergePathStashes}}(\mathcal{P}_l, \mathcal{S}, \mathit{pm})$\\
    
    \If{$\mathit{op} = \mathit{write}$}{
        $\mathit{data} \gets \mathit{data}^*$; 
        $v \gets \mathit{seq}$\\
    }
    \Else{
        $\langle \_, \mathit{data}, \langle v, \_, \_ \rangle \rangle \gets W[\mathit{addr}]$\\
    }
    $W[addr] \gets \langle addr, \mathit{data}, \langle v, \mathit{seq}, \mathit{seq} \rangle \rangle$\\
    $\langle \mathcal{P}_l^*, S, M_\mathit{l} \rangle \gets \underline{\mathit{populatePath}}(W, l, \mathit{addr}, \mathit{pm}, \mathit{seq})$\\
    $\mathsf{Server.evict}(c_i, M_\mathit{l}, \mathcal{P}_l^*, S)$\\
    \Return $\mathit{data}$
    }
}
\end{algorithm}

%% file: algorithms/alg_mvp_oram_server.tex
\begin{algorithm}[t!]
\SetKwProg{Fn}{Function}{}{}
\DontPrintSemicolon
\caption{MVP-ORAM server.}
\label{alg:wait_free_server}
{\small
\Proc{setup($\mathcal{T}, \mathcal{S}$)}{
    $\mathit{oramState} \gets \langle \mathcal{T}, \mathcal{S}, \emptyset \rangle$;
    $\mathit{nextSeq} \gets 1$\\
    $\forall c_i \in \Gamma: \mathit{context}[c_i] \gets \bot$
}

\Fn{$\mathsf{getPM}(c_i)$}{
    $\mathit{seq} \gets \mathit{nextSeq}$;
    $\mathit{nextSeq} \gets \mathit{nextSeq} + 1$\\
    $\mathit{context}[c_i] \gets \mathit{oramState}$\\
    $\langle \_, \_, \mathcal{H}_\mathit{pathMaps} \rangle \gets \mathit{oramState}$\\
    \Return $\langle \mathcal{H}_\mathit{pathMaps}, \mathit{seq}\rangle$
}

\Fn{$\mathsf{getPS}(c_i, l)$}{
    $\langle \mathcal{T}, \mathcal{S}, \_ \rangle \gets \mathit{context}[c_i]$\\
    \Return $\langle \mathcal{T}(l), \mathcal{S} \rangle$
}

\Proc{$\mathsf{evict}(c_i, M_\mathit{l}, \mathcal{P}_l^*, S)$}{
    $\langle \mathcal{T}, \mathcal{S}, \_ \rangle \gets \mathit{context}[c_i]$;
    $\mathit{context}[c_i] \gets \bot$\\
    $\langle \mathcal{T}^c, \mathcal{S}^c, \mathcal{H}^c_\mathit{pathMaps} \rangle \leftarrow \mathit{oramState}$\\
    \For(\tcp*[f]{update the tree}){$\mathit{sl} \in \mathcal{T}(l)$}{
        $\mathcal{T}^*(l,\mathit{sl}) \gets (\mathcal{T}^c(l,\mathit{sl}) \setminus \mathcal{T}(l,\mathit{sl})) \cup \mathcal{P}^*_l(\mathit{sl})$\\
    }    
    $\mathcal{S}^* \gets (\mathcal{S}^c \setminus \mathcal{S}) \cup \{S\}$\\
    $\mathcal{H}^*_\mathit{pathMaps} \gets \mathcal{H}^c_\mathit{pathMaps} \cup \{  M_\mathit{l} \}$\\
    $\mathit{oramState} \gets \langle \mathcal{T}^*, \mathcal{S}^*, \mathcal{H}^*_\mathit{pathMaps} \rangle$
}
}
\end{algorithm}

%% file: content/sec_bft_oram.tex
\section{Byzantine Fault-Tolerant ORAM}
\label{sec:bft_oram}

The previous section detailed MVP-ORAM, a protocol that can handle concurrent clients accessing an ORAM while satisfying wait-freedom and linearizability.
We now describe how MVP-ORAM can be replicated using BFT SMR to tolerate fully malicious servers, ensuring data integrity and availability while preserving data and access-pattern secrecy.
For this, we use a Byzantine Fault-Tolerant State Machine Replication (BFT SMR) protocol~\cite{castro1999practical,schneider1990implementing}.

BFT SMR is a classical technique for implementing fault-tolerant systems by replicating stateful, deterministic services on multiple fault-independent servers~\cite{schneider1990implementing}.
Most BFT SMR implementations allow tolerating $t$ Byzantine faults with $n > 3t$ servers. 
This is possible by ensuring that each server starts in the same initial state and executes the same sequence of operations deterministically.
Ensuring such a total order of operations on all correct servers requires executing Byzantine consensus~\cite{castro1999practical,yin_2019} to make the replicas agree on the next set of client operations to be executed.

MVP-ORAM solves three fundamental challenges that are required for replicating an ORAM through BFT SMR.
First, it ensures the server-side algorithm is fully deterministic.
Second, it requires only three invocations of state machine operations ($\mathsf{getPM}$, $\mathsf{getPS}$, and $\mathsf{evict}$) for performing an access.
Third, and most importantly, it makes ORAM accesses wait-free.

In detail, we execute $n$ server replicas using a BFT SMR middleware (e.g.,~\cite{bessani_2014}) to ensure that the invocation of the three server operations used in Algorithm~\ref{alg:wait_free_client} is reliably disseminated in total order to all servers.
Each server executes those functions locally, exactly as specified in Algorithm~\ref{alg:wait_free_server}, and sends replies to the invoking clients, which consolidate a single response for each invocation by waiting for $t+1$ matching replies.

\textbf{Improving BFT ORAM performance.}
However, a direct implementation of MVP-ORAM in a BFT SMR system will significantly increase bandwidth usage, making the protocol highly inefficient.
As such, we propose a series of optimizations that make the BFT version of MVP-ORAM more practical.
Most of these optimizations aim to decrease bandwidth requirements (analyzed in \S\ref{sec:performance_analysis}).

The first optimization is related to the execution of consensus over metadata.
Byzantine consensus protocols typically select one (the leader, as in PBFT~\cite{castro1999practical}) or more (the DAG block proposers, as in  Mysticeti~\cite{babel2025mysticeti}) proposers to disseminate batches of requests to be ordered.
In such protocols, the client sends its request to all replicas, and proposers re-disseminate the request along with ordering information.
However, if the clients' requests are large (as in our case, where $\mathsf{evict}$ must send a path, path map, and stash), the proposer's bandwidth will be easily exhausted.
To solve this, the client can send the $\mathsf{evict}$ parameters directly to the servers without ordering them and send only their hash for ordering through BFT SMR.
Using the hashes, replicas retrieve the operation parameters and process the request as usual, thus significantly reducing the best-case bandwidth usage.

The second optimization aims to decrease the bandwidth usage of server responses.
In traditional BFT-SMR, all correct servers respond with the result of executing the client-issued operation.
This negatively affects bandwidth usage, especially during $\mathsf{getPS}$ when servers send multiple paths and stashes.
We reduce this impact by employing an optimization introduced in PBFT~\cite{castro1999practical} in which, for each ordered request, we choose a server that responds with the full reply while others respond with its cryptographic hash.

The client randomly selects a server that will send a full reply and accepts the response when the hash of this reply matches $t$ hashes sent by other servers.
If the obtained response does not match the hashes, the client asks $t$ servers to send the full content.



\textbf{Encryption keys management.} 
ORAM protocols typically assume that clients manage and coordinate the shared cryptographic keys needed to encrypt the stored data or that there is a trusted third party (e.g., a proxy) that manages those keys.
In MVP-ORAM, we remove this assumption through the use of secret sharing~\cite{shamir1979share}, more specifically, Dynamic Proactive Secret Sharing (DPSS), which is more appropriate for practical confidential BFT SMR systems~\cite{vassantlal_2022}.

When the servers are set up, the client generates a new random cryptographic key and secretly shares it, sending a different share to each server.
The servers keep this share as part of their internal state.
Then, when a client starts a new access and invokes $\mathsf{Server.getPM}$, the servers send their stored shares along with the response.
The client reconstructs the key using the received shares and uses it in all cryptographic operations during an access.

For simplicity, we assume the same key is used to encrypt all ORAM data.
However, using this approach, we could have different keys for different data structures or even for different versions of the same data structures.




%% file: content/sec_security_analysis.tex
\section{Security and Complexity Analysis}
\label{sec:analysis}

We now discuss the security, correctness, and theoretical performance of MVP-ORAM. 

\subsection{Security Analysis}
\label{sec:security_analysis}

The security of MVP-ORAM is defined by Theorem~\ref{the:main}.

\begin{theorem}{1}\label{the:main}
    Given an ORAM of size $N$ with $c$ concurrent clients issuing requests from a distribution of accesses $\mathcal{D}$, then MVP-ORAM is an $\mu(N,c,\mathcal{D})$-secure Asynchronous Wait-Free ORAM as per Definition~\ref{def:opram}.
\end{theorem}

To prove this theorem, we must show that MVP-ORAM fulfills both the \textit{Correctness} and \textit{Obliviousness} properties of Definition~\ref{def:opram}. 
\ifbool{extendedVersion}{
We next outline these proofs, leaving their complete versions for Appendices~\ref{ap:sec:correctness_proofs} and~\ref{ap:sec:obliviousness}, respectively.
}{
We next outline these proofs, leaving their complete versions for the extended version of this paper~\cite{mvp_oram_extended_version}.\footnote{
The numbering of theorems and lemmas referenced in this paper is the same as in the extended version.}
}

\textbf{Correctness.}
MVP-ORAM provides an abstraction of a memory that can be written and read through the $\mathsf{access}$ operation without revealing to the server which memory/block address was accessed.
Therefore, from the distributed computing point of view, we have to prove our construction implements $N$ safe and live atomic read/write registers~\cite{lamport1986}.
This requires proving all memory operations finish (wait-freedom~\cite{herlihy_1991}) and that they are safe under concurrent accesses (Linearizability~\cite{herlihy_1990}).

We prove \emph{safety} by first showing that all access operations preserve the most up-to-date version of each accessed block on evictions (Lemma~\ref{the:statepreservation}).
This is important because it enables us to prove the safety of each logical block individually.

\ifbool{extendedVersion}{}{
    \begin{lemma}{3}[State Preservation]\label{the:statepreservation}
        With the exception of the block $b$ accessed during a write, an execution of $\mathsf{access}$ preserves the state of the ORAM.
    \end{lemma}
}

To prove safety under concurrent access of a single block, we have to prove each concurrent history of operations invoked on MVP-ORAM satisfies linearizability~\cite{herlihy_1990}.
To prove this, we present a series of transformations of the observed history (respecting the MVP-ORAM algorithms) until we prove that the resulting high-level history containing only read and write operations is linearizable (Theorem~\ref{the:linearizability}).
As part of this proof, we show that every read of a block returns the value written in the closest preceding write. 

\ifbool{extendedVersion}{}{
    \begin{theorem}{2}[Linearizability]\label{the:linearizability}
        For each memory position $b$, MVP-ORAM's read ($\mathsf{access}(\_,read,b,\bot)$) and write ($\mathsf{access}(\_,write,$ $b,\_)$) operations satisfy linearizability.
    \end{theorem}
}

We prove that MVP-ORAM guarantees \emph{wait-freedom} by showing that every step of the MVP-ORAM protocol terminates, assuming that the underlying BFT SMR guarantees liveness.
Hence, every invocation of $\mathit{acccess}$ by a correct client terminates (Theorem~\ref{the:waitfreedom}).

\ifbool{extendedVersion}{}{
    \begin{theorem}{3}[Wait-freedom]\label{the:waitfreedom}
        Every invocation of MVP-ORAM's $\mathsf{access}$ by a correct client terminates. 
    \end{theorem}
}

\textbf{Obliviousness.}
For this analysis, we assume $\mathcal{D}$ follows a Zipfian distribution~\cite{kingsley1932selected}, meaning that the frequency $f(r)$ of accessing the $r^\mathit{th}$ most frequently accessed block (rank~$r$) decreases proportionally to $r^{-\alpha}$.
For example, when $\alpha = 1$, $27\%$ of the blocks are accessed much more frequently than the others, with their access frequencies decreasing as their rank increases.

In Path ORAM, each block is mapped to a specific path in the tree, and when a client accesses a block, it retrieves the entire path and randomly re-assigns the block to a new path before eviction.
In contrast, MVP-ORAM allows a block to be accessed through any path that contains it.
Besides, the block is not reassigned to a new path after access; instead, it remains in the stash until it is evicted to the path in a future access to a different block.
The next path used to request the block is only determined when it is accessed again.
Additionally, only $Z$ random blocks are evicted at a time from the stash, and instead of blocks being randomly placed in a path, they are sorted so that more frequently accessed blocks (i.e., with a higher timestamp) are placed up in the tree, giving them more possible paths for future requests.

Given this, we analyze MVP-ORAM's security in three different scenarios:
(1) a single client accesses the ORAM once per timestep,
(2) multiple clients access different blocks per timestep, and
(3) multiple clients access the same block in the same timestep.

When a sequence of requests~$\overrightarrow{y}$ is performed, the servers see $A(\overrightarrow{y})$, which is the same sequence of requests but transformed by the ORAM. 
When $c$ clients access different blocks within the same timestep (case 2), each of them selects a random leaf, resembling the behavior of a single client performing $c$ accesses across $c$ consecutive timesteps (case 1).
In Lemma~\ref{thm:sec:cases:1:2} we show that $A(\overrightarrow{y})$ becomes indistinguishable from a random sequence of requests with high probability in case 1 (case 2 is ommitted, since they are similar).

\ifbool{extendedVersion}{}{
    \begin{lemma}{7}\label{thm:sec:cases:1:2}
        When a single client accesses the ORAM per timestep, the access pattern $A(\overrightarrow{y})$ observed by the server during a sequence of requests $\overrightarrow{y}$ is computationally indistinguishable from a random sequence with high probability.
    \end{lemma}
}

For the last case, we must show that the access pattern generated by MVP-ORAM might be distinguishable from a random access pattern, particularly for blocks located near the leaves.
We establish this result by computing the statistical distance~\cite{reyzin2011} between the access pattern generated by MVP-ORAM and a random access pattern (Theorem~\ref{thm:oram-same-block-not-negligible}).
The intuition behind computing such a distance is as follows.
In a random sequence of size $c$, we expect to observe $c$ distinct leaves being accessed.
However, when $c$ clients simultaneously access the same block --- particularly if the block is located near the leaves --- the expected number of distinct leaves involved may be less than $c$. 
As such, we compare the distribution of the number of distinct leaves in a random sequence with the distribution of the number of distinct leaves generated in the worst-case execution of MVP-ORAM.

\ifbool{extendedVersion}{}{
    \begin{theorem}{4}\label{thm:oram-same-block-not-negligible}
        Given $c, N \in \mathbb{N}$, $\alpha \in \mathbb{R}$, and $D\in\mathcal{U}$, the statistical distance between a random sequence of size~$c$ and the access pattern generated by MVP-ORAM is bounded by $\mu(N,c,D(\alpha))$.
    \end{theorem}
}

\begin{figure}[t]
    \centering
    \includegraphics[width=\columnwidth]{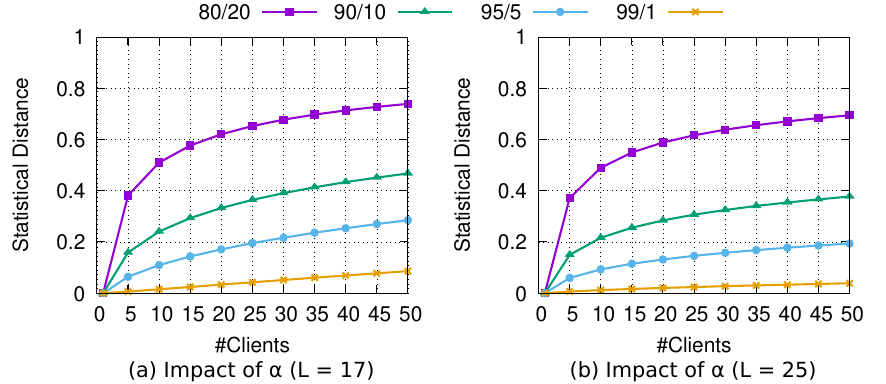}
    \vspace{-0.5cm}
    \caption{Statistical distance simulation for different heights ($L$) and percentage of frequently accessed blocks (i.e., Zipfian parameter $\alpha$).
    $80/20$ means $20\%$ of blocks are accessed with probability $80\%$ (i.e., $\alpha=0.90$ in (a) and $\alpha=0.88$ in (b)).}
    \label{fig:statistical_distance_analyses}
\end{figure}

Fig.~\ref{fig:statistical_distance_analyses} shows the statistical distance of the distributions (a worst-case measurement), considering different values of $L$ and $\alpha$.
As expected, the distance decreases as we make the accesses more skewed (i.e., as we increase $\alpha$) and as we increase the tree size $L$.
If we decrease $\alpha$ to a point where the number of frequently accessed blocks approximates $N$ (i.e., blocks are uniformly accessed), the statistical distance will be near $1$ (worst security).
Nonetheless, it is worth noting that although the statistical security in this case is far from good, the concrete probabilities of leakage are very small.
The expected probability of $c$ clients accessing the same block using the same path with uniform access distribution is $\frac{1}{2N^{c-1}}$.
For example, this probability for two clients and a tree of $N=2^{18}$ is less than $0.0002\%$.

\subsection{Stash Size Analysis}
\label{sec:outline_stash_size_analyses}

The performance of MVP-ORAM is directly tied to its stash size, with larger stash sizes leading to reduced performance.
Hence, it is crucial to ensure that the size of the stash does not grow indefinitely.
\ifbool{extendedVersion}{
Here, we outline the stash size analysis, with the complete proofs presented in Appendix~\ref{ap:sec:stash_size_analysis}.
}{
Here, we outline the stash size analysis, with the complete proofs presented in the extended version~\cite{mvp_oram_extended_version}.
}

Recall that the adversary can control the number of concurrent clients in MVP-ORAM through network scheduling.
We argue that the adversary can maximize the stash size by maximizing concurrency.
In other words, the largest stash size occurs when, at each timestep, the maximum number of clients $c$ concurrently access the ORAM.

To analyze the stash size under this worst-case scenario, we show that concurrent clients add approximately $cZ$ blocks to the stash in each timestep.
When the stash size is small, clients may be unable to remove $cZ$ blocks due to overlaps caused by multiple clients selecting the same blocks.
However, as the stash size grows to $O(c\log{N})$, the system reaches a point where concurrent clients can remove approximately $cZ$ blocks from the stash, being $Z$ a small constant. 
At this point, the stash size stabilizes as the rate of blocks being added to the stash approximately matches the rate of blocks being removed.
We formalize this result in Theorem~\ref{thm:stash:size}.
\ifbool{extendedVersion}{}{
\begin{theorem}{5}\label{thm:stash:size}
    Under the worst-case scenario concerning concurrency, the expected stash size at any timestep is $O(c \log{N})$.
\end{theorem}
}

\subsection{Bandwidth and Storage Analysis}
\label{sec:performance_analysis}

We now analyze the communication complexity of MVP-ORAM.
To simplify this analysis, we omit the cost of sending constant values such as block and client ids.
We begin by considering a single-server (non-BFT) setup.
The $\mathsf{getPM}$ reply contains up to $c$ path maps of size $O(c + \log N)$, resulting in $O(c(c+\log N))$ bits.
The $\mathsf{getPS}$ reply has size $O(c^2 \log N)$, as it contains $c$ paths and stashes.
Finally, an $\mathsf{evict}$ has size $O(c\log N)$ (a consolidated path and stash).

When considering an $n$-server setup, by default, the communication goes up by at least a factor of $n$, as data must be sent to all servers.
Note that this multiplicative factor only occurs in requests, as replies benefit from the optimization of $n-1$ servers sending hashes; i.e., $O(|\mathit{Reply}|+n)$.
Requests also account for the cost of Byzantine consensus, which we use only for metadata ordering, which results in $O(n|\mathit{Request}|+\mathit{Consensus})$.

By using a linear consensus protocol (e.g., HotStuff~\cite{yin_2019}) or a quadratic protocol~\cite{castro1999practical,bessani_2014} in a small group of replicas (i.e., $n \ll c\log N$), the $\mathit{Consensus}$ term loses importance, and the complexity boils down to the values of Table~\ref{tab:opt_bft_mvp_oram_communication_complexity}.

Putting it all together, MVP-ORAM incurs a communication complexity of 
$O((n+c) c\log N)$.
This shows two nice properties of our protocol.
First, it is an \emph{adaptive wait-free construction}~\cite{afek_95}, meaning that its performance depends on the number of active concurrent clients $c$, not on the total number of existing clients.
Second, its bandwidth usage is linearly proportional to the number of servers.
This means that, with low concurrency, the bandwidth usage approximates that of Path ORAM replicated to $n$ servers.

\setlength{\tabcolsep}{3.5pt}
\begin{table}[!t]
 \centering
 \caption{MVP-ORAM communication complexity when storing $N$ blocks using $n$ servers with $c$ active clients.}
 {\footnotesize
 \begin{tabular}{l|c|c|c|}
    \cline{2-4}
        & \multicolumn{3}{c|}{\textbf{$n$ servers with optimizations}}\\
     \hline
     \multicolumn{1}{|l|}{\textbf{Operation}} & \textbf{Request} & \textbf{Response} & \textbf{Total Operation} \\
     \hline
     \multicolumn{1}{|l|}{$\mathsf{getPM}$} & $O(n)$ & $O(n + c(c + \log N))$ & $O(n + c(c + \log N))$\\
     \multicolumn{1}{|l|}{$\mathsf{getPS}$} & $O(n)$ & $O(n + c^2 \log N)$ & $O(n + c^2 \log N)$\\
     \multicolumn{1}{|l|}{$\mathsf{evict}$} & $O(n c \log N)$ & $O(n)$ & $O(n c \log N)$\\ 
    \hline
 
 \end{tabular}
}
 \label{tab:opt_bft_mvp_oram_communication_complexity}
\end{table}

In terms of storage, each server needs to store the original position map ($O(N)$), the tree database ($O(N)$), and stash ($O(c \log N)$) plus up to $c$ updates performed by different clients ($O(c^2 \log N )$) and not yet consolidated in the database.
This leads to $O(N+c^2 \log N)$ server storage requirement.

%% file: content/sec_strong_mvp_oram.tex
\section{Strong Multi-Version Path ORAM}
\label{sec:strong_mvp_oram}

MVP-ORAM gives the same guarantees as non-wait-free ORAMs when a single client accesses the ORAM per timestep, or multiple clients concurrently access distinct addresses.
However, obliviousness can be compromised if clients try to access the same block in the same timestep, particularly if this block is deep in the tree.
In the worst-case scenario (ORAM accesses follow a uniform distribution), clients will have few paths available to access the blocks since most of them will be located in slots near the leaves.
In this case, the adversary can potentially observe concurrent clients requesting the same paths, compromising obliviousness.

In this section, we outline a variant of MVP-ORAM that preserves wait-freedom and obliviousness, at the cost of executing extra dummy requests for each real access and assuming the relative speed of clients in executing an access is approximately the same, i.e., no client is significantly faster or slower than the others.
\ifbool{extendedVersion}{
A full description of the Strong MVP-ORAM protocol and its proof is presented in Appendix~\ref{ap:strong_mvp_oram_algorithms}.
}{
A full description of the Strong MVP-ORAM protocol and its proof is presented in the extended version of the paper~\cite{mvp_oram_extended_version}.
}

\subsection{Mitigating the Risk of Concurrent Accesses}

Let $\sigma \geq 0$ be a security parameter defining the number of dummy accesses sent for each real access.
To \emph{access} a block $b$ with address $\mathit{addr}$, each client implicitly builds a schedule of MVP-ORAM accesses (with the three steps described before) for $\sigma + 1$ consecutive timesteps.
The key idea is to make at most one client access $\mathit{addr}$ in a timestep, while others execute dummy accesses.
Let $\tau_i\in \{0, \dots, \sigma\}$ be the timestep when $c_i$ performs the real access.
The schedule of $c_i$ is built as follows: $c_i$ first executes $\tau_i$ dummy accesses, then its real access, and concludes with $\sigma-\tau_i$ dummy accesses.

For this strategy to be effective, concurrent clients must define distinct timesteps for their actual accesses.
A client $c_i$ discovers $\tau_i$ during its first MVP-ORAM access.
In the invocation of $\mathsf{getPM}$, each client additionally sends the (encrypted) $\mathit{addr}$ that it wants to access, and the server stores this information in a set $\mathcal{A}$, which is returned together with the $\mathsf{getPM}$'s response.

When $c_i$ receives $\mathcal{A}$, it defines $\tau_i$ by counting the number of ongoing accesses to $\mathit{addr}$, i.e., it sets $\tau_i$ as the number of accesses to $\mathit{addr}$ in $\mathcal{A}$ minus one (to ignore its own access).
For example, if there is one access to $\mathit{addr}$ in $\mathcal{A}$ (it's own access, just declared), then $\tau_i=0$ (the first access will retrieve the target block, and all the other $\sigma$ will be dummy).
If there are two accesses, then $\tau_i=1$, which implies a single dummy access, the real access, and $\sigma-1$ dummy accesses.

\subsection{Security Analysis of Strong MVP-ORAM}

Our strategy requires every client to always perform $\sigma+1$ accesses, ensuring that in case concurrent clients access the same address, they do so at different timesteps.
Since the server does not know which address a client is accessing, the adversary will be unable to distinguish between real and dummy accesses and whether clients are coordinating to hide access to the same address.
Hence, as long as $\sigma$ matches the maximum concurrency of the system ($c \leq \sigma + 1$), then this approach ensures no two clients access the same block in the same timestep and it is possible to show that the access pattern will be indistinguishable from a random access pattern, matching the obliviousness of parallel ORAM without giving up wait-freedom but assuming synchrony.

\begin{figure}[t]
    \centering
    \includegraphics[width=\columnwidth]{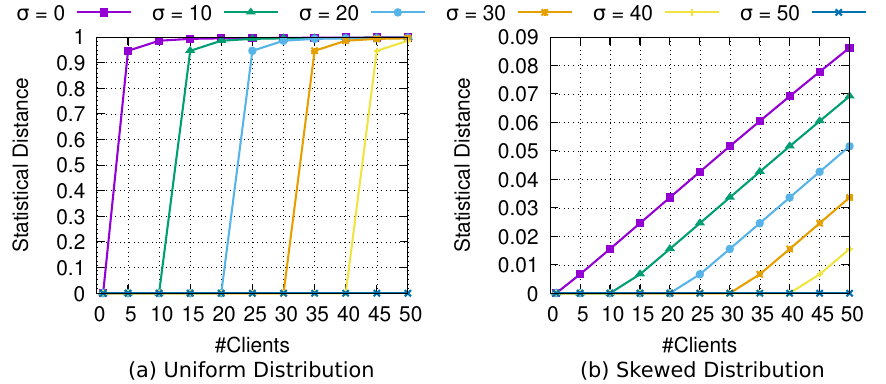}
    \caption{Statistical distance simulation for different numbers of dummy accesses considering uniform ($\alpha << 1$) and skewed ($\alpha >> 1$) block selection distributions in a tree of height $17$. In (a), clients access $99\%$ of blocks with probability $99\%$. In (b), clients access $1\%$ of blocks with probability $99\%$.}
    \label{fig:tvd-sigma}
\end{figure}

Fig.~\ref{fig:tvd-sigma} shows how different values of $\sigma$ affect the statistical distance between access distributions.
The distance is zero (statistical security) when the number of clients $c$ is at most $\sigma + 1$.
However, when $\sigma$ is less than $c$, the distance increases as multiple clients might access the same block in the same timestep.
Nevertheless, it shows that our stronger variant improves statistical distance even when block selection follows a near-uniform distribution, i.e., most accessed blocks are near leaves.
Specifically, it reduces the statistical distance from near $1$ to zero when $c \leq \sigma+1$.
In this condition, MVP-ORAM's obliviousness approximate that of collision-free parallel ORAM~\cite{boyle2015} (Theorem~\ref{thm:poram_equivalence}).

\ifbool{extendedVersion}{}{
    \begin{theorem}{6}
    \label{thm:poram_equivalence}
        Given $c, \sigma, N \in \mathbb{N}$, if $c \leq \sigma + 1$, then Strong MVP-ORAM's access pattern is indistinguishable from a random access pattern with negligible probability in $N$.
    \end{theorem}
}

%% file: content/sec_evaluation.tex
\section{Implementation \& Evaluation}
\label{sec:evaluation}

We implemented a prototype of MVP-ORAM and conducted a set of experiments on AWS to evaluate its stash size and concrete performance under different configurations.

\subsection{Implementation}
\label{sec:implmentation}

We built a prototype of MVP-ORAM in Java by extending COBRA~\cite{vassantlal_2022}, a confidential BFT SMR framework based on DPSS.
COBRA itself relies on BFT-SMaRt~\cite{bessani_2014}, a replication library that provides all the features required for practical BFT SMR systems.
BFT-SMaRt implements a Verifiable and Provable Consensus~\cite{sousa_2012} based on Cachin's Byzantine Paxos~\cite{cachin_2009}, which is similar to PBFT~\cite{castro1999practical}, i.e., it requires three communication steps and has a quadratic message complexity in the common case.
This is the consensus algorithm executed during the invocation of $\mathsf{Server.getPM}$, $\mathsf{Server.getPS}$, and $\mathsf{Server.evict}$.
Thus, through COBRA (and BFT-SMaRt), we can easily implement all features required by MVP-ORAM.

We have also implemented a safeguard against an unbounded number of concurrent clients exhausting bandwidth and storage by making ORAM servers only allow $c_\mathit{max}$ clients to perform concurrent accesses.
This is important to avoid memory trashing and ensure the stability of the system under high load. 
Given enough resources, $c_\mathit{max}$ could naturally match the number of clients accessing the shared ORAM.

Our implementation and all the code used for the experiments are available on the project's web page~\cite{mvp_oram}.

\subsection{Setup and Methodology}

Our experimental evaluation was performed in the AWS cloud using two types of instances.
The servers were executed in $n$ \textit{r5n.2xlarge} instances, each having $8$ vCPU, $64$ GB of RAM, and $8.1$ Gbps baseline network bandwidth.
The clients were executed in $6$ \textit{c5n.2xlarge} instances, each with $8$ vCPU, $21$ GB of RAM, and $10$ Gbps baseline network bandwidth.

We installed Ubuntu Server 22.04 LTS and OpenJDK 11 on all of the machines.
The experiment analyzing stash size was conducted by simulating concurrent accesses to ORAM on a single machine.
Performance was measured by executing servers in $n$ machines and clients in the remaining ones.
Throughput and latency measurements were collected from a single server and client machine, respectively.
Unless stated otherwise, we measure MVP-ORAM considering a database of 
$1$ GB configured in a tree of height $L=17$ and with each node containing $Z=4$ blocks of $4096$ bytes each.
Therefore, our database contains 
$N = 262143$ blocks.
For each operation, clients randomly choose among all the blocks in a Zipfian distribution with $\alpha = 1.0$, and the operation type is picked uniformly at random between \emph{read} and \emph{write}.
The graphs showing throughput and latency display the average (plus standard deviation) of data collected over $6$ minutes for each experiment.

\subsection{MVP-ORAM Stash Size}
\label{sec:evaluation_stash_size}

We start by studying the impact of the number of frequently accessed blocks and concurrent clients on the stash usage.
Since the stash size changes over time with the number of accesses and does not depend on the size of the blocks, we run local simulations on a single server with a small block size to execute $500k$ concurrent accesses.

Since altering $\alpha$ affects the number of frequently accessed addresses and $c$ affects concurrency, this experiment analyzes the impact of changing those parameters.
The bucket size $Z$ also affects the stash usage.
However, it has the expected result: the stash size decreases as we increase $Z$ since the probability of selecting empty slots from the path increases, which increases the number of blocks evicted from the stash.

The results are presented in Fig.~\ref{fig:stash_experiments}.
There are two main takeaways from these experiments:
(1) the stash size stabilizes after some operations,
and (2) the maximum stash size increases as we decrease $\alpha$ and increase $c$.

\begin{figure}[!t]
    \centering
    \includegraphics[width=\columnwidth]{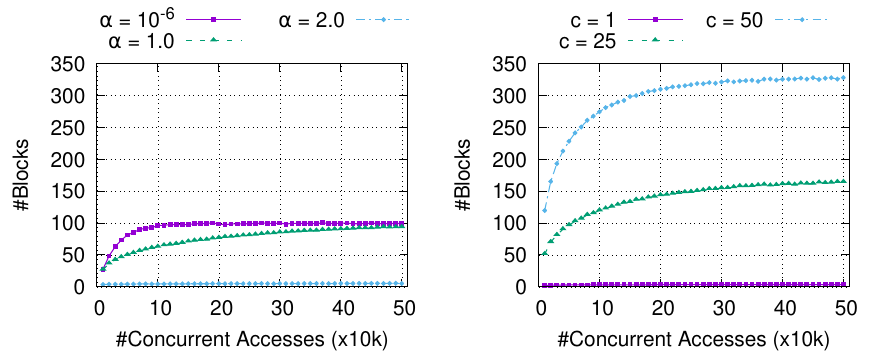}
    \vspace{-0.5cm}
    \caption{Stash size with $c=15$ and different Zipfian's exponent values, and $\alpha=1.0$ and different number of clients.
    Each point represents an average of 50k accesses.}
    \label{fig:stash_experiments}
\end{figure}

As analyzed in \S\ref{sec:outline_stash_size_analyses}, the rate of moving blocks from the stash to the tree and vice-versa is equal when the stash size reaches its expected maximum size.


We study the impact of Zipfian exponent $\alpha$ by experimenting with low ($\alpha=10^{-6}$), medium ($\alpha=1.0$), and high ($\alpha=2.0$) contention levels in the accessed blocks. 
Accordingly, clients approximately access $90\%$, $27\%$, and $0.001\%$ of blocks, respectively, with a probability greater than $90\%$.
Decreasing this parameter increases the number of distinct blocks that clients access and move to the stash.
However, the number of concurrent clients heavily influences the number of distinct blocks swapped between the stash and the tree.
Since $c$ is fixed, the stash size is defined by the number of frequently accessed blocks, which is higher for small $\alpha$ values.

Increasing the number of concurrent clients also increases the stash size.
Since blocks are uniformly sampled, increasing the number of clients increases the number of common blocks selected from the stash.
Thus, the clients remove a few distinct blocks from the stash, increasing its size.
This experiment shows that the number of concurrent clients heavily dominates the stash size, confirming our theoretical analysis (\S\ref{sec:outline_stash_size_analyses}).

\subsection{MVP-ORAM Performance}
\label{sec:evaluation_fault_tolerance}

The next set of experiments aims to measure the impact of BFT replication and the number of concurrent clients on MVP-ORAM's performance.
Fig.~\ref{fig:fault_tolerance} shows the throughput and latency of MVP-ORAM with varying $n$ and $c$.


\begin{figure}[!t]
    \centering
    \includegraphics[width=\columnwidth]{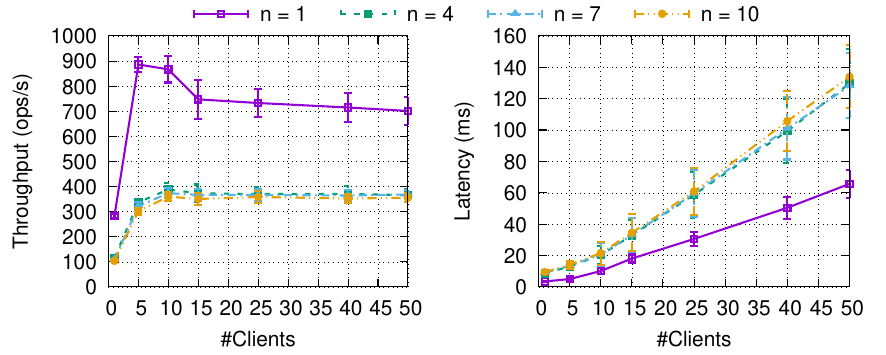}
    \vspace{-0.5cm}
    \caption{MVP-ORAM throughput and latency for different numbers of replicas ($n$) and concurrent clients ($c$).}
    \label{fig:fault_tolerance}
\end{figure}

The overall throughput of the single-server setup is higher than that of the system tolerating failures.
Specifically, the peak throughput is $2.2\times$ the system's throughput tolerating one failure.
Recall that a system tolerating $t$ faults requires $n > 3t$ servers, and clients must wait for at least $t+1$ matching server responses before continuing.
Therefore, as the number of tolerated faults increases, clients must send requests to more servers and wait for their responses, which increases latency and reduces the number of requests they send.
In turn, this slightly reduces the overall throughput, i.e., it drops by $7\%$ when the number of servers goes from $4$ to $10$.

Although throughput remains stable regardless of the number of clients, latency increases as the number of clients increases.
This is the effect of bounding the maximum number of concurrent accesses to $c_{max} = 10$ to avoid memory and bandwidth trashing.

To better understand the factors contributing to MVP-ORAM's performance, we break down the access latency in the three operations that constitute an access.
Fig.~\ref{fig:latency_breakdown} shows the individual latency of each protocol phase, which corresponds to an SMR operation for different numbers of servers and up to $c_{max}$ clients, when queueing is not a factor.
With a single client (no concurrency - first bar of each group), $\mathsf{evict}$ is the most costly phase.
When the number of concurrent clients increases (second and third bars), $\mathsf{getPS}$ becomes more prominent.
This can be explained by the operations' request/response sizes for different clients, as shown in Table~\ref{tab:request_size}.
Among the three phases of MVP-ORAM, $\mathsf{getPS}$ requires bandwidth that is quadratic in the number of clients (see Table~\ref{tab:opt_bft_mvp_oram_communication_complexity}).

\begin{figure}[!t]
    \centering
    \includegraphics[width=0.6\columnwidth]{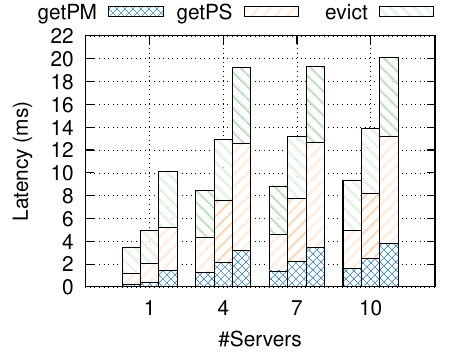}
    \caption{Access latency breakdown. Each bar group considers runs with $1$ (left), $5$ (center), and $10$ (right) concurrent clients.}
    \label{fig:latency_breakdown}
\end{figure}

\begin{table}[t]
    \centering
    \caption{MVP-ORAM access request/response sizes in bytes, for workloads with $1$, $5$, and $10$ clients.}
    \begin{tabular}{|c|c|c|c|}
        \hline
        \textbf{No. clients} & $\mathsf{getPM}$ & $\mathsf{getPS}$ & $\mathsf{evict}$ \\
        \hline
        $1$ & $8/609$ & $8/272415$ & $307328/1$ \\
        \hline 
        $5$ & $8/1489$ & $8/445932$ & $314101/1$ \\
        \hline
        $10$ & $8/2781$ & $8/940107$ & $339344/1$ \\
        \hline
    \end{tabular}
    \label{tab:request_size}
\end{table}

\textbf{A note on the Strong MVP-ORAM performance.}
The strong variant of MVP-ORAM discussed in \S\ref{sec:strong_mvp_oram} requires each ORAM access to perform $\sigma+1$ MVP-ORAM accesses.
This means the throughput observed for this variant would be MVP-ORAM throughput divided by $\sigma+1$.
For example, if one wants to offer perfect ORAM guarantees for up to $10$ clients, each access would require eleven MVP-ORAM accesses, which means a throughput of about the $\approx 390/11 = 35$ accesses/sec.
A similar degradation also affects latency.

\subsection{Performance with Different Configurations}

To better understand how the performance of MVP-ORAM varies in different configurations, we conducted additional experiments for $n=4$, considering various tree heights, bucket sizes, block sizes, and $\alpha$ values.


The first set of experiments considers different tree heights ($L$) and bucket sizes ($Z$).
Fig.~\ref{fig:height_bucket_performance} shows the throughput and latency for different configurations $(L,Z)$ of the database. 
Unsurprisingly, the results show that the performance gets worse when $(L,Z)$ increases, since the amount of 4KB-blocks contained in a tree path is $L\times Z$.

The second set of experiments considers different values of $\alpha$, the Zipfian distribution parameter used for selecting blocks to be accessed.
Recall that smaller values of $\alpha$ make block accesses approximate a uniform distribution.

Fig.~\ref{fig:zipf_performance} shows the throughput and latency for the experiments.
Although the effect of parameter $\alpha$, which represents how skewed the system workload is, primarily reflects the obliviousness guarantee of the system, it also affects performance.
This is due to the influence of skewness on the size of the stashes produced by clients (see \S\ref{sec:evaluation_stash_size}).
The performance results show that when used in applications that induce skewed workloads, MVP-ORAM not only provides better security guarantees but also exhibits better performance. 

\begin{figure}[!t]
    \centering
    \includegraphics[width=\columnwidth]{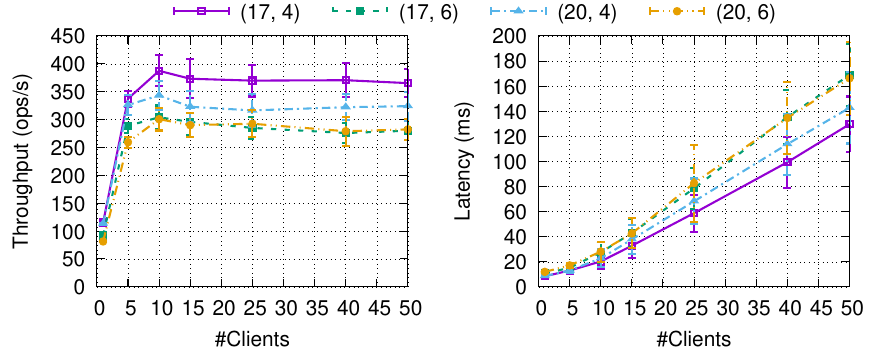}
    \caption{MVP-ORAM throughput and latency for different tree heights and bucket sizes. (17, 4) means tree height 17 and bucket size 4.}
    \label{fig:height_bucket_performance}
\end{figure}

\begin{figure}[!t]
    \centering
    \includegraphics[width=\columnwidth]{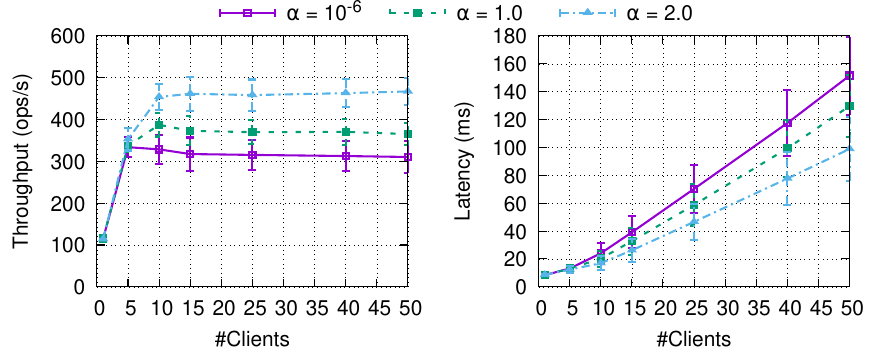}
    \caption{MVP-ORAM throughput and latency for different values of $\alpha$.}
    \label{fig:zipf_performance}
\end{figure}

Our final set of experiments evaluates the performance of MVP-ORAM with block sizes of 256 bytes, 1024 bytes, and 4096 bytes.
Fig.~\ref{fig:block_performance} shows the throughput and latency for the experiments. 

For small blocks of 256 bytes, the system reaches almost a thousand accesses per second.
As the block size increases, performance decreases accordingly, since $4\times$ and $16\times$ more data is transferred with the other block sizes.

\begin{figure}[!t]
    \centering
    \includegraphics[width=\columnwidth]{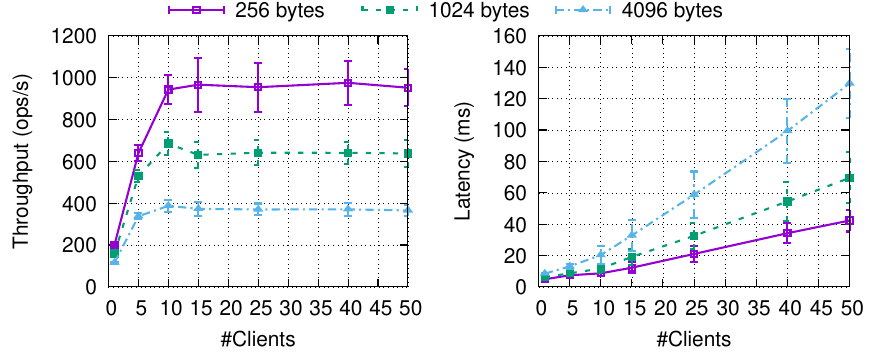}
    \caption{MVP-ORAM throughput and latency for different block sizes.}
    \label{fig:block_performance}
\end{figure}

Overall, these experiments confirm our theoretical observation that MVP-ORAM performance gets worse in configurations that require more data to be transferred. 
This aligns with observations made for other ORAM protocols and supports the fundamental goal of enhancing the overall bandwidth of ORAM schemes.

\subsection{Experimental comparison with other systems}
\label{sec:comparisonquoram}
Table~\ref{tab:mvp_oram_vs_quorum_performance} compares MVP-ORAM performance with COBRA and QuORAM.
COBRA~\cite{vassantlal_2022} tolerates Byzantine faults and ensures Secrecy, but not Obliviousness.
Since it accesses data by invoking a single request that leverages secret sharing, its throughput and latency are approximately an order of magnitude better than what was observed for MVP-ORAM.
This illustrates the cost of adding Obliviousness to a BFT datastore without resorting to trusted components.

QuORAM~\cite{maiyya2022quoram} is the only replicated ORAM service we are aware of.
It tolerates crash faults and uses trusted proxies, while MVP-ORAM tolerates Byzantine faults without requiring trusted components.
The existence of trusted proxies collocated with servers in the same machine eliminates the need to execute bandwidth-hungry ORAM operations through the network, as they are executed only between the proxy (which acts as a single Path ORAM client) and the server.
We confirmed the benefit QuORAM gained with this approach by executing it in the same setting as MVP-ORAM.
With over $100$ clients and for $n=4$ and $n=7$, QuORAM achieves a maximum throughput of around $1000$ operations per second.
However, our evaluation also shows that QuORAM performs much worse than MVP-ORAM with a restricted number of clients.
As sumarized in Table~\ref{tab:mvp_oram_vs_quorum_performance}, MVP-ORAM with $50$ client can process $356$ ($n=4$) and $355$ ($n=7$) ops/s, while QuORAM only processes between $163$ and $183$ ops/s.
The reason behind the low performance of QuORAM is its higher latency, which is caused by the use of proxies.
In our experiments, MVP-ORAM's maximum latency is $130$ ms, while QuORAM's minimum latency is $272$ ms.

\begin{table}[!t]
    \centering
    \caption{Performance comparison of MVP-ORAM with COBRA and QuORAM with 50 clients.}
    \begin{tabular}{l|c|c|c|c|}
        \cline{2-5}
        & \multicolumn{2}{c|}{$n=4$} & \multicolumn{2}{c|}{$n=7$} \\
        \hline
        \multicolumn{1}{|l|}{\textbf{Protocol}} & \textbf{Throughput} & \textbf{Latency} & \textbf{Throughput} & \textbf{Latency}\\
        \hline
        \multicolumn{1}{|l|}{COBRA} & \multicolumn{1}{r|}{3767 ops/s} & \multicolumn{1}{r|}{12 ms} & \multicolumn{1}{r|}{3446 ops/s} & \multicolumn{1}{r|}{13 ms} \\
        \multicolumn{1}{|l|}{MVP-ORAM} & \multicolumn{1}{r|}{356 ops/s} & \multicolumn{1}{r|}{130 ms} & \multicolumn{1}{r|}{355 ops/s} & \multicolumn{1}{r|}{128 ms} \\
        \multicolumn{1}{|l|}{QuORAM} & \multicolumn{1}{r|}{183 ops/s} & \multicolumn{1}{r|}{272 ms} & \multicolumn{1}{r|}{163 ops/s} & \multicolumn{1}{r|}{305 ms}\\
        \hline
    \end{tabular}
    \label{tab:mvp_oram_vs_quorum_performance}
    \vspace*{-4mm}
\end{table}

Notice that having too many clients increases the service attack surface, as ORAM clients must be mutually trusted.
Therefore, we argue that achieving high performance with fewer clients is more important than achieving good numbers with $100$ or more concurrent clients accessing the ORAM.


%% file: content/sec_conclusions.tex
\section{Conclusions and Future Work}

This paper presented MVP-ORAM, the first Byzantine fault-tolerant ORAM protocol.
It enables fail-prone concurrent clients to access a shared data store without revealing any information about the accessed data or their access patterns.
We show that in asynchronous networks, satisfying wait-freedom fundamentally compromises collision-freedom, affecting the security guarantees of our construction.
To account for this, we propose a weaker security definition for asynchronous wait-free ORAMs, which may be secure enough for typical storage applications with skewed block accesses.
We devise MVP-ORAM as a deterministic wait-free ORAM service and integrate it into a confidential BFT data store, which shows promising performance results.
Additionally, we introduce a stronger variant of MVP-ORAM that ensures perfect access-pattern secrecy with additional assumptions. 

This paper opens many avenues for future work.
For example, it seems impossible to implement a perfect wait-free ORAM in asynchronous systems, but this remains to be proved. 
Furthermore, more efficient variants of MVP-ORAM can be devised to implement perfect obliviousness.
Finally, state-of-the-art information dispersal~\cite{OptimalAVID} can potentially be used to improve the bandwidth requirements and concrete performance of replicated/BFT ORAM.

%% file: content/ap_mpv_oram_auxiliary_functions.tex
\ifbool{extendedVersion}{
}{
\clearpage
}
\section{MVP-ORAM Auxiliary Functions}
\label{ap:mvp_oram_auxiliary_functions}

Algorithm~\ref{alg:wait_free_auxiliary_functions} specifies the auxiliary functions used by clients to update their position map ($\mathsf{consolidatePathMaps}$), merge multiple versions of path and stashes ($\mathsf{mergePathStashes}$), and to create a new version of path and stash ($\mathsf{populatePath}$), as used in Algorithm~\ref{alg:wait_free_client} (\S\ref{sec:mvp-oram}).

The $\underline{\mathit{consolidatePathMaps}}$ function (A\ref{alg:wait_free_auxiliary_functions}, L1-7) receives as input a history of path maps containing the location and version updates of blocks that have been evicted until the access that invoked this function.
Using this history, it updates the local position map by keeping for each block address the location update with the highest timestamp, while discarding location updates with older versions.

The $\underline{\mathit{mergePathStashes}}$ function (A\ref{alg:wait_free_auxiliary_functions}, L8-11) receives as input a multi-version path, a set of stashes, and a consolidated position map. Using the position map as a reference point, the function filters blocks received in the path and stashes by keeping the ones with received in the correct slot and with timestamp according to the position map. Older and duplicated blocks are ignored.

Finally, the $\underline{\mathit{populatePath}}$ function (A\ref{alg:wait_free_auxiliary_functions}, L12-39) receives a working set, a path identification, accessed address, consolidate position map, and the sequence number of the access as input and populates a new path and stash with blocks contained in the working set.
This is achieved in four steps on the $\underline{\mathit{populatePath}}$ auxiliary function.
First (A\ref{alg:wait_free_auxiliary_functions}, L15-20), $c_i$ populates the new path $\mathcal{P}^*_l$ by putting blocks from $W$ in their correct slots according to $\mathit{pm}$.
If multiple clients have evicted different blocks to the same slot during the previous concurrent accesses, then $c_i$ selects the block with the highest sequence among them and keeps others in $W$ (A\ref{alg:wait_free_auxiliary_functions}, L17).
Additionally, $c_i$ keeps track of non-empty slots in $S_\mathit{used}$.

In the second step (A\ref{alg:wait_free_auxiliary_functions}, L21-25), $c_i$ exchanges $Z$ blocks from the stash with $Z$ blocks from $\mathcal{P}^*_l$.
This is done by sampling $Z$ slots from $\mathcal{P}^*_l$, including the accessed block if it was not previously in the stash, ensuring that the maximum number of paths is available to retrieve this block in the next accesses.
Then, $c_i$ samples $Z$ blocks from $W$ uniformly at random, except the accessed block, and puts them in set $B_Z$.
After selecting $Z$ blocks and slots, $c_i$ iterates over selected slots.
For each slot $\mathit{sl}^*$, if it contains a real block, it is moved to $W$, and a block from $B_Z$ is moved to $\mathit{sl}^*$.
Note that the stash size can decrease if some of the selected slots are empty, as we remove blocks from the stash rather than substituting them.

During the third step (A\ref{alg:wait_free_auxiliary_functions}, L26-31), $c_i$ reorders blocks by placing recently accessed blocks higher in the path. 
It first collects all blocks from $\mathcal{P}^*_l$ into $B_l$.
Then, iterates over each slot in the path after the exchange (i.e., $S_\mathit{used} \cup S_Z$) \emph{from the lowest slot to the highest} and evicts block with highest access among $B_l$ to it.
Additionally, $c_i$ updates the timestamp of the evicted blocks and their location on the path map $M_l$.

Finally, during the fourth step (A\ref{alg:wait_free_auxiliary_functions}, L32-39), $c_i$ builds the new stash $S$ by adding the remaining blocks in $W$, updating the timestamp and location of newly added blocks to the stash, including the accessed block.
The function then returns $\mathit{P}^*_l$, $S$, and $M_l$.

\input{algorithms/alg_mvp_oram_auxiliary_functions}

%% file: algorithms/alg_mvp_oram_auxiliary_functions.tex
\begin{algorithm}[t!]
\SetKwProg{Fn}{Function}{}{}
\DontPrintSemicolon
\caption{MVP-ORAM auxiliary functions.}
\label{alg:wait_free_auxiliary_functions}
{\small
\Fn{\underline{consolidatePathMaps}($\mathcal{H}_\mathit{pathMaps}$)}{
    $\forall \mathit{addr} = 0..N : \mathit{pm}[\mathit{addr}] \gets \langle \bot, \langle -1, -1, -1 \rangle \rangle$\\
    \ForEach{$M_l \in \mathcal{H}_\mathit{pathMaps}$}{
        \ForEach{$\langle \mathit{addr}, \mathit{sl}, \mathit{ts} \rangle \in M_l$}{
            $\langle \_, \mathit{ts}^* \rangle \gets \mathit{pm}[\mathit{addr}]$\\
            \lIf{$\mathit{ts} > \mathit{ts}^*$}{
                $\mathit{pm}[\mathit{addr}] \gets \langle \mathit{sl}, \mathit{ts} \rangle$
            }
        }
    }
       
    \Return $\mathit{pm}$
}

\Fn{\underline{mergePathStashes}($\mathcal{P}_l, \mathcal{S}, \mathit{pm}$)}{
    $W_s \gets \{ \langle \mathit{a}, \mathit{s}, \mathit{ts} \rangle \in S : S \in \mathcal{S} \wedge \mathit{pm}[\mathit{a}]=\langle \bot, \mathit{ts} \rangle\}$\\
    $W_p \gets \{\langle \mathit{a}, d, \mathit{ts} \rangle \in \mathcal{P}_l(\mathit{sl}) : \mathit{sl} \in \mathcal{P}_l \wedge$ 
    $\mathit{pm}[\mathit{a}] = \langle \mathit{sl}, \mathit{ts}\rangle\}$
    
    \Return $W_s \cup W_p$
}

\Fn{\underline{populatePath}($W, l, \mathit{addr}, \mathit{pm}, \mathit{seq}$)}{
    $\mathcal{P}^*_l, S, M_l, S_\mathit{used} \gets \emptyset$\\
    $\langle \mathit{sl}, \_ \rangle \gets \mathit{pm}[\mathit{addr}]$\\
    \ForEach(\tcp*[f]{put blocks on path $l$}){$\mathit{sl}^* \in \mathcal{P}^*_l$}{
        $B_\mathit{sl} \gets \{\langle \mathit{addr}^*, s \rangle : \langle \mathit{addr}^*, \_, \langle \_, \_, s \rangle \rangle \in W \wedge $ $\hspace*{1mm}\mathit{pm}[\mathit{addr}^*] = \langle \mathit{sl}^*, \_ \rangle\}$\\
        $\langle \mathit{addr}', \_ \rangle \gets$ entry with highest $s$ from $B_\mathit{sl}$\\
        $\mathcal{P}^*_l(\mathit{sl}^*) \gets \{W[\mathit{addr}']\}$\\
        $W \gets W \setminus \{\mathcal{P}^*_l(\mathit{sl}^*)\}$\\
        $S_\mathit{used} \gets S_\mathit{used} \cup \{\mathit{sl}^*\}$\\
    }
    $S_Z \gets$ $Z$ random slots from $\mathcal{P}^*_l$ including $\mathit{sl}$ if $\mathit{sl} \neq \bot$\\
    $B_Z \gets$ $Z$ random blocks from $W \setminus \{ \langle\mathit{addr},$ $\hspace*{1mm}... \rangle\}$\\
    
    \ForEach(\tcp*[f]{exchange $Z$ blocks}){$ \mathit{sl}^* \in S_Z$}{
        \lIf{$\mathcal{P}^*_l(\mathit{sl}^*) \neq \bot$}{
            $W \gets W \cup \mathcal{P}^*_l(\mathit{sl}^*)$
        }
        $\mathcal{P}^*_l(\mathit{sl}^*) \gets$ set with a block from $B_Z$\\
    }
    $B_l \gets \{b \in \mathcal{P}^*_l(\mathit{sl}^*) : \mathit{sl}^* \in \mathcal{P}^*_l \wedge \mathcal{P}^*_l(\mathit{sl}^*) \neq \bot\}$\\
    \ForEach(\tcp*[f]{reorder blocks}){$\mathit{sl}^* \in$ \textbf{sort}$(S_\mathit{used} \cup S_Z)$}{
        $\langle \mathit{addr}^*, d, \langle v, a, s \rangle \rangle \gets$ block with highest $a$ from $B_l$\\
        $B_l \gets B_l \setminus \{\langle \mathit{addr}^*, d, \langle v, a, s \rangle \rangle\}$\\
        $\mathcal{P}^*_l(\mathit{sl}^*) \gets \{\langle \mathit{addr}^*, d, \langle v, a, \mathit{seq} \rangle \rangle\}$\\
        $M_l \gets M_l \cup \{\langle \mathit{addr}^*, \mathit{sl}^*, \langle v, a, \mathit{seq} \rangle \rangle\}$\\
    }

    \ForEach(\tcp*[f]{build stash}){$\langle \mathit{addr}^*, b, \langle v, a, s \rangle\rangle \in W$}{
        $\langle \mathit{sl}^*, \_ \rangle = \mathit{pm}[\mathit{addr}^*]$\\
        \If{$\mathit{sl}^* \neq \bot \vee \mathit{addr}^* = \mathit{addr}$}{
            $\mathit{ts} \gets \langle v, a, \mathit{seq} \rangle$\\
            $M_l \gets M_l \cup \{\langle \mathit{addr}^*, \bot, \mathit{ts} \rangle\}$
        }
        \lElse {
            $\mathit{ts} \gets \langle v, a, s \rangle$
        }
        $S \gets S \cup \{\langle \mathit{addr}^*, b, \mathit{ts} \rangle\}$
    }
    
    \Return $\langle \mathcal{P}^*_l, S, M_l \rangle$
}
}
\end{algorithm}

%% file: content/ap_correctness.tex
\section{Correctness Proofs}
\label{ap:sec:correctness_proofs}

In this appendix we prove that MVP-ORAM fulfills the \textit{Correctness} property of Asynchronous Wait-Free ORAM, as per Definition~\ref{def:opram}.

\subsection{Preliminary Definitions}

We start with some preliminary definitions and then proceed to the protocol' Safety and Liveness proofs. 

\textbf{ORAM access operations.}
MVP-ORAM requires every operation to be executed in the same order on every server.
When replicated, this order is defined by the underlying consensus (or total order broadcast) protocol implemented by the BFT SMR.
Due to this, every operation invoked in the ``logical'' MVP-ORAM server can be totally ordered based on its execution order.

The high-level memory operations (read and write) are implemented through the $\mathsf{access}$ functionality (Algorithm~\ref{alg:wait_free_client}), which requires the invocation of the BFT SMR operations specified in Algorithm~\ref{alg:wait_free_server}.
In this way, given a client $i$, we define a memory \emph{write} of $v$ to block $b$ by $i$ as $\mathsf{access}(i,\mathit{write}, b, v)$ and a memory \emph{read} of block $b$ by $i$ as 
$\mathsf{access}(i,\mathit{read}, b, \bot)$.

We also define the \emph{sequence number of an access operation} $\mathit{seq}$, returned in the access' $\mathsf{getPM}(i)$, and the \emph{timestamp of a stored block} $\mathit{ts} = \langle v, a, s\rangle$ as the sequence number of the last access operation that updated ($v$), accessed ($a$), and touched ($s$) the block.
Block timestamps are compared using the following rule:

$\langle v,a,s \rangle > \langle v',a',s' \rangle \implies (v > v') \lor (v=v' \land a>a') \lor$\\
\hspace*{3.92cm}$(v=v' \land a=a' \land s>s')$

\textbf{Operation histories.}
The interaction between a set of clients and our replicated state machine is modeled by a sequence of operation invocation and reply events called \emph{history}~\cite{herlihy_1990}.
For an operation $o$, we denote by $inv(o)$ its invocation event and by $rep(o)$ its corresponding reply event.
A history $\mathcal{H} = (H, <_H)$ comprises the set of events $H$ and a total order relation among these events $<_H$. 
The order of events in a history $\mathcal{H}$ defines a \emph{partial order of operations} $\rightarrow_H$.
We say an operation $o_1$ \emph{precedes} (resp. \emph{succeeds}) an operation $o_2$, denoted by $o_1 \rightarrow_H o_2$ (resp. $o_2 \rightarrow_H o_1$), if $rep(o_1) <_H inv(o_2)$ (resp. $rep(o_2) <_H inv(o_1)$).
If an operation neither precedes nor succeeds another, we say the two operations are \emph{concurrent}.

Given a history $\mathcal{H}$, $\mathcal{H}|i$ is the sub-sequence of $\mathcal{H}$ made up of all events generated by client $c_i$.
Two histories $\mathcal{H}$ and $\mathcal{H}'$ are said to be \emph{equivalent} if they have the same local histories, i.e., for each client $c_i$, $\mathcal{H}|i = \mathcal{H}'|i$. 
In other words, both histories are built from the same set of events.

A history $\mathcal{H}$ is \emph{sequential} if its first event is an invocation, and then (1) each invocation event is immediately followed by its matching reply event, and (2) an invocation event immediately follows each reply event until the execution terminates.
Due to these properties, we can list $o$ in the history instead of $inv(o), rep(o)$.
If $\mathcal{H}$ is a sequential history, it has no overlapping operations, and consequently, the order of its operations is a \emph{total order}. 
A history that is not sequential is \emph{concurrent}.

We say that a sequential history $\mathcal{H}$ of (high-level, see below) ORAM operations is \emph{legal} if, for each (ORAM) memory position/ block $b$, the sequence $\mathcal{H}|b$ is such that each of its read operations returns the value written by the closest preceding write in $\mathcal{H}|b$ (or $\bot$, if there is no preceding write).

\begin{definition}[Linerizable history~\cite{herlihy_1990}]\label{def:linearizability}
A register history $\mathcal{H}$ is \emph{linearizable} (or atomic) if there is a ``witness'' history $\mathcal{L}$ such that:

\begin{enumerate}
\item $\mathcal{H}$ and $\mathcal{L}$ are equivalent;
\item $\mathcal{L}$ is sequential and legal;
\item $\rightarrow_H \subseteqq \rightarrow_L$.
\end{enumerate}
\end{definition}

Intuitively, these three properties require $\mathcal{H}$ and $\mathcal{L}$ to have the same events (1), the sequence of events in $\mathcal{L}$ is of the form $inv(o_1), rep(o_1), inv(o_2), rep(o_2), ...$ and it respects the specification of a read/write memory (2), and $\mathcal{L}$ respects the occurrence of order of the operations as defined in $\rightarrow_H$ (3).

When all histories produced by clients interacting with an object (a memory block, or register, in distributed computing parlance) are linearizable, we say the \emph{object satisfies linearizability}.

\textbf{High- and low-level histories.}
The interaction between clients and the ``logical'' MVP-ORAM server can be modeled on two levels.
At a high level, clients execute reads and writes at the ORAM by invoking $\mathsf{access}$.
This \emph{high-level history} $H_h$ is a sequence of events $inv(\mathsf{access}_1), ..., rep(\mathsf{access}_1), ...$
At a lower level, these accesses correspond to the three operations supported by the ``logical'' MVP-ORAM server.
Hence, the \emph{low-level history} $H_l$ is a sequence of events $inv(\mathsf{getPM}_1), ..., rep(\mathsf{getPM}_1),$ ..., $inv(\mathsf{getPS}_1), ..., $ $rep(\mathsf{getPS}_1), ...,$  $inv(\mathsf{evict}_1), ..., rep(\mathsf{evict}_1), ...$
Naturally, there is a correspondence between $H_h$ and $H_l$, as will be described next.

\subsection{Safety: Preservation of Memory State}

The following results prove that the transformation (i.e., merge, update and reshuffle of blocks) of position maps, trees, and stashes by clients preserves the most recent version of every accessed ORAM block.


\begin{lemma}{1}\label{lem:pm_merge}
Let $\mathcal{H}_\mathit{pathMaps} = \{M_1, ..., M_k\}$ be a history of path maps up to access $k$. 
The $\mathsf{consolidatePathMaps}$ function consolidates all path maps in $\mathcal{H}_\mathit{pathMaps}$ into a single position map $\mathit{pm}$ such that $\forall b \in \{0,...,N-1\}: \mathit{pm}[b] = \langle \mathit{slot}, \mathit{ts} \rangle \wedge (\exists M^* \in \mathcal{H}_\mathit{pathMaps}$ s.t. $\langle b, \mathit{slot}, \mathit{ts} \rangle \in M^* \wedge (\nexists M' \in \mathcal{H}_\mathit{pathMaps}$ s.t. $\langle b, \_, \mathit{ts}' \rangle \in M' \wedge \mathit{ts}' > \mathit{ts}))$. 
\end{lemma}
\begin{proof}
The proof follows from the direct observation of the code in Lines 1-7 of Algorithm~\ref{alg:wait_free_auxiliary_functions}. 
For each updated block address (Line 4), the stored slot of a block is defined by the highest timestamp associated with the block, which is returned by the function.
\end{proof}

\begin{lemma}{2}\label{lem:ps_merge}
Let $\mathcal{P}_l$, $\mathcal{S}$, and $\mathit{pm}$ be path $l$, a set of stashes and the consolidated position map, respectively, within the same $\mathsf{access}$ operation.
The $\mathsf{mergePathStashes}$ function returns a working set $W$ containing blocks from path $\mathcal{P}_l$ and stash set $\mathcal{S}$ whose timestamps and locations match the ones in the consolidated position map.
\end{lemma}
\begin{proof}
Paths and stashes contain a set of blocks, each associated with a logical timestamp defined when the block content was modified (A\ref{alg:wait_free_client}, L12).
The inspection of the code of $\mathit{mergePathStashes}$ in Lines 8-11 of Algorithm~\ref{alg:wait_free_auxiliary_functions} shows that only blocks in stashes and associated with timestamps and slots consistent with $\mathit{pm}$ are added to $W$.
\end{proof}

\begin{lemma}{3}[State Preservation]\label{the:statepreservation}
With the exception of the block $b$ accessed during a write, an execution of $\mathsf{access}$ preserves the state of the ORAM.
\end{lemma}
\begin{proof}
By the state of ORAM, we mean the most recent version (highest timestamp) of a block will be maintained in the server state.
Lemma~\ref{lem:pm_merge} shows the client consolidates all path maps, maintaining the entries with higher timestamps for each block.
Lemma~\ref{lem:ps_merge} shows the client adds the most recent timestamp of each block to its working set $W$.
The path and stash sent back to the server in $\mathsf{evict}$ are created on the $\mathit{populatePath}$ function in Lines 12-39 of Algorithm~\ref{alg:wait_free_auxiliary_functions}.
This function returns blocks to their corresponding slots in path $l$ while keeping other conflicting blocks (that would be in the same slots) in $W$ (A\ref{alg:wait_free_auxiliary_functions}, L15-20). 
It also swaps $Z$ blocks from $W$ with blocks in $Z$ slots of path $l$ (A\ref{alg:wait_free_auxiliary_functions}, L23-25).
Ultimately, it returns the populated path and a new stash containing the remaining blocks.
The server receives the new path map, stash and path, which are included in new $\mathit{oramState}$ by updating the multi-version tree, the set of stashes, and 
$\mathcal{H}_\mathit{pathMaps}$ (A\ref{alg:wait_free_server}, L15-19).
\end{proof}

\subsection{Safety: Memory Linearizability}

Knowing that every client access preserves the state of the ORAM, with the exception of the block being written (if it is a write), we are ready to prove that MVP-ORAM is safe under concurrent accesses, i.e., that it satisfies linearizability.
We start by showing any low-level history can be transformed into an equivalent sequential history. 
To do that, we first need to define the sequence number of an MVP-ORAM server operation.

\begin{definition}[Operation sequence number]
Let the sequence number of an operation $o$ executed on MVP-ORAM, denoted by $sn(o)$, be the order in which this operation was executed in the system, as assigned by the BFT SMR.
\end{definition}

\begin{lemma}{4}\label{lem:seqlowlevel}
A low-level history $\mathcal{H}_l$ generated by MVP-ORAM can be used to generate a history $\mathcal{S}_l$ that (1) is equivalent to $\mathcal{H}_l$ and (2) sequential.
\end{lemma}
\begin{proof}
$\mathcal{S}_l$ is constructed in the following way.
First, every operation invocation in $\mathcal{H}_l$ must appear in $\mathcal{S}_l$ in the order of their sequence number, i.e., for every pair of invocations $inv(o_1), inv(o_2)$, if $sn(o_1) < sn(o_2)$, then $inv(o_1) <_{\mathcal{S}_l} inv(o_2)$.
Second, every operation reply in $\mathcal{H}_l$ must appear in $\mathcal{S}_l$ exactly after its invocation.
It is easy to see that $\mathcal{S}_l$ is equivalent to $\mathcal{H}_l$ (they contain the same events) and sequential (all replies directly follow their invocations).
\end{proof}

We show now that any low-level history $\mathcal{H}_l$ can be mapped to a high-level history $\mathcal{H}_h$.

\begin{lemma}{5}\label{lem:lowhighlevel}
A low-level history $\mathcal{H}_l$ generated by MVP-ORAM can be used to generate a high-level history $\mathcal{H}_h$ in which every $\mathsf{access}$ in $\mathcal{H}_h$ is mapped to exactly one $\mathsf{getPM}$, one $\mathsf{getPS}$, and one $\mathsf{evict}$ in $\mathcal{H}_l$.
\end{lemma}
\begin{proof}
Given a low-level history $\mathcal{H}_l$, we use the result of Lemma \ref{lem:seqlowlevel} to generate an equivalent sequential history $\mathcal{S}_l$.
$\mathcal{H}_h$ is generated in the following way.
First, for every client $c_i$, we generate $\mathcal{S}_l|i =$ $\mathsf{getPM}(i)_1$, $\mathsf{getPS}(i,...)_1$, $\mathsf{evict}(i,...)_1$, $\mathsf{getPM}(i)_2$, ...
Then, for every sub-sequence $\mathsf{getPM}(i)_s$, $\mathsf{getPS}(i,...)_s$, $\mathsf{evict}(i,...)_s$, we have one $\mathsf{access}(i,...)_s$ in $\mathcal{H}_h$.
The invocations and replies of access operation in $\mathcal{H}_h$ are done as follows.
For every pair of clients $c_i, c_j$, and any sub-sequences $s \in \mathcal{S}_l|i$ and $q \in \mathcal{S}_l|j$, $inv(\mathsf{access}(i,...)_s) <_{\mathcal{H}_h} inv(\mathsf{access}(j,...)_q)$ if and only if $\mathsf{getPM}(i)_s <_{\mathcal{S}_l} \mathsf{getPM}(j)_q$ and $rep(\mathsf{access}(i,...)_s) <_{\mathcal{H}_h} rep(\mathsf{access}(j,...)_q)$ if and only if $\mathsf{evict}(i,...)_s$ $<_{\mathcal{S}_l} \mathsf{evict}(j,...)_q$.
\end{proof}

Having a high-level history containing the $\mathsf{access}$ operations generated from a low-level history of executed server operations, we are ready to prove this history satisfies the closest preceding write rule.

\begin{definition}[Closest preceding write]
The write access $w = \mathsf{access}(i,write,b,x)$ is the \emph{closest preceding write} of read access $r = \mathsf{access}(j,read,b,\bot)$ if and only if $sn(\mathsf{evict}_w) < sn(\mathsf{getPM}_r)$ and there is no other write access $w' = \mathsf{access}(k,write,b,x')$ such that $sn(\mathsf{getPM}_w) < sn(\mathsf{getPM}_{w'})$ and $sn(\mathsf{evict}_{w'}) < sn(\mathsf{getPM}_r)$.
\end{definition}

\begin{lemma}{6}\label{lem:writeread}
If $w = \mathsf{access}(\_,write,b,x)$ is the closest preceding write of read $r = \mathsf{access}(\_,read,$ $b,\bot)$, then $r$ returns $x$.
\end{lemma}
\begin{proof}
If $w$ is the closest preceding write of $r$, then $\mathsf{getPM}_r$ returns a path map $M_l \in \mathcal{H}_\mathit{pathMaps}$ in which $\langle b, \_, \langle v, \_, \_ \rangle \rangle \in M_l$ (A\ref{alg:wait_free_client}, L2), where $v$ was defined during the execution of $w$ (A\ref{alg:wait_free_client}, L12), as no other write affected block $b$ between $\mathsf{evict}_w$ and $\mathsf{getPM}_r$.
Consequently, the update $\langle b, \mathit{sl}, \langle v, \_, \_ \rangle \rangle$ will be preserved in consolidated position map $\mathit{pm}$ by $\mathsf{consolidatePathMaps}_r$ (Lemma~\ref{lem:pm_merge}).
In the remaining execution of $r$, $\mathsf{getPS}_r$ returns path $l$ that passes through $\mathit{sl}$ and the corresponding stash $s$.
All blocks will be added to a working set $W$ (A\ref{alg:wait_free_client}, L6-7), which contains blocks with timestamps matching the info in $\mathit{pm}$ (Lemma~\ref{lem:ps_merge}), and it will contain block $b$ updated in $w$.
$r$ retrieves $b$ from $W$ and returns its data at the end of the access (A\ref{alg:wait_free_client}, L11,16).
\end{proof}

\begin{theorem}{2}[Linearizability]\label{the:linearizability}
For each memory position $b$, MVP-ORAM's read ($\mathsf{access}(\_,read,b,\bot)$) and write ($\mathsf{access}(\_,write,b,\_)$) operations satisfy linearizability.
\end{theorem}
\begin{proof}
Consider the low-level history $\mathcal{H}_l$ with all operations issued to the MVP-ORAM service.
We use the result of Lemma \ref{lem:lowhighlevel} to generate an equivalent high-level history $\mathcal{H}_h$.
We have to prove that, for each memory position $b$, $\mathcal{H}_h|b$ satisfies linearizability.

To do that, we need to build a ``witness'' history $\mathcal{L}_b$ from $\mathcal{H}_h|b$ satisfying Definition~\ref{def:linearizability}.
$\mathcal{L}_b$ is constructed in the following way:

\begin{enumerate}
\item Every write access $w$ in $\mathcal{H}_h|b$ is added to $\mathcal{L}_b$ satisfying their invocation order in $\mathcal{H}_h|b$, i.e., $\forall w, w'$, $w <_{\mathcal{L}_b} w'$ if and only if $inv(w) <_{\mathcal{H}_h|b} inv(w')$;
\item Every read access $r$ in $\mathcal{H}_h|b$ is added to $\mathcal{L}_b$ between two consecutive writes $w$ and $w'$ if and only if $r$ started after $w$ and before $w'$, i.e., $w <_{\mathcal{L}_b} r <_{\mathcal{L}_b} w'$ if and only if $rep(w) <_{\mathcal{H}_h|b} inv(r)  <_{\mathcal{H}_h|b} inv(w')$;
\item Every set of reads between two consecutive writes are added to $\mathcal{L}_b$ in accordance to their invocation order in $\mathcal{H}_h|b$, i.e., $\forall w, w'$ and $\forall r, r'$ s.t. $w <_{\mathcal{L}_b} r, r'  <_{\mathcal{L}_b} w'$, $r <_{\mathcal{L}_b} r'$ if and only if $inv(r) <_{\mathcal{H}_h|b} inv(r')$.
\end{enumerate}

For each client $i$, $\mathcal{L}_b|i = \mathcal{H}_h|b|i$ since (1) the three rules above ensure all events of $\mathcal{H}_h|b$ are included in $\mathcal{L}_b$, and (2) a client makes its accesses sequentially and thus accordingly to $\mathcal{L}_b|i$, which respects invocation order. 
Therefore $\mathcal{L}_b$ and $\mathcal{H}_h|b$ are \emph{equivalent}.

The resulting history $\mathcal{L}_b$ is also \emph{sequential} by construction.
Further, the history is \emph{legal} since each read $r \in \mathcal{L}_b$ returns the value written by the closest preceding write $w \in \mathcal{L}_b$ according to Lemma \ref{lem:writeread}.

Finally, the partial order of events in $\mathcal{H}_h|b$ is preserved in $\mathcal{L}_b$ since the three rules used to construct the later ensure that each $o,p$ in $\mathcal{H}_h|b$ such that $o \rightarrow_{H_h|b} p$, also appear in $\mathcal{L}_b$ respecting $o \rightarrow_{L_b} p$.
\end{proof}

\subsection{Liveness}

\begin{theorem}{3}[Wait-freedom]\label{the:waitfreedom}
Every invocation of MVP-ORAM's $\mathsf{access}$ by a correct client terminates. 
\end{theorem}

\begin{proof}
Operation $\mathsf{access}$ (Algorithm~\ref{alg:wait_free_client}) has no wait clauses, and all local functions terminate.
It remains to be shown that replicated server operations ($\mathsf{getPM}$, $\mathsf{getPS}$, and $\mathsf{evict}$) terminates.
This holds because (1) the underlying BFT SMR protocol ensures all invocations to the service are delivered to the servers (e.g., due to consensus liveness in our system model), (2) all correct servers receive the invocation and perform the operation, using only local operations (Algorithm~\ref{alg:wait_free_server}), and (3) the correct servers send the operation reply back to the client.
\end{proof}

%% file: content/ap_obliviousness.tex
\section{Obliviousness Proofs}
\label{ap:sec:obliviousness}

We now prove that MVP-ORAM fulfills the \textit{Obliviousness} property of Asynchronous Wait-Free ORAM, as per Definition~\ref{def:opram}.

The access pattern of MVP-ORAM depends on the number of concurrent clients and the distribution of the accessed blocks, and it can be described by three scenarios:
(1) a single client accesses the ORAM once per timestep,
(2) multiple clients access different blocks per timestep, and
(3) multiple clients access the same block in the same timestep.
In practice, an access pattern will be a mixture of these scenarios, i.e.,  in a timestep, a single client accesses a block, while in another timestep, $d$ clients access distinct blocks and $c-d$ clients access the same block.
However, we analyze those scenarios separately to show the best- and worst-case security of MVP-ORAM.
The first two scenarios represent the best-case scenarios as clients access different blocks in each timestep, which generates random leaves. However, the generated leaves in the third scenario are correlated, which is why that scenario represents the worst-case.

Before we start, we recall the three main features of MVP-ORAM.
(F1) The position map stores the mapping of each block's logical address to its current \emph{slot} in the tree, and clients access that block by selecting uniformly at random a path that passes through that slot.
(F2) In each access, each client selects $Z$ slots from the accessed paths (including the slot of the accessed block if the block was in the path) and moves non-dummy blocks to \emph{stash}. Additionally, it evicts up to $Z$ blocks from the stash to the previously selected slots.
(F3) After exchanging the blocks between the path and the stash, the client rearranges blocks in the path based on their popularity, with the most popular blocks placed higher in the path to maximize the number of available paths for future accesses.
We start by analyzing the first scenario.

\begin{lemma}{7}\label{thm:sec:cases:1:2}
When a single client accesses the ORAM once per timestep, the access pattern $A(\overrightarrow{y})$ observed by the server during a sequence of requests $\overrightarrow{y}$ is computationally indistinguishable from a random sequence with high probability.
\end{lemma}
\begin{proof}
Let $\overrightarrow{y}=(b_1,b_2,\dots,b_M)$ represent a sequence of $M$ blocks requested by the client, with one request made per timestep. 
Besides, let $S(\overrightarrow{y}) = (\mathit{slot}_1,\mathit{slot}_2,\dots,\mathit{slot}_M)$ denote the sequence of slot addresses, where $\mathit{slot}_i$ ($1 \leq i \leq M$) is the address of the slot storing block $b_i$.
If $b_i$ is in the stash, $\mathit{slot}_i$ is considered as one of the slots of the root node.
Note that servers see $A(\overrightarrow{y}) = (x_1, x_2, \dots, x_M)$, where $x_i$ ($1 \leq i \leq M$) is a random leaf node whose path passes through the slot $\mathit{slot}_i$ up to the root.
For any two accesses $i$ and $j$ ($1 \leq i < j\leq M$), we argue $x_i$ and $x_j$ are statistically independent under the following two possible cases:
\begin{itemize}
    \item If $b_i = b_j$. 
    Note that during the $i$th access, block~$b_i$ is added to the stash.
    We consider two possible sub-cases: 
    (1) This block is removed from the stash during access~$k$ ($i < k < j$), and (2) The block remains in the stash from the $(i+1)$th access to the $j$th access.
    In the first sub-case, notice that $b_i$ is removed from the stash and placed in the path accessed by $k$. This path can be the same that was requested in $i$, or another, but either way, as it is a recently accessed block, it is placed in the root due to F3.
    During the $j$th access, if this block is still in the same node, it can be accessed through multiple paths due to F1, making $x_i$ independent of $x_j$. 
    Otherwise, if the block is moved to another path of the tree during $k$ or another access, $x_i$ remains independent of $x_j$.
    In the second sub-case, as the block is in the stash, $x_j$ can be any leaf; hence, $x_i$ is independent of $x_j$.
    \item If $b_i \neq b_j$.
    Note that the only scenario in which $x_i$ could be connected to $x_j$ in this case is as follows: 
    during the $i$th access, $b_i$ is moved to the stash, $b_j$ is already in the stash, and $b_j$ is subsequently moved to the path $\mathcal{P}_{x_i}$.
    Although this connection exists, due to F3, $b_j$ is placed close to the root, allowing $x_j$ to be almost any leaf. 
    Consequently, $x_i$ remains independent of $x_j$.
\end{itemize} 
\end{proof}

The second scenario is similar to the first one, as the accessed blocks per timestep differ.
Note that the position map changes after evictions. When $c$ clients concurrently start an access operation, they access blocks from the same locations.
Accordingly, if they access the same block, they generate an access pattern containing the leaves from a sub-tree rooted at the node where the block is stored. 
If the block is at the root of the tree or stash, they can access the block by accessing a path identified by any leaf.
This implies that the security depends on the distribution of the access pattern of $c$ concurrent clients.

We now analyze the third scenario using statistical distance~\cite{reyzin2011}.
We first present an insight into how this distance is computed.  
Consider a tree of sufficiently large height (e.g., $L = 17$).
In a random sequence of size $c$, we expect to observe $c$ distinct leaves being accessed.
However, when $c$ clients simultaneously access the same block \---- particularly if the block is located near the leaves \---- the expected number of distinct leaves involved is typically less than~$c$. 
This difference forms the basis of our comparison.
Specifically, we compare \---- using statistical distance \---- the distribution of the number of distinct leaves in a random sequence with the distribution of the number of distinct leaves generated in the worst-case execution of MVP-ORAM.
To formalize this comparison, let $X$ and $Y$ be random variables denoting the number of distinct leaves accessed in the random sequence and in the worst-case execution of MVP-ORAM, respectively, during sequences of size~$c$.
The statistical distance $\Delta(X, Y)$ between the two distributions is defined as:
\begin{align}\label{eq:statistical:distance}
    \Delta(X,Y) = \dfrac{1}{2}\sum_{k\in \{1,\dots,c\}} \left|\Pr(X=k) - \Pr(Y=k)\right|
\end{align}

In our analysis, we consider that clients access blocks following a Zipfian distribution with parameter~$\alpha$. As proved in Lemma~\ref{thm:mvp-follows-zipfian}, MVP-ORAM's algorithm closely models this distribution.
This allows us to characterize and study the security aspect of MVP-ORAM under different access patterns with a single parameter; for example, by choosing a low value of~$\alpha$, we simulate scenarios where clients access blocks uniformly at random, while higher values of~$\alpha$ capture skewed access patterns in which certain blocks are accessed significantly more frequently than others.
\begin{lemma}{8}
    \label{thm:mvp-follows-zipfian}
    The distribution of blocks in the binary tree after execution of the MVP-ORAM algorithm follows a Zipfian distribution with parameter $\alpha$, where frequently accessed blocks are near the root while less frequently accessed blocks are near leaves.
\end{lemma}
\begin{proof}
    A Zipfian distribution means that the frequency $f(r)$ of accessing the $r^\mathit{th}$ most frequently accessed block (rank $r$) decreases proportionally to $r^{-\alpha}$.
    The rank in MVP-ORAM is defined by the access sequence number.
    We show that the MVP-ORAM algorithm distributes frequently accessed blocks near the top of the tree and less frequently accessed blocks near the bottom due to the following two features of the algorithm.
    First, each accessed block is stored in the stash. While a block is being accessed, it will remain in the stash until no client accesses it. Then, the second feature, clients evict the blocks to a path and reorder them according to the sequence number. This ordering ensures that recently accessed blocks remain near the top of the tree while less frequently accessed blocks are evicted deeper in the tree.
    The combination of both features proves the lemma.
\end{proof}

Before presenting the main result, we present two preliminary lemmas.

\begin{lemma}{9}\label{lem:1:X}
    Given a tree with height $L$ and a random sequence of size $c$, let $X$ be the random variable representing the number of distinct accessed leaves.  
    The probability mass function of $X$ is given by:
    \begin{align*}
        \Pr(X=k \mid C=c) = P(2^L, k) \cdot S\left(c,k\right)/2^{L \cdot c},
    \end{align*}
    where $k \in \{1,\dots,c\}$, $P(2^L, k)$ denotes the number of permutations of $k$ distinct leaves from a total of $2^L$, and $S(c, k)$ is the Stirling number of the second kind.
\end{lemma}
\begin{proof}
    Since the tree's height is $L$, it has $2^L$ leaves.  
    Consider a random access sequence of size $c$, where each access independently selects a leaf uniformly at random.  
    Let $X$ be the random variable denoting the number of distinct leaves accessed in the sequence.  
    We aim to compute $\Pr(X = k \mid C = c)$ for $k \in \{1, \dots, c\}$.
    To do so, we count the number of sequences of $c$ accesses that involve exactly $k$ distinct leaves and divide by the total number of possible access sequences, which is $(2^L)^c$.
    
    First, we choose and order the $k$ distinct leaves to be accessed.  
    There are $P(2^L, k)$ ways to do this.
    Next, we assign the $c$ accesses to the $k$ leaves such that each leaf is accessed at least once.  
    This is equivalent to partitioning the $c$ accesses into $k$ non-empty subsets, which can be done in $S(c, k)$ ways.
    Multiplying these two quantities gives the number of favorable sequences with exactly $k$ distinct leaves.  
    Thus, the desired probability is:
    \begin{align*}
        \Pr(X=k \mid C=c) = P(2^L, k) \cdot S\left(c,k\right)/2^{L \cdot c}.
    \end{align*}
\end{proof}

\begin{lemma}{10}\label{lem:2:Y}
    Given a tree with height $L$ and a sequence of $c$ concurrent accesses to the same block selected following a Zipfian distribution with parameter $\alpha$, let $Y$ be the random variable representing the number of distinct accessed leaves.  
    The probability mass function of $Y$ is given by:
    \begin{align*}
        \Pr(Y = k \mid C=c)
        &= \sum_{d=0}^{L}   \frac{\sum_{j=2^d}^{2^{d+1}-1}j^{-\alpha}}{\sum_{j=1}^{N}j^{-\alpha}} \cdot \frac{P(2^{L-d}, k) \cdot S(c,k)}{2^{(L-d) \cdot c}},
    \end{align*}
    where $k \in \{1,\dots,c\}$.
\end{lemma}
\begin{proof}
    Since the tree's height is $L$, it has $2^L$ leaves.
    Let $Y$ be the random variable representing the number of distinct leaves accessed during the $c$ concurrent accesses to the same block.
    We compute the probability mass function of $Y$ by conditioning on the depth~$d$ of the selected block and summing over all possible levels $d = 0, \dots, L$.
    Let $d^*$ denote the level where the block is located. 
    We have:

    {\footnotesize
    \begin{align*}
        \Pr(Y = k \mid C = c)
        = \sum_{d=0}^{L}\Pr(d^* = d) \cdot \Pr(Y = k \mid C = c \land d^* = d)
    \end{align*}
    }
    
    The probability that a block resides at level $d$ is determined by the Zipfian distribution over block indices. 
    Hence, we have:
    \begin{align}\label{eq:d:star:d}
        \Pr(d^* = d)
        = \frac{\sum_{j=2^d}^{2^{d+1}-1}j^{-\alpha}}{\sum_{j=1}^{N}j^{-\alpha}}
    \end{align}
    Note that if the block resides at level~$d$ of the tree, then it can be accessed through one of $2^{L - d}$ possible leaves.
    Thus, all $c$ clients accessing the same block randomly choose leaves among those $2^{L - d}$.
    Similar to the proof of Lemma~\ref{lem:1:X}, we have:
    \begin{align}\label{eq:Y:C:d:star}
        \Pr(Y = k \mid C = c, d^* = d) = \frac{P(2^{L - d}, k) \cdot S(c, k)}{(2^{L - d})^c}
    \end{align}
    By substituting Expressions~\ref{eq:d:star:d} and~\ref{eq:Y:C:d:star} into the original summation, the lemma holds.
\end{proof}

Let $\mu(N, c, \mathcal{D}(\alpha))$ represent the value obtained by applying the distributions from Lemmas~\ref{lem:1:X} and~\ref{lem:2:Y} to Expression~\eqref{eq:statistical:distance}, with $\mathcal{D}(\alpha)$ representing the distribution of accesses that results from zipfian parameter $\alpha$.

\begin{theorem}{4}\label{thm:oram-same-block-not-negligible}
    Given $c, N \in \mathbb{N}$, $\alpha \in \mathbb{R}$, and $D\in\mathcal{U}$, the statistical distance between a random sequence of size~$c$ and the access pattern generated by MVP-ORAM is bounded by $\mu(N,c,D(\alpha))$.
\end{theorem}
\begin{proof}
    The theorem follows directly from applying Lemmas~\ref{lem:1:X} and~\ref{lem:2:Y} to Expression~\eqref{eq:statistical:distance}. 
\end{proof}

Theorem~\ref{thm:oram-same-block-not-negligible} proves that the statistical distance depends on the size of the tree and the distribution of requests, specifically it decreases as the size of the tree increases (i.e., the height of the tree), and as the number of frequently accessed blocks decreases (i.e., $\alpha$ increases).
If $\alpha$ is low, approximating uniform distribution, the majority of accessed blocks will be located in the last two depths of the tree. Hence, the number of available leaves for $c$ clients to access the block will be lower than accessing a block located in the root or stash, increasing statistical distance and degrading security.

%% file: content/ap_stash.tex
\section{Stash Size Analysis}
\label{ap:sec:stash_size_analysis}

In this Appendix, we prove there is an upper bound on the stash size of MVP-ORAM.
With this goal, we first define a scenario and argue that it results in the largest stash size. 
To formalize this scenario, we introduce the following notation:
Let $C(\tau)$ represent the number of clients whose access operations have not yet been completed by timestep $\tau \in \mathbb{N}$. 
We define $c$ as the maximum number of clients concurrently accessing ORAM, i.e., $c = \max(\{C(\tau) : \tau \in \mathbb{N}\})$.
%
\begin{lemma}{11}
\label{thm:worst:case}
    In MVP-ORAM, the largest stash size occurs when, at each timestep~$\tau$, $c$ clients execute access operations concurrently.
\end{lemma}
\begin{proof}
Note that if $C(\tau) = 1$ at any timestep~$\tau$, the stash size does not grow, as a client executing an access operation adds at most $Z$ blocks to the stash and evicts $\min(\{Z,\mathit{st}\})$ blocks from the stash to the accessed path, where $\mathit{st}$ is the current stash size.
However, when $x \geq 2$ clients execute access operations concurrently, the number of blocks evicted by those clients is not necessarily equal to $\min(\{Z \cdot x,\mathit{st}\})$ due to the following two reasons:
(1) multiple clients might select a common block to evict from the stash to their accessed paths, and
(2) multiple clients might select the same slot to place different blocks.
Indeed, when $x = c$, the overlap between the blocks chosen by concurrent clients and the overlap in selected slots increases.
This reduces the number of blocks that can be evicted from the stash to the accessed paths.
\end{proof}

The next step in analyzing the stash size is to compute the stash size under the worst-case scenario, i.e., when $c$ clients execute access operations concurrently at each timestep.
In this scenario, concurrent clients add approximately $c \cdot Z$ blocks to the stash in each timestep.
When the stash size is small, clients may be unable to remove $c \cdot Z$ blocks due to overlaps caused by multiple clients selecting the same blocks.
However, as the stash size grows to $O(c \cdot \log{N})$, the system reaches a point where concurrent clients can remove approximately $c \cdot Z$ blocks from the stash.
At this point, the stash size stabilizes because the rate of blocks being added to the stash approximately matches the rate of blocks being removed.
We formalize this result in the following theorem.

\begin{theorem}{5}\label{thm:stash:size}
    Under the worst-case scenario concerning concurrency, the expected stash size at any timestep is $O(c \cdot \log{N})$.
\end{theorem}
\begin{proof}
The stash size in the MVP-ORAM can be modeled as a stochastic process.
Let $\mathit{st}_{\tau}$ represent the stash size at timestep~$\tau \in \mathbb{N}$.
The size of the stash changes at timestep~$\tau$ according to the following three factors.
(1) At most $Z$ blocks are added to the stash by each client;
If $x$ concurrent clients access the tree, up to $x \cdot Z$ blocks may be added to the stash. 
(2) Blocks are removed from the stash when a client or multiple concurrent clients evict blocks from the stash back to the selected paths.
(3) Some blocks might be moved to the stash when rearranging the blocks.
This case happens when multiple clients place different blocks in the same slot.
Accordingly, the evolution of $\mathit{st}_{\tau}$ can be expressed as follows:
$
    \mathit{st}_{\tau} = \mathit{st}_{\tau-1} + X_{\tau}^{\mathit{in}_1} + X_{\tau}^{\mathit{in}_2} - X_{\tau}^{\mathit{out}}, 
$
where $X_{\tau}^{\mathit{in}_1}$ and  $X_{\tau}^{\mathit{in}_2}$ are the number of blocks added to the stash due to the first and third factors, and $X_{\tau}^{\mathit{out}}$ is the number of blocks removed from the stash, at time~$\tau$.
The process begins with an initial state of~$\mathit{st}_{0}=0$.
In the remainder of the proof, we demonstrate that when $\mathit{st}_{\tau} = O(c \cdot \log{N})$, $X_{\tau}^{\mathit{in}_1} + X_{\tau}^{\mathit{in}_2} \approx X_{\tau}^{\mathit{out}}$; hence the stash size stabilizes, as the rate of blocks being added to the stash approximately matches the rate of blocks being removed.

We first analyze $X_{\tau}^{\mathit{in}_1}$ and $X_{\tau}^{\mathit{in}_2}$.
Let $Y$ denote the number of distinct slots clients select at timestep~$\tau$.
Further, assume $Y'$ denotes the number of distinct empty slots among the selected slots.
Observe that $X_{\tau}^{\mathit{in}_1} = Y-Y'$.
We compute the expected value of $Y'$ under the worst-case scenario.
Note that the value of $Y'$ is minimized when the last level of the tree is empty (as blocks tend to be placed in the higher levels of the tree due to the rearrangement phase, this situation indeed occurs.)
Accordingly, the blocks are not randomly distributed in slots; as a result, when a client selects $Z$ slots from a path, almost all slots are filled.
In this situation, the expected number of empty slots selected by a client is $Z/(Z\cdot \log{N}) = 1/\log{N}$.
As there are $c$ clients, the expected number of empty slots equals $c/\log{N}$, i.e., $\mathbb{E}[Y']=c/\log{N}$. 

Among the blocks in the stash, at most $c$ were accessed during the previous timestep. If clients select these blocks for removal from the stash at timestep $\tau$, all of them will be moved to the root during the rearrangement phase. However, since the root's capacity is limited, all but one of these blocks will be returned to the stash.
We show that if $\mathit{st}_{\tau} = c \cdot Z \cdot \log{N}$ (i.e., $\mathit{st}_{\tau} = O(c \cdot \log{N})$), when clients select blocks to remove from the stash, the expected number of blocks selected from those most recently accessed blocks equals $c/\log{N}$.
With this aim, assume $b$ is one of those blocks.
The probability that a single client selects $b$ is $Z/(c \cdot Z \cdot \log{N})=1/(c \cdot \log{N})$.
Hence, the probability that no client selects $b$ is $(1-1/(c \cdot \log{N}))^c$.
Consequently, in expectation, $c \cdot (1-(1-1/(c \cdot \log{N}))^c) \approx c/\log{N}$ of the most recently accessed blocks are selected to be removed from the stash.
Hence, $\mathbb{E}[X_{\tau}^{\mathit{in}_2}]=c/\log{N}$.

We now analyze $X_{\tau}^{\mathit{out}}$.
Note that $X_{\tau}^{\mathit{out}}$ equals the number of distinct blocks that $c$ concurrent clients select from the stash.
We show that when $c \cdot Z \cdot \log{N}$ blocks are in the stash, the clients select almost $Z\cdot c$ distinct blocks from the stash.
With this aim, let $X$ be the number of distinct blocks selected by $c$ clients.
For each block $i$ in the stash, we define an indicator random variable~$I_i$, equal to~$1$ if at least one client selects $i$; otherwise, it is $0$.
Thus, $\mathbb{E}[X] = \sum_{i\in \mathit{stash}}\mathbb{E}[I_i]$.
Note that the probability of selecting~$i$ by a specific client is $Z/(c\cdot Z\cdot\log{N}) = 1/(c\cdot\log{N})$.
Hence, the probability that~$i$ is not selected by that specific client is $1-1/(c\cdot\log{N})$.
Since clients select blocks independently, the probability that no client selects~$i$ is $(1-1/(c\cdot\log{N}))^c$.
Thus, $\mathbb{E}[I_i] = 1 - (1-1/(c\cdot\log{N}))^c$.
Now we are ready to compute $\mathbb{E}[X]$. 
Since $1-x \approx \exp(-x)$, we have:
\begin{align*}
\mathbb{E}[X] 
&= \sum_{i\in \mathit{stash}}\mathbb{E}[I_i] 
\\&= c\cdot Z\cdot\log{N} \cdot ( 1 - (1-1/(c\cdot\log{N}))^c) 
\\&\approx c\cdot Z\cdot\log{N} \cdot (1-\exp(-1/\log{N})) 
\\&\approx c \cdot Z.
\end{align*}

As the number of distinct blocks clients select is approximately $c \cdot Z$, clients can fill almost all $Y$ distinct slots with different blocks; hence, $\mathbb{E}[X_{\tau}^{\mathit{out}}] \approx \mathbb{E}[Y]$.  
Finally, we have: 
$\mathbb{E}[X_{\tau}^{\mathit{in}_1}+X_{\tau}^{\mathit{in}_2}-X_{\tau}^{\mathit{out}}]
    \approx \mathbb{E}[Y-Y'] + c/\log{N} - \mathbb{E}[Y] 
    = \mathbb{E}[Y]-\mathbb{E}[Y'] + c/\log{N} - \mathbb{E}[Y]
    \approx 0$.  
\end{proof}

%% file: content/ap_strong_mvp_oram_algorithms.tex
\section{Strong MVP-ORAM}
\label{ap:strong_mvp_oram_algorithms}

\input{algorithms/alg_strong_mvp_oram_client}
\input{algorithms/alg_strong_mvp_oram_server}

The Strong MVP-ORAM protocol makes access pattern oblivious even in the worst-case scenario when ORAM accesses follow a uniform distribution by executing $\sigma \geq 0$ dummy accesses for each real access.
The proposed approach creates a scheduler of accesses containing $\sigma$ dummy accesses and a single real access, such that in each timestep, no two clients will access the same block if $c \leq \sigma + 1$.
To construct this scheduler, each client specifies the address it will access when invoking $\mathsf{getPM}$, and they utilize this information to create the scheduler.
Note that the address is encrypted; however, for simplicity, we omit the encryption details from the algorithm specification. The modified client and server specifications are presented in Algorithm~\ref{alg:strong_mvp_oram_client} and Algorithm~\ref{alg:strong_mvp_oram_server}, respectively.

The implementation of this strategy requires changing a few lines in the MPV-ORAM protocol. The first change is the addition of the list $\mathcal{A}$ in the server. For each client $c_i$, the server stores in $\mathcal{A}$ the encrypted address that $c_i$ will access and a counter that tracks the number of accesses executed by $c_i$ (A\ref{alg:strong_mvp_oram_server}, L7-10).
The address is used by $c_i$ to determine the timestep at which it will perform the real access, while the counter is used by the server to remove $c_i$'s access context on the $(\sigma+1)^\mathit{th}$ access's eviction (A\ref{alg:strong_mvp_oram_server}, L18-20).

On the client-side, client $c_i$ executes the access procedure, i.e., three steps of access, $\sigma + 1$ times to access address $\mathit{addr}$, except it only modifies the block during the real access. 
To determine the timestep $\tau_i$ when $c_i$ will execute the real access, $c_i$ uses $\mathit{addr}$s received in $\mathcal{A}$, i.e., addresses concurrent clients are going to access.
During the first access, $c_i$ counts the number of ongoing accesses to $\mathit{addr}$ and subtracts its own access to determine $\tau_\mathit{real}$ when it is going to execute its real access (A\ref{alg:strong_mvp_oram_client}, L5-6).
When the security parameter $\sigma < c-1$, $c_i$ limits $\tau_{\mathit{real}}$ to $\sigma$.

Using $\tau_\mathit{real}$, $c_i$ verifies twice whether the current timestep $\tau$ matches $\tau_\mathit{real}$ during the access.
If the check returns true, $c_i$ executes the real access by selecting a random path that passes through the slot where the block is stored (A\ref{alg:strong_mvp_oram_client}, L9-10) and applies the operation of $\mathit{read}$ or $\mathit{write}$ on the block (A\ref{alg:strong_mvp_oram_client}, L15-20). 
Otherwise, it selects a random path (A\ref{alg:strong_mvp_oram_client}, L12) to execute dummy access.
Note that if the number of concurrent clients $c$ accessing the same address is more than $\sigma + 1$, then $c-\sigma$ clients will have to execute real access in the same last timestep.
Hence, if $c \leq \sigma + 1$, then this strategy guarantees obliviousness even when the worst-case access pattern, i.e., all clients are concurrently accessing the same address, is not skewed.
Next, we prove this result.

\subsection{Obliviousness of Strong MVP-ORAM}
When all clients access the same block in a timestep, and the accessed block is located near the leaves, the access pattern will be distinguishable from a random sequence with high probability.
To address this, recall that we employ a mitigation strategy: each client performs $\sigma$ dummy accesses in addition to its actual access. 
This strategy is effective for the following reason.
When a client invokes $\mathsf{getPM}$, it registers the block address $\mathit{addr}$ in the encrypted data structure $\mathcal{A}$. 
Since $\mathcal{A}$ contains block addresses of all ongoing accesses and is returned as part of the $\mathsf{getPM}$'s response, clients can observe whether $\mathit{addr}$ has already been registered.
If a client is the $k^\mathit{th}$ to access $\mathit{addr}$, where $k\leq \sigma + 1$, it performs $k-1$ dummy accesses before issuing its actual access.
This strategy ensures that no two clients access $\mathit{addr}$ in the same timestep, as long as the number of concurrent accesses for $\mathit{addr}$ does not exceed $\sigma + 1$.
As a result, the access pattern of the first $\sigma + 1$ clients remains indistinguishable from those generated under random access.
However, the access pattern of the remaining $c - \sigma$ clients might be distinguishable from a random sequence. 
The following theorem characterizes the statistical distance between a random sequence and the access pattern generated by MVP-ORAM when dummy accesses are employed.

Let $\mu(N, c, \mathcal{D}(\alpha), \sigma)$ be defined as follows in this case:
\begin{align*}
    \mu(N, c, \mathcal{D}(\alpha), \sigma) = 
    \begin{cases}
        0 & c \leq \sigma
        \\
        \displaystyle\dfrac{1}{2}\sum^{c^*}_{k=1}|\Pr(X=k \mid C = c^*)- \\
        \indent \Pr(Y=k \mid C = c^*)| & \text{o.w.}
    \end{cases}
\end{align*}

where $c^* = c-\sigma$, $\Pr(X=k \mid C = c^*)$ is the distribution computed in Lemma~\ref{lem:1:X}, and $\Pr(Y=k \mid C = c^*)$ is the distribution computed in Lemma~\ref{lem:2:Y}.

\begin{lemma}{12}
    \label{thm:worst-tvd}
    Given $c,\sigma, N \in \mathbb{N}$, $\alpha \in \mathbb{R}$, and $\mathcal{D}\in\mathcal{U}$,
    the statistical distance between a random sequence of size~$c$ and the access pattern generated by Strong MVP-ORAM under the dummy access strategy is bounded by $\mu(N, c, \mathcal{D}(\alpha), \sigma)$.
\end{lemma}
\begin{proof}
    Let $c$ and $\sigma$ be the number of clients accessing the same block address $\mathit{addr}$ in a given timestep~$t$ and the number of dummy accesses each client performs in Strong MVP-ORAM, respectively. 
    Let $X$ be a random variable representing the number of distinct leaves observed in a random sequence of size~$c$. 
    Let $Y$ denote the number of distinct leaves observed in the access pattern generated by Strong MVP-ORAM. 
    We want to show that the statistical distance between~$X$ and~$Y$ is given by:

    {\footnotesize
    \begin{align*} 
        \Delta(X,Y) = 
        \begin{cases}
            0 & c \leq \sigma
            \\
            \displaystyle\dfrac{1}{2}\sum^{c^*}_{k=1}\left|\Pr(X=k \mid C = c^*)-\Pr(Y=k \mid C = c^*)\right| & \text{o.w.}
        \end{cases}
    \end{align*}
    }
    
    where $c^* = c-\sigma$.
    
    We divide the proof in two parts:
    \begin{itemize}
        \item If $c \leq \sigma + 1$.
        Since clients can detect whether $\mathit{addr}$ is being accessed by others during timestep~$t$, a client that is the $k$th to access $\mathit{addr}$, where $k \leq c$, performs $k - 1$ dummy accesses before issuing its actual access.
        Indeed, the mitigation strategy ensures that accesses to $\mathit{addr}$ are performed so that no two clients access $\mathit{addr}$ in the same timestep.
        As each block is moved to the stash at the end of every access (if it is not already there), each client selects a random, independent path when accessing the block.
        As a result, the access pattern generated by Strong MVP-ORAM is indistinguishable from a random sequence of size~$c$.
        \item If $c > \sigma + 1$.
        As established in the first case, perfect security can be ensured for the first $\sigma + 1$ clients accessing the block.
        However, the remaining $c - \sigma$ clients will access the same block at timestep $t + \sigma$.
        According to Theorem~\ref{thm:oram-same-block-not-negligible}, the statistical distance between the number of distinct leaves observed in a random access sequence of size~$c$ and that observed in the access pattern generated by Strong MVP-ORAM can be quantified as follows:

        {\footnotesize
        \begin{align*}
            \displaystyle
            \Delta(X,Y) = \dfrac{1}{2}\sum^{c^*}_{k=1}\left|\Pr(X=k \mid C = c^*)-\Pr(Y=k \mid C = c^*)\right|,
        \end{align*}
        }
        
        where $c^* = c-\sigma$.
    \end{itemize}
\end{proof}

Next, we show that Strong MVP-ORAM still holds the \textit{Obliviousness} property of Definition~\ref{def:opram}.

\begin{lemma}{13}
    \label{thm:oram-obliviousness}
    Given $c,\sigma, N \in \mathbb{N}$, $\alpha \in \mathbb{R}$, and $\mathcal{D} \in \mathcal{U}$, let $\overrightarrow{b}_e=\{b_i\}_{i\in\{1, \dots, c\}}$ denote a set of $c$ concurrent queries in timestep $e$ and $\overrightarrow{y}=(\overrightarrow{b}_1, \overrightarrow{b}_2, \dots)$ denote a set of queries from multiple clients in each timestep. MVP-ORAM holds the \emph{Obliviousness} property of a Wait-Free ORAM, as per Definition~\ref{def:opram}, when clients perform $\sigma$ dummy accesses.
\end{lemma}
\begin{proof}
    Let $A(\overrightarrow{y})$ be an access pattern generated by Strong MVP-ORAM while executing query $\overrightarrow{y}$.
    According to Definition~\ref{def:opram}, we must show that there is a function $\mu$ that bounds the statistical distance between any two access patterns of the same length.
    To establish such a bound, we compare $A(\overrightarrow{y})$ with a random access pattern.
    From Lemma~\ref{thm:sec:cases:1:2}, $A(\overrightarrow{y})$ is indistinguishable from a random access pattern of the same size when a single client accesses the ORAM in each timestep or multiple clients access different blocks.
    From Lemma~\ref{thm:worst-tvd}, in the worst-case scenario, the statistical distance between $A(\overrightarrow{b}_e)$ and a random access pattern of the same size is bounded by $\mu(N, c, \mathcal{D}(\alpha), \sigma)$.
    Since after each concurrent access ORAM changes, i.e., locations of blocks accessed by clients change, the statistical distance between $A(\overrightarrow{y})$ and a random access sequence is still bounded by $\mu(N, c, \mathcal{D}(\alpha), \sigma)$, proving the lemma.
\end{proof}

To conclude our obliviousness analysis, we show under which conditions MVP-ORAM can provide the same security as a traditional parallel ORAM with collision-freedom~\cite{boyle2015}.

\begin{theorem}{6}
\label{thm:poram_equivalence}
    Given $c, \sigma, N \in \mathbb{N}$, if $c \leq \sigma + 1$, then Strong MVP-ORAM's access pattern is indistinguishable from a random access pattern with negligible probability in $N$.
\end{theorem}
\begin{proof}
    When there are at most $\sigma + 1$ clients accessing ORAM, they retrieve blocks in $\sigma + 1$ timesteps. When they are retrieving different blocks, they will access different random paths in the first timestep and continue accessing random paths in the remaining $\sigma$ timesteps. When they retrieve the same block, each client will perform the real access in different timesteps and dummy accesses in the other ones.
    In both cases, clients generate a sequence of leaves by selecting any of the $2^L$ leaves in each access.
    Thus, the probability of all clients accessing the same leaf is $2^{-L\, c \, (\sigma + 1)}$, which is the same probability of clients selecting random leaves and is negligible in $N$.
\end{proof}




%% file: algorithms/alg_strong_mvp_oram_client.tex
\begin{algorithm}[t!]
\SetKwProg{Fn}{Function}{}{}
\DontPrintSemicolon
\caption{Strong MVP-ORAM client $c_i$.}
\label{alg:strong_mvp_oram_client}
{\small

\Fn{$\mathsf{access}(c_i, \mathit{op}, \mathit{addr}, \mathit{data}^*)$}{
    $\mathit{data} \gets \bot$;
    $\tau_\mathit{real} = 0$\\
    \For{$\tau_i = 0, \dots, \sigma$}{
    $\langle \mathcal{H}_\mathit{pathMaps}, \mathit{seq}, \mathcal{A}\rangle \gets \mathsf{Server.getPM}(c_i, \mathit{addr})$\\
    \If{$\tau_i = 0$}{
        $\tau_\mathit{real} \gets \mathsf{min}(\sigma,$ number $\mathit{addr}$ in $\mathcal{A} - 1)$
    }
    $\mathit{pm} \gets \underline{\mathit{consolidatePathMaps}}(\mathcal{H}_\mathit{pathMaps})$\\
    \If{$\tau_i = \tau_\mathit{real}$}{
        $\langle \mathit{sl}, \_ \rangle \gets \mathit{pm}[\mathit{addr}]$\\
        $l \gets $ random path that passes through slot $\mathit{sl}$
    }
    \Else{
        $l \gets $ random path
    }
    $\langle \mathcal{P}_l, \mathcal{S} \rangle \gets \mathsf{Server.getPS}(c_i, l)$\\
	$W \gets \underline{\mathit{mergePathStashes}}(\mathcal{P}_l, \mathcal{S}, \mathit{pm})$\\

    \If{$\tau_i = \tau_\mathit{real}$}{
        \If{$\mathit{op} = \mathit{write}$}{
            $\mathit{data} \gets \mathit{data}^*$; 
            $v \gets \mathit{seq}$\\
        }
        \Else{
            $\langle \_, \mathit{data}, \langle v, \_, \_ \rangle \rangle \gets W[\mathit{addr}]$\\
        }
        $W[addr] \gets \langle addr, \mathit{data}, \langle v, \mathit{seq}, \mathit{seq} \rangle \rangle$\\
    }
    
    $\langle \mathcal{P}_l^*, S, M_\mathit{l} \rangle \gets \underline{\mathit{populatePath}}(W, l, \mathit{addr}, \mathit{pm}, \mathit{seq})$\\
    $\mathsf{Server.evict}(c_i, M_\mathit{l}, \mathcal{P}_l^*, S)$\\
    }
    \Return $\mathit{data}$
    }
}
\end{algorithm}

%% file: algorithms/alg_strong_mvp_oram_server.tex
\begin{algorithm}[t!]
\SetKwProg{Fn}{Function}{}{}
\DontPrintSemicolon
\caption{Strong MVP-ORAM server.}
\label{alg:strong_mvp_oram_server}
{\small
\Proc{setup($\mathcal{T}, \mathcal{S}$)}{
    $\mathit{oramState} \gets \langle \mathcal{T}, \mathcal{S}, \emptyset \rangle$;
    $\mathit{nextSeq} \gets 1$\\
    $\forall c_i \in \Gamma: \mathit{context}[c_i] \gets \bot; \mathcal{A}[c_i] \gets \bot$\\
}

\Fn{$\mathsf{getPM}(c_i, \mathit{addr})$}{
    $\mathit{seq} \gets \mathit{nextSeq}$;
    $\mathit{nextSeq} \gets \mathit{nextSeq}+1$\\
    $\mathit{context}[c_i] \gets \mathit{oramState}$\\
    \If{$\mathcal{A}[c_i] = \bot$} {
        $\mathcal{A}[c_i] \gets \langle \mathit{addr}, 1 \rangle$
    }
    \Else{
        $\langle \_, \mathit{ct} \rangle \gets \mathcal{A}[c_i]$;
        $\mathcal{A}[c_i] \gets \langle \mathit{addr}, \mathit{ct}+1 \rangle$
    }
    $\langle \_, \_, \mathcal{H}_\mathit{pathMaps} \rangle \gets \mathit{oramState}$\\
    \Return $\langle \mathcal{H}_\mathit{pathMaps}, \mathit{seq}, \mathcal{A} \rangle$
}

\Fn{$\mathsf{getPS}(c_i, l)$}{
    $\langle \mathcal{T}, \mathcal{S}, \_ \rangle \gets \mathit{context}[c_i]$\\
    \Return $\langle \mathcal{T}(l), \mathcal{S} \rangle$
}

\Proc{$\mathsf{evict}(c_i, M_\mathit{l}, \mathcal{P}_l^*, S)$}{
    $\langle \mathcal{T}, \mathcal{S}, \_ \rangle \gets \mathit{context}[c_i]$;
    $\mathit{context}[c_i] \gets \bot$\\
    $\langle \_, \mathit{ct} \rangle \gets \mathcal{A}[c_i]$\\
    \If{$\mathit{ct} = \sigma + 1$}{
        $\mathcal{A}[c_i] \gets \bot$
    }
    $\langle \mathcal{T}^c, \mathcal{S}^c, \mathcal{H}^c_\mathit{pathMaps} \rangle \leftarrow \mathit{oramState}$\\
    \For(\tcp*[f]{update the tree}){$\mathit{sl} \in \mathcal{T}(l)$}{
        $\mathcal{T}^*(l,\mathit{sl}) \gets (\mathcal{T}^c(l,\mathit{sl}) \setminus \mathcal{T}(l,\mathit{sl})) \cup \mathcal{P}^*_l(\mathit{sl})$\\
    }    
    $\mathcal{S}^* \gets (\mathcal{S}^c \setminus \mathcal{S}) \cup \{S\}$\\
    $\mathcal{H}^*_\mathit{pathMaps} \gets \mathcal{H}^c_\mathit{pathMaps} \cup \{  M_\mathit{l} \}$\\
    $\mathit{oramState} \gets \langle \mathcal{T}^*, \mathcal{S}^*, \mathcal{H}^*_\mathit{pathMaps} \rangle$
}
}
\end{algorithm}

%% file: content/artifact_extended.tex




\clearpage

\section{Artifact Appendix}
\label{ap:artifact}


 

The artifact includes the implementation of MVP-ORAM, QuORAM,\footnote{\url{https://github.com/SeifIbrahim/QuORAM/}} and COBRA.\footnote{\url{https://github.com/bft-smart/cobra}} The former was used for performance assessment, while the latter two were used for comparison.
The experiments were conducted on multiple AWS servers, but for ease of reproduction, we provide a scaled-down version of a single-machine evaluation setup in this appendix.

\subsection{Description \& Requirements}




\subsubsection{How to access}
The artifact is available on Zenodo,\footnote{\url{https://doi.org/10.5281/zenodo.17842154}} which includes the MVP-ORAM, QuORAM, and COBRA implementations adapted for easier benchmarking.

\subsubsection{Hardware dependencies}
This artifact was tested on a machine with a $2.59$ GHz CPU and $16$ GB of RAM. Note that the experimental results presented in the paper were obtained using multiple more powerful machines instead of a single machine.

\subsubsection{Software dependencies}
Execution of this artifact requires \emph{unzip}, \emph{gnuplot}, and \emph{OpenJDK 11}. It can be executed on either Linux or the Linux subsystem on Windows.\footnote{\url{https://learn.microsoft.com/windows/wsl}}

\subsubsection{Benchmarks}
None.

\subsection{Artifact Installation \& Configuration}
Download the artifact from Zenodo in the location where the experiments will be executed and open a terminal in that location. Let us designate the terminal as \emph{builder} terminal during the experiments.
Execute the following command in the \emph{builder} terminal to extract the artifact:

\begin{lstlisting}
unzip MVP-ORAM-Artifact.zip
\end{lstlisting}

After extracting, you should have a folder named \verb|MVP-ORAM-Artifact| containing the artifact. Navigate to this folder in \emph{builder} terminal and, \textbf{for all the remaining instructions, assume a relative path from it}.

The experiments are automated using a custom benchmarking tool, configured by setting parameters in the \verb|config/benchmark.config| file.

\subsection{Experiment Workflow}

During the experiments, the collected and processed data will be stored in subfolders located at the path specified in \verb|output.path|:
\begin{itemize}
    \item \verb|output/raw_data| will contain the raw data.
    \item \verb|output/processed_data| will contain the processed data used for plotting.
    \item \verb|output/plots| will contain the produced plots.
\end{itemize}
Hence, for all experiments, set \verb|output.path| to the same location, e.g., the path to the \verb|MVP-ORAM-Artifact| folder.

The artifact is divided into the MVP-ORAM and QuORAM projects. The experiment workflow of both projects is composed of the following steps: 
\begin{enumerate}
    \item Navigate to the project folder in \emph{builder} terminal.
    \item Build the project by executing the following command in \emph{builder} terminal:
    \begin{lstlisting}
./gradlew localDeploy -PnWorkers=x
    \end{lstlisting}
    where $x$ is the number of servers plus the number of client workers. We recommend $x=3 \times \mathtt{max}(\mathit{fault\_thresholds}) + 3$. After a successful execution of the command, you should have $x+1$ folders in \verb|<project>/build/local/| named \verb|controller| and \verb|worker<i>|, where $\verb|i| \in [0, \dots, x-1]$. These folders contain all the necessary materials to execute the experiments.
    \item Open a new terminal and navigate to \verb|<project>/build/local/controller| folder. Let us designate this terminal as \emph{controller} terminal.
    \item Configure \verb|benchmarking.config| located in the \verb|<project>/build/local/controller/config| folder using a text editor (e.g., \emph{nano}) for a given experiment.
    \item Execute the following command in the \emph{controller} terminal:
    \begin{lstlisting}
./smartrun.sh controller.BenchmarkControllerStartup config/benchmark.config
    \end{lstlisting}
    \item Execute the following command in the \emph{builder} terminal:
    \begin{lstlisting}
./runscripts/startLocalWorkers.sh <x> 127.0.0.1 12000
    \end{lstlisting}
    The execution of this command will trigger the experiment in the \emph{controller} terminal, and it will display the experiment status during the execution. The experiment concludes when the \emph{controller} terminal prints the execution duration.
    %
    %
    %
\end{enumerate}

\textbf{Tip:} An active experiment can be terminated at any time by executing \verb|CTRL+C| in the \emph{controller} terminal.


\subsection{Major Claims}



We make the following claims in our paper:
\begin{itemize}
    \item (C1): The stash size stabilizes after some number of accesses, and the maximum stash size increases as we decrease $\alpha$ and increase $c$. This is proven by experiment (E1), whose results are shown in Fig.~\ref{fig:stash_experiments}.
    \item (C2): The overall throughput decreases and the latency increases as we increase the number of servers. This is demonstrated by experiment (E2), and reported in Fig.~\ref{fig:fault_tolerance}.
    \item (C3): The throughput stabilizes and the latency increases as we increase the number of clients. This is demonstrated by experiment (E2) and reported in Fig.~\ref{fig:fault_tolerance}.
    \item (C4): MVP-ORAM outperforms QuORAM both in throughput and latency with $50$ clients. This is proven by experiment (E3), with results reported in Table~\ref{tab:mvp_oram_vs_quorum_performance}.
    \item (C5): MVP-ORAM's performance is an order of magnitude lower than COBRA with $50$ clients. This is proven by experiment (E4), with results reported in Table~\ref{tab:mvp_oram_vs_quorum_performance}.
    \item (C6): The statistical distance between access patterns decreases as the parameter $\alpha$ and $L$ increase, i.e., accesses are more skewed and tree size increases. 
    This is demonstrated by experiment (E5), whose results are shown in Fig.~\ref{fig:statistical_distance_analyses}.
    \item (C7): The statistical distance between access patterns decreases as the parameter $\sigma$ increases, confirming that MVP-ORAM's security improves with higher $\sigma$ values. 
    This is demonstrated by experiment (E6), whose results are shown in Fig.~\ref{fig:tvd-sigma}.
\end{itemize}

\subsection{Evaluation}





\subsubsection{Experiment (E1)}
[Stash]
[10 human-minutes + 1.1 compute-hour]: This experiment shows the impact of $\alpha$ and $c$ on the stash size. 
The results show that the stash size stabilizes over time, and its maximum size increases as $\alpha$ decreases (i.e., as the access distribution becomes more uniform) and $c$ increases.

\textit{[Preparation]}
This experiment requires executing the steps described in the experimental workflow twice, with different values of \verb|zipf_parameters| and \verb|clients_per_round|.
In both executions, use the following parameters:
\begin{itemize}
    \item \verb|global.worker.machines=3|
    \item \verb|controller.benchmark.strategy=oram.|\ \verb|benchmark.MeasurementBenchmarkStrategy|
    \item \verb|fault_thresholds=0|
    \item \verb|tree_heights=16|
    \item \verb|bucket_sizes=4|
    \item \verb|block_sizes=8|
    \item \verb|concurrent_clients=15|
    \item \verb|measurement_duration=600|
\end{itemize}

\textbf{First execution:} Set \verb|clients_per_round=5| and  \verb|zipf_parameters=0.000001 1.0 2.0|

\textbf{Second execution:} Set \verb|clients_per_round=1 10 15| and  \verb|zipf_parameters=1.0|

\textit{[Execution]}
Consider the project \verb|MVP-ORAM| and follow the steps described in the experimental workflow to execute the experiment twice, setting the parameters defined above.

\textit{[Results]}
Execute the following command, in the \emph{builder} terminal, to produce Fig.~\ref{fig:stash_experiments}:
\begin{lstlisting}
gnuplot -e "O='<output.path>'; L='16'; Z='4'; B='8'; c_max='15'; D=10" plotScripts/stash_plot.gp
\end{lstlisting}
The correct execution of this command will create \verb|<output.path>/output/plots/stash.pdf| file containing the figure. Due to the scaled-down experiment, this figure only shows the overall trend of Fig.~\ref{fig:stash_experiments}, confirming (C1).

\textbf{Note:} For more acurate results, increase \verb|measurement_duration| and setting \verb|D='d'| during plotting, where $d=\lceil \mathit{measurement\_duration}/60 \rceil$.

\subsubsection{Experiment (E2)}
[Performance]
[5 human-minutes + 0.5 compute-hour]: This experiment shows the impact of $n$ servers and $c$ clients on the throughput and latency of the system. It shows that increasing $n$ decreases the throughput and increases the latency, and increasing $c$ increases the latency while keeping the throughput stable.

\textit{[Preparation]}
Use the following parameters for this experiment:
\begin{itemize}
    \item \verb|global.worker.machines=12|
    \item \verb|controller.benchmark.strategy=oram.|\ \verb|benchmark.MeasurementBenchmarkStrategy|
    \item \verb|fault_thresholds=0 1 2|
    \item \verb|clients_per_round=1 5 10 15|
    \item \verb|tree_heights=16|
    \item \verb|bucket_sizes=4|
    \item \verb|block_sizes=8|
    \item \verb|concurrent_clients=5|
    \item \verb|measurement_duration=60|
    \item \verb|zipf_parameters=1.0|
\end{itemize}

\textit{[Execution]}
Consider the project \verb|MVP-ORAM| and follow the steps described in the experimental workflow to execute the experiment using the above values.

\textit{[Results]}
Execute the following command in the \emph{builder} terminal to produce Fig.~\ref{fig:fault_tolerance}:
\begin{lstlisting}
gnuplot -e "O='<output.path>'; L='16'; Z='4'; B='8'; A='1.0'; c_max='5'" plotScripts/throughput_latency_plot.gp
\end{lstlisting}
The correct execution of this command will create the \verb|<output.path>/output/plots/performance.pdf| file containing the figure. Due to the scaled-down experiment, this figure only shows the overall trend of Fig.~\ref{fig:fault_tolerance}, confirming (C2) and (C3).

\subsubsection{Experiment (E3)}
[Comparison with QuORAM]
[5 human-minutes + 1 compute-hour]: This experiment shows that MVP-ORAM outperforms QuORAM both in throughput and latency. 
However, due to the resources required to run the experiments, we only show results for up to $15$ clients.

\textit{[Preparation]}
Use the following parameters for this experiment: 
\begin{itemize}
    \item \verb|global.worker.machines=7|
    \item \verb|fault_thresholds=1 2|
    \item \verb|clients_per_round=1 5 10 15|
    \item \verb|storage.sizes=1|
    \item \verb|bucket_sizes=4|
    \item \verb|block_sizes=8|
    \item \verb|measurement_duration=60|
    \item \verb|zipf_parameters=1.0|
\end{itemize}

\textbf{Note:} QuORAM configuration file is located at \verb|QuORAM/config/benchmark.config|.

\textit{[Execution]}
Consider the project \verb|QuORAM| and follow the steps described in the experimental workflow to execute the experiment using the above values.

\textbf{Note:} Execute (E3) after completing (E2), as (E3)'s plot relies on (E2)'s results for comparison.

\textit{[Results]}
Run the following command in \emph{builder} terminal:
\begin{lstlisting}
gnuplot -e "O='<output.path>'; L='16'; Z='4'; B='8'; A='1.0'; c_max='5'" plotScripts/mvp_oram_vs_quoram_plot.gp
\end{lstlisting}
The correct execution of this command will create the \verb|<output.path>/output/plots/quoram.pdf| file containing a figure comparing MVP-ORAM with QuORAM. Due to the scaled-down experiment, this figure only shows the overall trend of Table~\ref{tab:mvp_oram_vs_quorum_performance}, confirming (C4).

\subsubsection{Experiment (E4)}
[Comparison with COBRA] [5 human-minutes + 0.3 compute-hour]: This experiment shows that MVP-ORAM's performance is an order of magnitude lower than COBRA with up to $50$ clients. However, due to the resources required to run the experiments, we present results only for up to $15$ clients.

\textit{[Preparation]}
Use the following parameters for this experiment:
\begin{itemize}
    \item \verb|global.worker.machines=9|
    \item \verb|controller.benchmark.strategy=oram.|\ \verb|benchmark.cobra.ThroughputLatencyBench|\ \verb|markStrategy|
    \item \verb|fault_thresholds=1 2|
    \item \verb|clients_per_round=1 5 10 15|
    \item \verb|cobra.request.private_data_size=8|
    \item \verb|measurement_duration=60|
\end{itemize}

\textit{[Execution]}
Consider the project \verb|MVP-ORAM| and follow the steps described in the experimental workflow to execute the experiment using the above values.

\textbf{Note:} Execute (E4) after completing (E2), as (E4)'s plot relies on (E2)'s results for comparison.

\textit{[Results]}
Execute the following command in the \emph{builder} terminal:
\begin{lstlisting}
gnuplot -e "O='<output.path>'; L='16'; Z='4'; B='8'; A='1.0'; c_max='5'" plotScripts/mvp_oram_vs_cobra_plot.gp
\end{lstlisting}
The correct execution of this command will create \verb|<output.path>/output/plots/cobra.pdf| file containing a figure comparing MVP-ORAM with COBRA. Due to the scaled-down experiment, this figure only shows the overall trend of Table~\ref{tab:mvp_oram_vs_quorum_performance}, confirming (C5).

\subsubsection{Experiment (E5)}
[Impact of $\alpha$ and $L$ on security] [10 human-minutes + 0.01 compute-hour]: This numerical experiment demonstrates the security improvements of MVP-ORAM when the access distribution is skewed and the tree size increases.

\textit{[Preparation]}
This experiment requires executing the steps described in the experimental workflow twice, with different values of \verb|tree_heights| and  \verb|access_thresholds|.
In both executions, use the following parameters:
\begin{itemize}
    \item \verb|global.worker.machines=0|
    \item \verb|controller.benchmark.strategy=oram.|\ \verb|benchmark.tvd.TVDHeightBenchmarkStrategy|
    \item \verb|clients_per_round=1 5 10 15 20 25 30 35|\ \verb|40 45 50|
\end{itemize}

\textbf{First execution:} Set \verb|tree_heights=17| and \verb|access_thresholds=(0.90436, 0.8) (1.0945,|
\verb| 0.9) (1.2353, 0.95) (1.537, 0.99)|.

\textbf{Second execution:} Set \verb|tree_heights=25| and \verb|access_thresholds=(0.87683, 0.8) (1.0251,|
\verb|0.9) (1.1237, 0.95) (1.3101, 0.99)|.

\textit{[Execution]}
Consider the project \verb|MVP-ORAM| and follow the steps, except the last step, described in the experimental workflow to execute the experiment using the above values.

\textit{[Results]}
Execute the following command, in the \emph{builder} terminal, to produce Fig.~\ref{fig:statistical_distance_analyses}:
\begin{lstlisting}
gnuplot -e "O='<output.path>'; L1='17'; L2='25'" plotScripts/tvd_heights_plot.gp
\end{lstlisting}
The correct execution of this command will create \verb|<output.path>/output/plots/tvd_heights.pdf| file containing the figure, confirming (C6).

\subsubsection{Experiment (E6)}
[Impact of $\sigma$ on security] [5 human-minutes + 0.01 compute-hour]: This numerical experiment demonstrates the security improvements of Strong MVP-ORAM as we increase the number of dummy accesses $\sigma$.

\textit{[Preparation]}
Use the following parameters for this experiment:
\begin{itemize}
    \item \verb|global.worker.machines=0|
    \item \verb|controller.benchmark.strategy=oram.|\ \verb|benchmark.tvd.TVDSigmaBenchmarkStrategy|
    \item \verb|clients_per_round=1 5 10 15 20 25 30 35|\ \verb|40 45 50|
    \item \verb|tree_heights=17|
    \item \verb|access_thresholds=(0.0000001, 0.99)|\ \verb|(1.537, 0.99)|
    \item \verb|sigmas=0 10 20 30 40 50|
\end{itemize}

\textit{[Execution]}
Consider the project \verb|MVP-ORAM| and follow the steps, except the last step, described in the experimental workflow to execute the experiment using the above values.

\textit{[Results]}
Execute the following command, in the \emph{builder} terminal, to produce Fig.~\ref{fig:tvd-sigma}:
\begin{lstlisting}
gnuplot -e "O='<output.path>'; L='17'; T='0.99'" plotScripts/tvd_sigmas_plot.gp
\end{lstlisting}
The correct execution of this command will create \verb|<output.path>/output/plots/tvd_sigmas.pdf| file containing the figure, confirming (C7).

\subsection{Customization}
The experiments can be executed by setting different values for the used parameters, i.e., varying the fault threshold, tree height, bucket size, block size, Zipfian parameter, number of clients, and the maximum number of concurrent clients. However, to plot the measurements correctly, provide the correct values of the used parameters.

Additionally, the accuracy of the results can be improved by obtaining more data points by increasing the experiment duration, which can be achieved by modifying the \verb|measurement_duration| parameter.









%% file: content/artifact.tex




\clearpage

\section{Artifact Appendix}
\label{ap:artifact}

The artifact includes the implementation of MVP-ORAM, QuORAM,\footnote{\url{https://github.com/SeifIbrahim/QuORAM/}} and COBRA.\footnote{\url{https://github.com/bft-smart/cobra}} The former was used for performance assessment, while the latter two were used for comparison.
The experiments were conducted on multiple AWS servers, but for ease of reproduction, we provide a scaled-down version of a single-machine evaluation setup in this appendix.
\ifbool{extendedVersion}{
}{
We present instructions to reproduce the main results of our paper, which include the performance of MVP-ORAM and a comparison with QuORAM. Reproduction of other experiments is presented in the extended version of this paper~\cite{mvp_oram_extended_version}.
}

\subsection{Description \& Requirements}

\subsubsection{How to access}

The artifact is available on Zenodo,\footnote{\url{https://doi.org/10.5281/zenodo.17842154}} which includes the MVP-ORAM, QuORAM, and COBRA implementations adapted for easier benchmarking.

\subsubsection{Hardware dependencies}
This artifact was tested on a machine with a $2.59$ GHz CPU and $16$ GB of RAM. Note that the experimental results presented in the paper were obtained using multiple more powerful machines instead of a single machine.

\subsubsection{Software dependencies}
Execution of this artifact requires \emph{unzip}, \emph{gnuplot}, and \emph{OpenJDK 11}. It can be executed on either Linux or the Linux subsystem on Windows.\footnote{\url{https://learn.microsoft.com/windows/wsl}}

\subsubsection{Benchmarks}
None.

\subsection{Artifact Installation \& Configuration}
Download the artifact from Zenodo in the location where the experiments will be executed and open a terminal in that location. Let us designate the terminal as \emph{builder} terminal during the experiments.
Execute the following command in the \emph{builder} terminal to extract the artifact:

\begin{lstlisting}
unzip MVP-ORAM-Artifact.zip
\end{lstlisting}

After extracting, you should have a folder named \verb|MVP-ORAM-Artifact| containing the artifact. Navigate to this folder in \emph{builder} terminal and, \textbf{for all the remaining instructions, assume a relative path from it}.

The experiments are automated using a custom benchmarking tool, configured by setting parameters in the \verb|config/benchmark.config| file.

\subsection{Experiment Workflow}
During the experiments, the collected and processed data will be stored in subfolders located at the path specified in \verb|output.path|:
\begin{itemize}
    \item \verb|output/raw_data| will contain the raw data.
    \item \verb|output/processed_data| will contain the processed data used for plotting.
    \item \verb|output/plots| will contain the produced plots.
\end{itemize}
Hence, for all experiments, set \verb|output.path| to the same location, e.g., the path to the \verb|MVP-ORAM-Artifact| folder.

The artifact is divided into the MVP-ORAM and QuORAM projects. The experiment workflow of both projects is composed of the following steps: 
\begin{enumerate}
    \item Navigate to the project folder in \emph{builder} terminal.
    \item Build the project by executing the following command in \emph{builder} terminal:
    \begin{lstlisting}
./gradlew localDeploy -PnWorkers=x
    \end{lstlisting}
    where $x$ is the number of servers plus the number of client workers. We recommend $x=3 \times \mathtt{max}(\mathit{fault\_thresholds}) + 3$. After a successful execution of the command, you should have $x+1$ folders in \verb|<project>/build/local/| named \verb|controller| and \verb|worker<i>|, where $\verb|i| \in [0, \dots, x-1]$. These folders contain all the necessary materials to execute the experiments.
    \item Open a new terminal and navigate to \verb|<project>/build/local/controller| folder. Let us designate this terminal as \emph{controller} terminal.
    \item Configure \verb|benchmarking.config| located in the \verb|<project>/build/local/controller/config| folder using a text editor (e.g., \emph{nano}) for a given experiment.
    \item Execute the following command in the \emph{controller} terminal:
    \begin{lstlisting}
./smartrun.sh controller.BenchmarkControllerStartup config/benchmark.config
    \end{lstlisting}
    \item Execute the following command in the \emph{builder} terminal:
    \begin{lstlisting}
./runscripts/startLocalWorkers.sh <x> 127.0.0.1 12000
    \end{lstlisting}
    The execution of this command will trigger the experiment in the \emph{controller} terminal, and it will display the experiment status during the execution. The experiment concludes when the \emph{controller} terminal prints the execution duration.
\end{enumerate}

\textbf{Tip:} An active experiment can be terminated at any time by executing \verb|CTRL+C| in the \emph{controller} terminal.

\subsection{Major Claims}

We make the following claims in our paper:
\begin{itemize}
    \item (C1): The stash size stabilizes after some number of accesses, and the maximum stash size increases as we decrease $\alpha$ and increase $c$. This is proven by experiment (E1), whose results are shown in Fig.~\ref{fig:stash_experiments}.
    \item (C2): The overall throughput decreases and the latency increases as we increase the number of servers. This is demonstrated by experiment (E2), and reported in Fig.~\ref{fig:fault_tolerance}.
    \item (C3): The throughput stabilizes and the latency increases as we increase the number of clients. This is demonstrated by experiment (E2) and reported in Fig.~\ref{fig:fault_tolerance}.
    \item (C4): MVP-ORAM outperforms QuORAM both in throughput and latency with $50$ clients. This is proven by experiment (E3), with results reported in Table~\ref{tab:mvp_oram_vs_quorum_performance}.
    \ifbool{extendedVersion}{
    \item (C5): MVP-ORAM's performance is an order of magnitude lower than COBRA with $50$ clients. This is proven by experiment (E4), with results reported in Table~\ref{tab:mvp_oram_vs_quorum_performance}.
    \item (C6): The statistical distance between access patterns decreases as the parameter $\alpha$ and $L$ increase, i.e., accesses are more skewed and tree size increases. 
    This is demonstrated by experiment (E5), whose results are shown in Fig.~\ref{fig:statistical_distance_analyses}.
    \item (C7): The statistical distance between access patterns decreases as the parameter $\sigma$ increases, confirming that MVP-ORAM's security improves with higher $\sigma$ values. 
    This is demonstrated by experiment (E6), whose results are shown in Fig.~\ref{fig:tvd-sigma}.
    }{}
\end{itemize}

\subsection{Evaluation}

\subsubsection{Experiment (E1)}
[Stash]
[10 human-minutes + 1.1 compute-hour]: This experiment shows the impact of $\alpha$ and $c$ on the stash size. 
The results show that the stash size stabilizes over time, and its maximum size increases as $\alpha$ decreases (i.e., as the access distribution becomes more uniform) and $c$ increases.

\textit{[Preparation]}
This experiment requires executing the steps described in the experimental workflow twice, with different values of \verb|zipf_parameters| and \verb|clients_per_round|.
In both executions, use the following parameters:
\begin{itemize}
    \item \verb|global.worker.machines=3|
    \ifbool{extendedVersion}{
    \item \texttt{controller.benchmark.strategy=oram.\linebreak benchmark.MeasurementBenchmarkStrategy}
    }{}
    \item \verb|fault_thresholds=0|
    \item \verb|tree_heights=16|
    \item \verb|bucket_sizes=4|
    \item \verb|block_sizes=8|
    \item \verb|concurrent_clients=15|
    \item \verb|measurement_duration=600|
\end{itemize}

\textbf{First execution:} Set \verb|clients_per_round=5| and  \verb|zipf_parameters=0.000001 1.0 2.0|

\textbf{Second execution:} Set \verb|clients_per_round=1 10 15| and  \verb|zipf_parameters=1.0|

\textit{[Execution]}
Consider the project \texttt{MVP-ORAM} and follow the steps described in the experimental workflow to execute the experiment twice, setting the parameters defined above.

\textit{[Results]}
Execute the following command, in the \emph{builder} terminal, to produce Fig.~\ref{fig:stash_experiments}:
\begin{lstlisting}
gnuplot -e "O='<output.path>'; L='16'; Z='4'; B='8'; c_max='15'; D=10" plotScripts/stash_plot.gp
\end{lstlisting}
The correct execution of this command will create \verb|<output.path>/output/plots/stash.pdf| file containing the figure. Due to the scaled-down experiment, this figure only shows the overall trend of Fig.~\ref{fig:stash_experiments}, confirming (C1).

\textbf{Note:} For more acurate results, increase \verb|measurement_duration| and setting \verb|D='d'| during plotting, where $d=\lceil \mathit{measurement\_duration}/60 \rceil$.

\subsubsection{Experiment (E2)}
[Performance]
[5 human-minutes + 0.5 compute-hour]: This experiment shows the impact of $n$ servers and $c$ clients on the throughput and latency of the system. It shows that increasing $n$ decreases the throughput and increases the latency, and increasing $c$ increases the latency while keeping the throughput stable.

\textit{[Preparation]}
Use the following parameters for this experiment:
\begin{itemize}
    \item \verb|global.worker.machines=12|
    \ifbool{extendedVersion}{
    \item \texttt{controller.benchmark.strategy=oram.\linebreak benchmark.MeasurementBenchmarkStrategy}
    }{}
    \item \verb|fault_thresholds=0 1 2|
    \item \verb|clients_per_round=1 5 10 15|
    \item \verb|tree_heights=16|
    \item \verb|bucket_sizes=4|
    \item \verb|block_sizes=8|
    \item \verb|concurrent_clients=5|
    \item \verb|measurement_duration=60|
    \item \verb|zipf_parameters=1.0|
\end{itemize}

\textit{[Execution]}
Consider the project \texttt{MVP-ORAM} and follow the steps described in the experimental workflow to execute the experiment using the above values.

\textit{[Results]}
Execute the following command in the \emph{builder} terminal to produce Fig.~\ref{fig:fault_tolerance}:
\begin{lstlisting}
gnuplot -e "O='<output.path>'; L='16'; Z='4'; B='8'; A='1.0'; c_max='5'" plotScripts/throughput_latency_plot.gp
\end{lstlisting}
The correct execution of this command will create the \verb|<output.path>/output/plots/performance.pdf| file containing the figure. Due to the scaled-down experiment, this figure only shows the overall trend of Fig.~\ref{fig:fault_tolerance}, confirming (C2) and (C3).

\subsubsection{Experiment (E3)}
[Comparison with QuORAM]
[5 human-minutes + 1 compute-hour]: This experiment shows that MVP-ORAM outperforms QuORAM both in throughput and latency. 
However, due to the resources required to run experiments, we only show results for up to $15$ clients.

\textit{[Preparation]}
Use the following parameters for this experiment: 
\begin{itemize}
    \item \verb|global.worker.machines=7|
    \item \verb|fault_thresholds=1 2|
    \item \verb|clients_per_round=1 5 10 15|
    \item \verb|storage.sizes=1|
    \item \verb|bucket_sizes=4|
    \item \verb|block_sizes=8|
    \item \verb|measurement_duration=60|
    \item \verb|zipf_parameters=1.0|
\end{itemize}

\textbf{Note:} QuORAM configuration file is located at \verb|QuORAM/config/benchmark.config|.

\textit{[Execution]}
Consider the project \texttt{QuORAM} and follow the steps described in the experimental workflow to execute the experiment using the above values.

\textbf{Note:} Execute (E3) after completing (E2), as (E3)'s plot relies on (E2)'s results for comparison.

\textit{[Results]}
Run the following command in \emph{builder} terminal:
\begin{lstlisting}
gnuplot -e "O='<output.path>'; L='16'; Z='4'; B='8'; A='1.0'; c_max='5'" plotScripts/mvp_oram_vs_quoram_plot.gp
\end{lstlisting}
The correct execution of this command will create the \verb|<output.path>/output/plots/quoram.pdf| file containing a figure comparing MVP-ORAM with QuORAM. Due to the scaled-down experiment, this figure only shows the overall trend of Table~\ref{tab:mvp_oram_vs_quorum_performance}, confirming (C4).

\subsection{Customization}
The experiments can be executed by setting different values for the used parameters, i.e., varying the fault threshold, tree height, bucket size, block size, Zipfian parameter, number of clients, and the maximum number of concurrent clients. However, to plot the measurements correctly, provide the correct values of the used parameters.

Additionally, the accuracy of the results can be improved by obtaining more data points by increasing the experiment duration, which can be achieved by modifying the \verb|measurement_duration| parameter.
